\documentclass[twoside,11pt]{article}

\usepackage{blindtext, mathrsfs}
\usepackage{enumerate}
\newcommand{\xhdr}[1]{\vspace{0.1mm}\noindent{{\bf #1.}}}
\usepackage{float}
\usepackage{caption}
\usepackage{subcaption}
\usepackage[abbrvbib, preprint]{jmlr2e_updated}
\newtheorem{ass}{Assumption}

\usepackage{amsmath}
\usepackage{bbm}
\usepackage{chngcntr}
\usepackage{apptools}
\AtAppendix{\counterwithin{theorem}{section}}

\usepackage{lastpage}
\jmlrheading{23}{2023}{1-\pageref{LastPage}}{3/22; Revised 5/22}{9/22}{21-0000}{Tran, Liu and Donnat}

\ShortHeadings{Sparse topic modeling via spectral decomposition and thresholding}{Tran, Liu and Donnat}
\firstpageno{1}

\DeclareMathOperator*{\argmin}{arg\,min}
\newcommand{\One}{\boldsymbol{1}}
\newcommand{\bP}{\mathbb{P}}
\newcommand{\bE}{\mathbb{E}}
\newcommand{\cE}{\mathcal{E}}

\newcommand{\op}{\text{op}}
\newcommand{\Js}{{J*}}
\newcommand{\diag}{\text{diag}}

\begin{document}

\title{Sparse topic modeling via spectral decomposition and thresholding}

\author{\name Huy Tran \email huydtran@uchicago.edu \\
       \addr Department of Statistics\\
       The University of Chicago\\
       Chicago, IL 60637, USA
       \AND
       \name Yating Liu \email yatingliu@uchicago.edu \\
       \addr Department of Statistics\\
       The University of Chicago\\
       Chicago, IL 60637, USA
       \AND 
        \name Claire Donnat \email cdonnat@uchicago.edu \\
       \addr Department of Statistics\\
       The University of Chicago\\
       Chicago, IL 60637, USA}
\editor{NA}

\maketitle

\begin{abstract}
The probabilistic Latent Semantic Indexing model assumes that the expectation of the corpus matrix is low-rank and can be written as the product of a topic-word matrix and a word-document matrix. In this paper, we study the estimation of the topic-word matrix under the additional assumption that the ordered entries of its columns rapidly decay to zero. This sparsity assumption is motivated by the empirical observation that the word frequencies in a text often adhere to Zipf's law. We introduce a new spectral procedure for estimating the topic-word matrix that thresholds words based on their corpus frequencies, and show that its $\ell_1$-error rate under our sparsity assumption depends on the vocabulary size $p$ only via a logarithmic term. Our error bound is valid for all parameter regimes and in particular for the setting where $p$ is extremely large; this high-dimensional setting is commonly encountered but has not been adequately addressed in prior literature. Furthermore, our procedure also accommodates datasets that violate the \textit{separability} assumption, which is necessary for most prior approaches in topic modeling. Experiments with synthetic data confirm that our procedure is computationally fast and allows for consistent estimation of the topic-word matrix in a wide variety of parameter regimes. Our procedure also performs well relative to well-established methods when applied to a large corpus of research paper abstracts, as well as the analysis of single-cell and microbiome data where the same statistical model is relevant but the parameter regimes are vastly different. 
\end{abstract}


\section{Introduction}\label{sec:intro}
Topic modeling has proven to be a useful tool for dimensionality reduction and exploratory analysis in natural language processing. Beyond text analysis, it has also been successfully applied in areas such as population genetics \citep{pritchard2000inference, bicego2012investigating}, social networks \citep{curiskis2020evaluation} and image analysis \citep{li2010building}. 

\subsection{The statistical model}\label{section: pLSI}
    In this paper, we focus on the probabilistic Latent Semantic Indexing (pLSI) model introduced in \cite{hofmann1999probabilistic}. This simple bag-of-words model involves three variables, namely topics (which are unobserved), words and documents. 
    
   Suppose we observe $n$ documents written using a vocabulary of $p$ words.  For each $1 \leq i \leq n$, let $N_i$ denote the length of document $i$. The corpus matrix $D \in \mathbb{R}^{p \times n}$, which is a sufficient statistic under the pLSI model and which records the empirical frequency of each word in each document, is defined by 
    \begin{equation*}
        D_{ji} = \frac{\text{count of word $j$ in document $i$}}{N_i} \quad \text{for all } 1 \leq i \leq n, 1 \leq j \leq p 
    \end{equation*}

    Let $\{D_{*i}: 1 \leq i \leq n\}$ denote the columns of $D$, each of which contains only non-negative entries that sum up to 1. The pLSI model specifies that the raw word counts for each document $\{N_iD_{*i}: 1 \leq i \leq n\}$ are independently generated, with
    \begin{equation}\label{eq:pLSI-Mult}
        N_i D_{*i}\sim \text{Multinomial}(N_i, [D_0]_{*i})
    \end{equation}
    for some matrix $D_0 \in \mathbb{R}^{p \times n}$ whose columns are $\{[D_0]_{*i}: 1 \leq i \leq n\}$. Here, the columns of $D_0$ specify how words are assigned to documents, and these columns are required to be probability vectors with non-negative entries summing up to 1. Note that \eqref{eq:pLSI-Mult} implies $\bE(D) = D_0$. If we let $Z := D - D_0$, we can write the observation model in a ``signal plus noise'' form:
    \begin{equation}
        D = D_0 + Z
    \end{equation}
    
    The pLSI model further assumes that, for some unobserved $K \in \mathbb{N}$ (which denotes the number of topics), we can factorize $D_0$ as 
    \begin{equation}\label{eq: pLSI-Bayes}
        \mathbb{E}(D) = D_0 = AW 
    \end{equation}
    for some matrices $A \in \mathbb{R}^{p \times K}$ and $W \in \mathbb{R}^{K \times n}$. Like $D_0$, the columns of $A$ and $W$ are required to be probability vectors, so that they can only contain non-negative entries that sum up to 1. $A$ assigns words to topics, while $W$ assigns topics to documents. In this paper, we focus more specifically on estimating the topic-word matrix $A$. 
    
    One can think of \eqref{eq: pLSI-Bayes} as equivalent to requiring that the following Bayes formula holds for any word $j$ and document $i$:
    \begin{equation}
        \mathbb{P}(\text{word } j | \text{ document } i) = \sum_{k = 1}^K \mathbb{P}(\text{word }j | \text{ topic }k) \cdot \mathbb{P}(\text{topic }k | \text{ document }i)
    \end{equation}
    
    In most applications, $K \ll \min(n,p)$ and thus \eqref{eq: pLSI-Bayes} impose a low-rank structure on $\mathbb{E}(D)$. We note that the number of topics $K$ plays a role similar to that of the number of principal components in principal component analysis. For technical reasons, we will assume throughout this paper that $K$ is fixed as $n$, $p$ and the document lengths $N_i$'s vary. This is reasonable if one expects \textit{a priori} that the number of topics covered by the corpus is small and bounded.

    \subsection{Related works and unaddressed issues}

    Before outlining our contributions in Section \ref{sec: contributions}, it is important to provide context by discussing previous works that are relevant to the estimation of $A$ under the pLSI model. In particular, we want to highlight some of the unaddressed issues from prior papers that our work aims to resolve.

    \subsubsection{The separability condition}\label{section: separability}

    We first present the definition of anchor words and the separability condition. 
    \begin{definition}[Anchor words and separability]
        We call word $j$ an anchor word for topic $k$ if row $j$ of $A$ has exactly one nonzero entry at column $k$. The separability condition is said to be satisfied if there exists at least one anchor word for each topic $k \in \{1, \dots, K\}$.  
    \end{definition}

    Observe that the decomposition $D_0 = AW$ in general may not be unique, but under the separability condition, $A$ is identifiable. The separability condition was first introduced in \cite{donoho2003does} to ensure uniqueness in the Non-negative Matrix Factorization (NMF) framework. The interpretation in our context is that, for each topic, there exist some words which act as unique signatures for that topic.

    The separability condition greatly simplifies the problem of estimating $A$, as one can identify the anchor words for each topic as a first step. Prior works exploiting anchor words mainly differ in how anchor words are used to estimate the remaining non-anchor rows of $A$. \cite{arora2012learning} start from the \textit{word co-occurrence matrix} $DD^T$ and apply a successive projection algorithm to rows of $DD^T$ to find one anchor word per topic. The matrix $DD^T$ is then re-arranged into four blocks where the top left $K \times K$ block corresponds to the anchor words identified, and $A$ is estimated by taking advantage of the special structure of this block partition. More recently, \cite{bing2020optimal} consider a matrix $B \in \mathbb{R}^{p \times K}$ obtained from $A$ via multiplication by diagonal matrices. Unlike $A$, all rows of $B$ sum up to 1, so anchor rows of $B$ are simply canonical basis vectors in $\mathbb{R}^K$. The non-anchor rows of $B$ are then obtained via regression given the anchor rows of $B$. The topic matrix $A$ can subsequently be recovered through an appropriate normalization of $B$.

    A major drawback of these methods is that they rely heavily on the separability assumption, which suffices for uniqueness of the decomposition \eqref{eq: pLSI-Bayes} but is far from necessary. 
This issue is related to the following question, which is of central importance in the NMF literature:
 given a collection of points $\{r_1, \dots, r_m\} \subseteq \mathbb{R}^{K-1}$ presumed to lie within 
    the convex hull of unobserved vertices $\{v_1^*, \dots, v_K^*\}$, when is recovery of these vertices possible? In the NMF context, separability means that each vertex coincides with a point in the observed point cloud, in which case we only need to identify which of the $r_i$'s correspond to simplex vertices. However, this is a very strong assumption and several efforts have been made to relax it. \cite{javadi2020nonnegative} show that vertex recovery is still possible under a uniqueness assumption that generalizes separability. \cite{ge2015intersecting} introduce the notion of \textit{subset separability} which is also much weaker than separability. We note that many of the separability-based methods proposed in topic modeling, such as those in \cite{arora2012learning}, \cite{bing2020fast} and \cite{bing2020optimal}, have no obvious extension if the separability assumption is relaxed. This may not be important if the given corpus contains many specialized words and the topics are sufficiently distinct (an example is a collection of research
    papers), but may matter more if the topics overlap significantly and the vocabulary is generic (for instance, a collection of high school English essays).

    \subsubsection{\texorpdfstring{The SVD-based approach in \cite{ke2022using}}{The SVD-based approach in Ke and Wang (2022)}}
    
    \cite{ke2022using} are the first to establish the minimax-optimal rate of $\sqrt{\frac{p}{nN}}$ for the  $\ell_1$-loss $\|\hat{A}-A\|_1 := \sum_{j=1}^p \sum_{k=1}^K |\hat{A}_{jk} - A_{jk}|$ where, for simplicity, all document lengths are  assumed to be equal to $N$.
Their procedure links topic estimation to the NMF setting discussed in the previous subsection and is summarized as follows. Let $M := \diag(n^{-1} D\mathbf{1}_n)$ where $\mathbf{1}_n := (1, 1, \dots, 1)^T \in \mathbb{R}^n$. Given $K$, the approach proposed in \cite{ke2022using} considers the first $K$ left singular vectors $\check{\xi}_1, \dots, \check{\xi}_K \in \mathbb{R}^p$ of $M^{-1/2}D$. Elementwise division of $\check{\xi}_2, \dots, \check{\xi}_K$ by $\check{\xi}_1$ (also known as \textit{SCORE normalization} \citep{jin2015fast}) yields a matrix $\check{R}\in \mathbb{R}^{p \times (K-1)}$, whose rows $\check{r}_1, \dots, \check{r}_p \in \mathbb{R}^{K-1}$ can be shown to form a point cloud contained in a $K$-vertex simplex (up to stochastic errors). Since this corresponds precisely to the NMF setup discussed in the previous subsection, the simplex vertices can now easily be recovered using a suitable \textit{vertex hunting} algorithm. Once these vertices are identified, $A$ can then be estimated via a series of normalizations.

    The work by \cite{ke2022using} is an important contribution that motivates several other methods for topic modeling, including ours. 
    However, this method was developed using strong assumptions on the parameter regimes and the behavior of word frequencies. 
    More specifically, Corollary 3.1 of \cite{ke2022using} states that the error upper bound $\sqrt{\frac{p \log n}{nN}}$ is only applicable if we assume $N > p^{4/3}$ or $p \leq N < p^{4/3}$ and $n \geq \max(Np^2, p^3, N^2p^5)$. As the vocabulary size $p$ is typically large, these are highly unrealistic assumptions on $(n, N, p)$. For example, the Associated Press (AP) dataset used in \cite{ke2022using} (a corpus of news articles frequently used for topic model evaluation) has $n =  2,134$ and $p = 7,000$. A typical AP article has between $300$ and $700$ words, so it is clear that none of the above assumptions holds. The error bound provided without these assumptions is $\frac{p^2}{N\sqrt{N}}\sqrt{\frac{p\log n}{nN}}$, which, when $p$ is large and grows with $n$, may not necessarily converge to zero. Several other works that claim to establish minimax-optimal rates also do so by assuming $N > p$; see Theorem 4.1 of \cite{wu2022sparse} and Remark 10 of \cite{bing2020fast}.
    
    In this paper, we do not seek to re-establish the rate $\sqrt{\frac{p}{nN}}$. Rather, we aim to provide a consistent error bound valid for all realistic parameter regimes (especially when $p > \max(n, N)$).
    We propose to resolve some of the outstanding issues of the estimator in \cite{ke2022using} by leveraging a sparsity structure that is often empirically observed in text documents, resulting in:
    \begin{enumerate}
        \item {\bf Improved error bounds: } We observe that even the minimax-optimal rate $\sqrt{\frac{p}{nN}}$ of \cite{ke2022using} scales significantly with $p$. As the number of documents $n$ increases, we can expect several previously unobserved words to be added to the corpus, whereas the average document length $N$ may not change by much. However, many of these words may occur rarely, so the effective dimension of the parameter space may be quite small compared to the observed vocabulary size. This motivates us to restrict the parameter space by imposing a suitable column-wise sparsity assumption on $A$, which enables an error bound that does not scale with $p$ except for log factors. 
        \item {\bf An increased signal-to-noise ratio:} The approach in \cite{ke2022using} may not be suitable if many words in the corpus occur with low frequency. If for each word $j$ we define $h_j := \sum_{k=1}^K A_{jk}$, the theoretical guarantees in \cite{ke2022using} require $\min_{1 \leq j \leq p} h_j \geq \frac{cK}{p}$ for some $c \in (0,1)$. Note that since the columns of $A$ sum up to 1, we always have $\frac{1}{p}\sum_{j=1}^p h_j = \frac{K}{p}$. Therefore, since $h_j$ roughly indicates the frequency of word $j$ in the corpus, this assumption restricts the frequencies of the least frequent words to be of the same order as the average frequency of all words. 
        

        Such a restrictive assumption is needed in \cite{ke2022using} because when many low-frequency words exist in the corpus, their procedure involves division by small and noisy numbers. This is a problem with their \textit{pre-SVD normalization} step where $D$ is pre-multiplied by the $p\times p$ diagonal matrix $M^{-1/2}$, as the diagonal entries of $M:= \diag(n^{-1} D\mathbf{1}_n)$ corresponding to infrequent words are usually small. This is also an issue with their elementwise division step, thus leading to higher errors from infrequent words in the point cloud obtained from their procedure (see Figures ~\ref{fig:point-cloud-nlp} and ~\ref{fig:point-cloud-mic} for illustration). Although we also use SCORE normalization \citep{jin2015fast}, our removal of infrequent words leads to a point cloud with a higher signal-to-noise ratio. 
    \end{enumerate}
    \subsubsection{Sparse topic modeling approaches}
        To our knowledge, \cite{bing2020optimal} and \cite{wu2022sparse} are the only two prior works that, like ours, impose additional sparsity constraints on $A$. However, the sparsity assumptions proposed in these papers are not appropriate for dealing with large $p$; rather, they are more suitable for dealing with large $K$.
        \begin{enumerate}
            \item \cite{bing2020optimal} assume that $A$ is elementwise sparse, in the sense that the total number of nonzero entries of $A$ (denoted as $\|A\|_0$) is small. Their proposed procedure is then shown to satisfy the error upper bound 
            \begin{equation}\label{eq:Bing-rate}
                \|\hat{A}- A\|_1 \lesssim K\sqrt{\frac{\|A\|_0 \log(p \vee n)}{nN}} 
            \end{equation}
            We note here that $\|A\|_0$ can still be very large. Indeed, let $\tilde{p}$ denote the number of words whose corresponding rows in $A$ are \textit{not} entirely zero. Technically we can have $p > \tilde{p}$, but words corresponding to zero rows of $A$ are not observed with probability one, so $\tilde{p}$ covers the entire set of all distinct words observed in the corpus. We have
            \begin{equation}\label{eq:Bing-sparsity}
                \tilde{p} \leq \|A\|_0 \leq Kp
            \end{equation}
            
            In fact, one can see that their error bound depends on $p$ from the error decomposition $\|\hat{A} - A\|_1 \lesssim \text{I} + \text{II} + \text{III}$ in Theorem 2 of \cite{bing2020optimal}. For example, $\text{I} = \frac{K}{\underline{\gamma}}\sqrt{\frac{p \log (n\vee p)}{nN}} + \frac{pK\log (p\vee n)}{\underline{\gamma}nN}$ for some constant $\underline{\gamma}$. This, together with \eqref{eq:Bing-sparsity}, shows that the bound $\eqref{eq:Bing-rate}$ is not very different from the rate $\sqrt{\frac{p}{nN}}$ in \cite{ke2022using}, except for possibly better dependence on $K$.
            Moreover, their theoretical results depend on several strong assumptions on the frequency of anchor words selected by their procedure. In contrast, our procedure is less affected by the frequency of anchor words, both in theory and in practice. 
            
            \item \cite{wu2022sparse} assume that each row of $A$ has at most $s_A$ nonzero entries. Since $A$ has $K$ columns, this sparsity assumption is only useful if $K$ is large. Theorem 4.1 of \cite{wu2022sparse} then shows that their proposed estimator of $A$ satisfies 
            \begin{equation}\label{eq: Wu-bound}
                \|\hat{A} - A\|_1 \lesssim K\sqrt{\frac{s_A\log n}{nN}}
            \end{equation}
            However, upon close examination of their proof, the $\ell_1$ bound they achieve is actually $K\sqrt{\frac{\|A\|_0\log n}{nN}}$ (similar to \eqref{eq:Bing-rate}) so \eqref{eq: Wu-bound} is only possible by assuming that $p = O(1)$ and using $\|A\|_0 \leq p s_A$. Furthermore, their result assumes $N^{3/4} \geq p$ which, as we have noted in our discussion of \cite{ke2022using}, is highly restrictive. 
            
         \end{enumerate}
         In comparison with these two papers, our sparsity assumption is more compatible with the ``large $p$'' setting, and we do not assume $p = O(1)$ as in \cite{wu2022sparse}. 
    \subsection{Our contributions}\label{sec: contributions}
        We summarize the main contributions of this paper below.
        \begin{itemize}
            \item We propose a new spectral procedure (Definition \ref{actual-procedure}) for estimating $A$. This procedure takes into account the observation that, in most text datasets, the vocabulary size $p$ is often large but many words occur very infrequently in the corpus. When $K$ is unknown, a new estimator of $K$ is also proposed (see Lemma \ref{lem: K-hat}). 
            \item We introduce a new column-wise $\ell_q$-sparsity assumption (Assumption \ref{ass:sparsity}) for $A$. This assumption is motivated by Zipf's law \citep{zipf2013psycho} and links a word's frequency of occurrence in a topic to its rank. Our proposed procedure is then shown to be adaptive to the unknown sparsity level $s$ in the $\ell_q$-sparsity definition \eqref{eq: lq-sparsity-def}. 
            \item We provide an error bound for our procedure using the $\ell_1$ loss $\|\hat{A}-A\|_1$ in Theorem \ref{theorem: main-result}. Under our sparsity assumption \eqref{eq: lq-sparsity-def}, our error bound is shown to be valid for all parameter regimes and only depends on $p$ via weak factors. The common pre-processing step of removing infrequent words is incorporated into our procedure and accounted for in our analysis.
            \item Finally, in Section \ref{sec: relax}, we show that our theoretical results may still hold when the separability assumption is relaxed if we choose a suitable vertex hunting procedure for non-separable point clouds in Definition \ref{actual-procedure}.
        \end{itemize}

    Extensive experiments with synthetic datasets to confirm the effectiveness of our estimation procedure under a wide variety of parameter regimes are presented in Section \ref{sec:experiments}. Furthermore, we also demonstrate the usefulness of our method for text analysis, as well as for other applications where the pLSI model is also relevant, in Section \ref{sec: real-data}. 
        
    \subsection{Notations}\label{section: notations}
    For any set $S$, let $|S|$ denote its cardinality, and let $S^c$ denote its complement if it is clear in context with respect to which superset. For any $k \in \mathbb{N}$, let $[k]$ denote the index set $\{1, \dots, k\}$. We use $\mathbf{1}_d$ to denote the vector in $\mathbb{R}^d$ with all entries equal to 1. For a general vector $v \in \mathbb{R}^d$, let $\|v\|_r$ denote the vector $\ell_r$ norm, for $r =0, 1, \dots, \infty$, and let $\diag(v)$ denote the $d \times d$ diagonal matrix with diagonal entries equal to entries of $v$. For any $a, b \in \mathbb{R}$, let $a \vee b := \max(a, b)$ and $a \wedge b := \min(a,b)$. 
    

    Let $I_m$ denote the $m \times m$ identity matrix. For a general matrix $Q \in \mathbb{R}^{m \times l}$ and $r = 0, 1, \dots, \infty$, let $\|Q\|_r$ denote the vector $\ell_r$-norm of $Q$ if one treats $Q$ as a vector. Let $\|Q\|_F$ and $\|Q\|_\text{op}$ denote the Frobenius (i.e. $\|Q\|_F = \|Q\|_2$) and operator norms of $Q$ respectively. For any index $j \in [m]$ and $i \in [l]$, let $Q_{ji}$ or $Q(j,i)$ denote the $(j,i)$-entry of $Q$. For index sets $J \subseteq [m]$ and $I \subseteq [l]$, let $Q_{JI}$ denote the submatrix of $Q$ obtained by selecting only rows in $J$ and columns in $I$ (in particular, either $J$ or $I$ can be a single index). Also, let $Q_{J*}$ denote the submatrix of $Q$ obtained by selecting rows in $J$ and \textit{all} columns of $Q$; $Q_{*I}$ is similarly defined. This means $Q_{j*}$ and $Q_{*i}$ denote the $j^\text{th}$ row and $i^{th}$ column of $Q$ respectively. For an integer $k \leq m \wedge l$, let $\gamma_k(Q)$ denote the $k^\text{th}$ largest singular value of $Q$, and if $Q$ is a square matrix then, if applicable, let $\lambda_k(Q)$ denote the $k^\text{th}$ largest eigenvalue of $Q$. If $m = l$, then $\text{tr}(Q)$ denotes the trace of $Q$. 
    
    Let $C, c >0$ denote absolute constants that may depend on $K$ and $q$; we assume that $K$ and $q$ are fixed, unobserved constants. Let $C^*, c^* > 0$ denote numerical constants that do not depend on the unobserved quantities like $K$ and $q$ (this only matters when we discuss the estimation of $K$). The constants $C, c , C^*, c^*$ may change from line to line. 

    In our paper, for ease of presentation, we assume $N_1 = \dots = N_n = N$. Our results also hold if we assume the document lengths satisfy $\max_{i \in [n]} N_i \leq C^* \min_{i\in [n]} N_i$ (i.e. if $N_1 \asymp \dots \asymp N_n$), in which case $N = \frac{1}{n}\sum_{i=1}^n N_i$ denotes the average document length.

\section{Our procedure for estimating \texorpdfstring{$A$}{A} and its theoretical properties}\label{sec:sec1}
For simplicity, we will first assume separability in order to explain our procedure. A discussion of possible relaxations of this condition will be deferred to Section \ref{sec: relax}. 
\subsection{The oracle procedure to estimate \texorpdfstring{$A$}{A} given \texorpdfstring{$D_0$}{D0}}
Our oracle procedure concerns how $A$ can be estimated if the non-stochastic matrix $D_0$, rather than $D$, is observed. Let $J \subseteq [p]$ be an arbitrary collection of words in our vocabulary. We first need the definition of a vertex hunting procedure, which is relevant to the NMF setup discussed in Section \ref{section: separability}. 

\begin{definition}[Vertex hunting] Given $K$, a vertex hunting procedure is a function that takes a collection of points in $\mathbb{R}^{K-1}$ and returns $K$ points in $\mathbb{R}^{K-1}$. 
\end{definition}

\begin{remark}{\rm A good vertex hunting procedure should return the vertices of the smallest $K$-simplex containing the given point cloud. Throughout the paper, we will use $\mathcal{V}(\cdot)$ to denote such a procedure.}
\end{remark}

The following definition of an ideal point cloud is based on the separability assumption. Any reasonable vertex hunting procedure should be able to successfully recover the simplex vertices from an ideal point cloud. 

\begin{definition}[Ideal point cloud]\label{def: ideal-pt-cloud} Given $K$, an ideal point cloud is a collection of points in $\mathbb{R}^{K-1}$ contained in the simplex defined by $K$ vertices, such that the $K$ vertices themselves belong to the point cloud. 
\end{definition}

We are now ready to define the oracle procedure. 

\begin{definition}[Oracle procedure]\label{def: oracle-procedure} Given inputs $K$, $D_0$, vertex hunting procedure $\mathcal{V}(\cdot)$ and a set of words $J \subseteq [p]$, the oracle procedure returns $\tilde{A} \in \mathbb{R}^{p \times K}$ defined as follows: 
\begin{enumerate}
    \item (SVD) Perform SVD on $[D_0]_{J*}$ to obtain $\Xi = [\xi_1, \dots, \xi_K] \in \mathbb{R}^{|J|\times K}$ containing the first $K$ left singular vectors of $[D_0]_{J*}$. 
    \item (Elementwise division) Divide $\xi_2, \dots, \xi_K$ elementwise by $\xi_1$ to obtain $R \in \mathbb{R}^{|J| \times (K-1)}$. This means $R_{jk} = \xi_{k+1}(j) / \xi_1(j)$, for $k = 1, \dots, K-1$ and $j \in J$.
    \item (Vertex hunting) Treat the rows $\{r_j: j \in J\}$ of $R$ as a point cloud in $\mathbb{R}^{K - 1}$. Apply the vertex hunting procedure $\mathcal{V}(\cdot)$ on this point cloud to obtain vertices $v_1^*, \dots, v_K^*$. 
    \item (Recovery of $\Pi$) For each $j \in J$, solve for $\pi_j \in \mathbb{R}^K$ from the linear equation 
    \begin{equation}\label{eq: Pi-estimation-8}
        \begin{pmatrix}
            1 & \dots & 1 \\ v_1^* & \dots & v_K^*
        \end{pmatrix} \pi _j = \begin{pmatrix}
            1 \\ r_j
        \end{pmatrix}
    \end{equation}
    In other words, $\pi_j$ satisfies $\sum_{k=1}^K \pi_j(k) = 1$ and $ r_j = \sum_{k=1}^K \pi_j(k) v_k^*$, for each $j \in J$. Let $\Pi \in \mathbb{R}^{|J|\times K}$ be the matrix whose rows are $\{\pi_j: j \in J\}$.
    \item (Normalization) Normalize the columns of $\diag(\xi_1) \cdot \Pi \in \mathbb{R}^{|J| \times K}$ so that the entries of each column sum up to 1. This yields $\tilde{A}_{J*}$. Set $\tilde{A}_{J^c*} = 0$ to obtain $\tilde{A}$.  
\end{enumerate}
\end{definition}

    Our oracle procedure makes use of the SCORE normalization idea which was originally proposed for network data analysis \citep{jin2015fast}. The elementwise division step (Step 2) is the most important step, as it provides a connection between singular vectors of $D_0$ (or associated variables) and the NMF setup described in Section \ref{section: separability}. The words in $J$ are represented by the point cloud $\{r_j: j \in J\}$, which can be shown to be contained entirely in some $K$-vertex simplex. If the simplex vertices are identifiable and the vertex hunting procedure is successful in recovering them in Step 3, then \eqref{eq: Pi-estimation-8} allows us to exactly recover the probabilistic weights $\{\pi_j: j \in J\}$ associated with each word in $J$, which are connected to $A$ via the relation 
    \begin{equation}\label{eq: relation-9}
        \diag(\xi_1) \cdot \Pi = A_{J*} \cdot \diag(V_1)
    \end{equation}
    for some vector $V_1 \in \mathbb{R}^K$ containing only positive entries. This explains the column normalization step (Step 5), which essentially reverses the elementwise division step. For more details, we refer the reader to the proof of Lemma \ref{lem:B2} in the appendix.

Based on the relation \eqref{eq: relation-9}, we can show the following result.
\begin{lemma}\label{lem: tilde-A}
    Suppose the set $J$ contains at least one anchor word for each topic $k \in [K]$, and the vertex hunting procedure $\mathcal{V}(\cdot)$ can successfully recover the simplex vertices from any ideal point cloud. The oracle procedure in Definition \ref{def: oracle-procedure} then returns $\tilde{A}$ satisfying $\tilde{A}_{J^c*} = 0$ and
    \begin{equation}\label{eq: A-tilde-relation}
        \tilde{A}_{J*} = A_{J*}\cdot \diag(\|A_{J1}\|_1^{-1}, \dots, \|A_{JK}\|_1^{-1})
    \end{equation}
\end{lemma}
The proof of Lemma \ref{lem: tilde-A} is identical to that of Lemma \ref{lem:B2} provided in Appendix B. The sole difference is that in Lemma \ref{lem:B2}, the set $J$ is chosen as in \eqref{J-def}  and we use Assumption \ref{ass: separability}. 

\begin{remark}{\rm
    Our oracle procedure differs from that of \cite{ke2022using} in two important ways. First, note that Step 1 only requires SVD to be performed on a submatrix of $D_0$. In general, we want the set $J$ to contain words that occur with sufficiently high frequencies in the corpus so that the point cloud generated from our procedure has a higher signal-to-noise ratio. When $p$ is large, we can often expect the corpus to contain many infrequently occurring words whose corresponding rows in $A$ should be estimated as zero. Our oracle procedure yields $\tilde{A}$ which is a good oracle approximation of $A$ if $\|A_{J^c*}\|_1$ is small, as in that case the diagonal matrix in \eqref{eq: A-tilde-relation} is close to the identity matrix. 

    Second, note that we consider the SVD of a submatrix of $D_0$ and not $M_0^{-1/2}D_0$ as in \cite{ke2022using}, where $M_0:= \diag(n^{-1}D_0\mathbf{1}_n)$. This simplifies some parts of our theoretical analysis and allows us to obtain error bounds that depend less strongly on $p$ (see Appendix \ref{section: C}). }
\end{remark}
\subsection{Estimation procedure for \texorpdfstring{$A$}{A} given \texorpdfstring{$D$}{D}}
    Our procedure to estimate $A$ below is designed to closely approximate the oracle procedure. Here we first assume $K$ is known. The estimation of $K$ is deferred to Section \ref{sec: est-of-K}, and the choice of the vertex hunting procedure will be discussed in conjunction with identifiability assumptions on $A$. 
    \begin{definition}[Estimation procedure for $A$]\label{actual-procedure} Given inputs $K$, observation matrix $D$ and vertex hunting procedure $\mathcal{V}(\cdot)$, our estimation procedure returns $\hat{A}$ defined as follows: 
    \begin{enumerate}
        \item (Thresholding) Let $M := \diag(n^{-1}D\mathbf{1}_n)$ and $p_n := p \vee n$. Compute the set of words
        \begin{equation}\label{J-def}
            J := \left\{j \in [p]: M(j, j) \geq \alpha \sqrt{\frac{\log p_n}{nN}}\right\}
        \end{equation}
        Here, $\alpha$ is a user-specified universal constant (see Remark \ref{rem: choice-of-alpha}). 
        \item (Spectral decomposition) Compute the first $K$ eigenvectors $\hat{\xi}_1, \dots, \hat{\xi}_K \in \mathbb{R}^{|J|}$ of the submatrix $G_{JJ}$ of the $p \times p$ matrix $G$, where 
        \begin{equation}\label{def-of-G}
            G:= DD^T - \frac{n}{N} M
        \end{equation}
        Here, we assume all entries of $\hat{\xi}_1$ are of the same sign, in which case we can choose $\hat{\xi}_1$ to have all positive entries. If some entries of $\hat{\xi}_1$ are negative, choose $\hat{\xi}_1$ such that the majority of entries are positive, and apply Remark \ref{rem: xi1-positive}.
        \item (Elementwise division) Divide $\hat{\xi}_2, \dots, \hat{\xi}_K$ elementwise by $\hat{\xi}_1$ to obtain $\hat{R} \in \mathbb{R}^{|J| \times (K-1)}$, with rows $\{\hat{r}_j: j \in J\}$. This means $\hat{R}_{jk} = \hat{\xi}_{k+1}(j)/\hat{\xi}_1(j)$, for $k \in [K-1]$ and $j \in J$. 
        \item (Vertex hunting) Treat the rows of $\hat{R}$ as a point cloud in $\mathbb{R}^{K-1}$. Apply the vertex hunting procedure $\mathcal{V}(\cdot)$ to this point cloud to obtain vertices $\hat{v}_1^*, \dots, \hat{v}_K^*$.  
        \item (Estimation of $\Pi$) For each $j \in J$, solve for $\hat{\pi}_j^\diamond \in \mathbb{R}^K$ from 
        \begin{equation}
            \begin{pmatrix}
                1 & \dots & 1 \\ \hat{v}_1^* & \dots & \hat{v}_K^*
            \end{pmatrix} \hat{\pi}_j^\diamond = \begin{pmatrix}
                1 \\ \hat{r}_j
            \end{pmatrix}
        \end{equation}
        Obtain $\hat{\pi}_j$ from $\hat{\pi}_j^\diamond$ by first setting any negative entries to 0 and then normalizing so that the entries of $\hat{\pi}_j$ sum up to 1. Let $\hat{\Pi}$ be the matrix whose rows are $\{\hat{\pi}_j: j \in J\}$.
        \item (Normalization) Normalize all columns of $\diag(\hat{\xi}_1) \cdot \hat{\Pi}$ so that they have unit $\ell_1$-norm. This yields $\hat{A}_{J^*}$. Set all entries of $\hat{A}_{J^c*}$ to zero to obtain $\hat{A}$. 
    \end{enumerate}
    \end{definition}
    As Steps 3-5 are also based on the SCORE normalization idea \citep{jin2015fast}, we call this procedure the \textit{Thresholded Topic-SCORE} (TTS). However, Step 1, Step 2 and Step 6 contain significant differences when compared with Topic-SCORE in \cite{ke2022using}. 
    \begin{remark}[Choice of $\alpha$]{\rm \label{rem: choice-of-alpha}
        The set $J$ in \eqref{J-def} is chosen by examining the row sums of the observation matrix $D$, which indicate how frequently the words occur in the corpus. In \eqref{J-def}, $\alpha$ is meant to be a universal constant and thus does not affect our error rates, which are not optimized over constants. In our theoretical discussion, we choose $\alpha = 8$ for convenience, but for most datasets this value of $\alpha$ may result in too many words not meeting the threshold. 

        In practice, a good choice of $\alpha$ is important for obtaining a good estimator of $A$. Based on our experiments, we recommend a smaller value of $\alpha$, such as $\alpha = 0.005$. This choice of $\alpha$ should produce reasonable results for commonly observed values of $(n, N, p)$. Based on what we observe from experiments, if $n \in [1000, 5000], N \in [300, 700], p \in [5000, 10000]$, we can typically expect around 10-40\% of words to be removed.} 
    \end{remark}
    \begin{remark}[Signs of $\hat{\xi}_1$'s entries]\label{rem: xi1-positive} {\rm In the oracle procedure, $\xi_1$ is the first left singular vector of $[D_0]_{J*}$ and so by Perron's theorem, the entries of $\xi_1$ are all positive. In Step 2, $\hat{\xi}_1$ is the first eigenvector of $G_{JJ}$ which is not necessary a Perron matrix, so $\hat{\xi}_1$ technically may contain negative entries. Any word $j$ for which $\hat{\xi}_1(j)$ is negative should have corresponding rows of $A$ set to zero after Step 2, and then in Step 3 we form the point cloud by computing $\hat{\xi}_{k+1}(j)/\hat{\xi}_1(j)$ for $k \in [K-1]$ and $j \in J$ with $\hat{\xi}_1(j) > 0$ only. 

    In our theoretical analysis as well as in practice, however, this scenario will not happen with high probability. This is because $G$ is chosen so that $\max_{j \in J}|\hat{\xi}_1(j) - \xi_1(j)|$ is small. Since any word $j$ that meets our threshold occurs with sufficiently high frequency, $\xi_1(j)$ will also be sufficiently large for any $j \in J$, which implies $|\hat{\xi}_1(j) - \xi_1(j)| \ll \xi_1(j)$ and thus $\hat{\xi}_1(j) \geq \xi_1(j)/2$ for all $j \in J$. See Lemmas \ref{lem:D8} and \ref{lem: D.9} in the appendix.}         
    \end{remark}
    
    The set $J$ as defined in \eqref{J-def} is data-dependent. It is quite useful to note that $J$ can be approximated by the non-stochastic sets \eqref{Jpm-def} with high probability. The proof of the lemma below can be found in Theorem \ref{lem:A3}(b) in the appendix. 
    \begin{lemma}\label{lem: J-sandwiched}
        Let $M_0 := \diag(n^{-1}D_0\mathbf{1}_n)$, and let 
        \begin{equation}\label{Jpm-def}
            J_\pm := \left\{j \in [p]: M_0(j,j) > \alpha_\pm \alpha \sqrt{\frac{\log p_n}{nN}}\right\}
        \end{equation}
        where $\alpha$ is from the definition of $J$ in \eqref{J-def} and $\alpha_- > 1$ and $\alpha_+ \in (0,1)$ are some suitably chosen constants depending on $\alpha$ (for example if $\alpha = 8$, we can let $\alpha_+ = \frac{1}{2}$, $\alpha_- = 2$). 
        Then the event $\mathcal{E} := \{J_- \subseteq J \subseteq J_+\}$
        occurs with probability at least $1 - o(p_n^{-1})$. 
    \end{lemma}
    The following lemma bounds the size of $J$, and is obtained by bounding $|J_+|$ and using Lemma \ref{lem: J-sandwiched}.  
    \begin{lemma}[Size of $J$] With probability at least $1 - o(p_n^{-1})$, 
        \begin{equation}\label{J-size-bound}
        |J| \leq \left(\frac{K}{\alpha \alpha_+}\sqrt{\frac{Nn}{\log p_n}}\right) \wedge p 
    \end{equation}
    \end{lemma}
    Our procedure requires the eigenvalue decomposition of a symmetric $|J|\times |J|$ matrix. The bound \eqref{J-size-bound} can be significantly smaller than $\min(n,p)$ if $nN \ll p^2$ and $N \ll n$ (ignoring weak factors), which are reasonable assumptions for many text datasets. We can therefore expect the eigenvalue decomposition step in our procedure to be more computationally scalable than the SVD step (on a $p \times n $ matrix) in \cite{ke2022using}.

    \subsection{\texorpdfstring{Error bounds for $\hat{A}$ under separability}{Error bounds under separability}}
    We first discuss our theoretical results under separability, which is assumed in all of our proofs in the appendix. We begin by listing the assumptions underlying our analysis. 
    
    \begin{ass}[$A$ and $W$ are well-conditioned]\label{ass:1} Let $\Sigma_W := n^{-1}WW^T$. For some constant $c \in (0,1)$,
        \begin{equation}\label{eq:AW-well-conditioned}
            \sigma_K(A) \geq c\sqrt{K} \quad \text{and} \quad   \sigma_K(\Sigma_W) \geq c
        \end{equation}
    \end{ass}
    \begin{ass}[The topic-topic correlation matrix is regular]\label{ass:2} The entries of $A^TA$ satisfy the following for some constant $c > 0$:
        \begin{equation}\label{eq: topic-topic-correlation}
            \min_{1 \leq k,l \leq K} A^TA(k,l) \geq c
        \end{equation}        
    \end{ass}
    \begin{ass}[Separability]\label{ass: separability}
        Each topic $k \in [K]$ has at least one associated anchor word $j$ belonging to the set $J_-$ defined in \eqref{Jpm-def}. 
    \end{ass}
    \begin{ass}[Vertex hunting efficiency]\label{ass:VH} Given $K$ and an ideal point cloud defined in Definition \ref{def: ideal-pt-cloud}, the vertex hunting function $\mathcal{V}(\cdot)$ recovers the $K$ vertices correctly. Furthermore, whenever $\mathcal{V}(\cdot)$ is given as inputs two point clouds $\{x_1, \dots, x_m\}$ and $\{x_1', \dots, x_m'\}$, the outputs $\{v_1, \dots, v_K\}$ and $\{v_1', \dots, v_K'\}$ satisfy for some absolute constant $C > 0$ (up to a label permutation)
    \begin{equation}
        \max_{k \in [K]}\|v_k - v_k'\|_2 \leq C \max_{j \in [m]}\|x_j - x_j'\|_2
    \end{equation}
    \end{ass}
    \begin{ass}[Column-wise $\ell_q$-sparsity]\label{ass:sparsity} Let the entries of each column $A_{*k}$ of $A$ be ordered as $A_{(1)k} \geq \dots \geq A_{(p)k}$. For some $q \in (0,1)$ and $s > 0$, the columns of $A$ satisfy 
        \begin{equation}\label{eq: lq-sparsity-def}
            \max_{k \in [K]} \left(\max_{j\in [p]} j A_{(j)k}^q \right) \leq s
        \end{equation}   
    Here, we assume that $q$ is a fixed constant, whereas $s$ is allowed to grow with $n$. 
    \end{ass}

    \begin{remark}{\rm We justify why Assumptions \ref{ass:1}, \ref{ass:2} and \ref{ass: separability} are reasonable below.
        \begin{enumerate}
            \item Equation \eqref{eq:AW-well-conditioned} assumes the topic vectors in $A$ are not too correlated. The assumption on $W$ in \eqref{eq:AW-well-conditioned} is necessary even when $W$ is known, as its role is similar to that of the design matrix in the regression setting. Note that since the columns of $A$ and $W$ sum up to 1, we always have $\sigma_1(A) \leq \sqrt{K}$ and $\sigma_1(\Sigma_W) \leq 1$ (see Lemma \ref{lem:B1}(a) in the appendix). 
            \item The matrix $A^TA \in \mathbb{R}^{K \times K}$ can be thought of as the topic-topic correlation matrix, since its entries are inner products of the columns of $A$. Therefore, \eqref{eq: topic-topic-correlation} is especially true if the $K$ topics are related to one another. However, even if the corpus covers unrelated topics, we expect all columns of $A$ to assign significant weights to grammatical function words (such as `and', `the' in English) and filler words, which occur frequently in all documents regardless of the topics involved. 
            \item In light of Lemma \ref{lem: J-sandwiched}, Assumption \ref{ass: separability} requires that each topic has at least one anchor word that occurs in the corpus frequently enough so that it is included in $J$. Such an assumption on the frequency of anchor words is also commonly seen in other works that exploit the separability condition, and Assumption \ref{ass: separability} is not strong since the threshold level of order $\sqrt{\frac{\log p_n}{nN}}$ in the definition of $J_-$ is quite low. For comparison, \cite{bing2020optimal} makes the same assumption but with the threshold level of order $\frac{\log p_n}{N}$, which may be higher than ours if the number of documents $n$ far exceeds the average document length $N$. 
        \end{enumerate}}
    \end{remark}
   \begin{remark}[Vertex hunting for separable point clouds]\rm{\cite{ke2022using} mentions two vertex hunting algorithms which are suitable for separable point clouds, namely Successive Projection (SP) \citep{araujo2001successive} and Sketched Vertex Search (SVS) \citep{jin2017estimating}.
    
    Given a point cloud $r_1, \dots, r_m$, SP starts by finding the point $r_j$ whose Euclidean norm is the largest and sets this as the first estimated vertex $\hat{v}_1$. Then, for each $2 \leq k \leq K$, we can obtain $\hat{v}_k$ from $\hat{v}_1, \dots, \hat{v}_{k-1}$ by setting $\hat{v}_k$ as the point $r_j$ that maximizes $\|(I-P_{k-1})r_j\|_2$, where $P_{k-1}$ denotes the projection matrix on the linear span of $\hat{v}_1, \dots, \hat{v}_{k-1}$. SP can be shown to satisfy Assumption \ref{ass:VH} when the volume of the true simplex is lower bounded by a constant \citep{gillis2013fast}, which is a simple consequence of Theorem \ref{lem:B1}(f) in the appendix.
    
    On the other hand, SVS starts by applying $k$-means clustering on the point cloud $\{r_1, \dots, r_m\}$ to obtain cluster centers $\hat{c}_1, \dots, \hat{c}_L$, where $L$ is a tuning parameter that is much larger than $K$. These clusters are meant to reduce the noise levels in the point cloud. Next, SVS exhaustively searches for all simplexes whose $K$ vertices are located on these cluster centers, in order to find the simplex $S$ such that the maximum distance from any $\hat{c}_l$ to $S$ is minimized. In comparison to SP, SVS is more robust to noise in the point cloud but is computationally much slower if $K$ is not small. SVS satisfies Assumption \ref{ass:VH} under mild regularity conditions \citep{jin2017estimating}. 

    Note that these vertex hunting algorithms are only meant for separable point clouds, as the simplex vertices they produce are designed to belong to the convex hull of the point cloud. For more implementation details of SVS and SP, we refer the reader to Section A of \cite{ke2022using}. 
    }   
    \end{remark}
    
    \begin{remark}[$\ell_q$-sparsity]\label{rem: lq-sparsity}{\rm To our knowledge, our work is the first to consider the $\ell_q$-sparsity assumption \eqref{eq: lq-sparsity-def} in the topic modeling context, although similar assumptions have been adopted in other statistical settings such as sparse PCA and sparse covariance estimation (see for example \cite{ma2013sparse} and \cite{cai2012optimal}). \eqref{eq: lq-sparsity-def} imposes an assumption on the decay rate of the ordered entries of the columns of $A$, but does not restrict how small (or large) the smallest (or largest, assuming $s \geq 1$) entries of $A$'s columns can be. Thus, our theoretical results are valid even in the presence of severe word frequency heterogeneity. 

    Note that if columns $A_{*k}$ has $s$ nonzero entries, then we always have $\max_{j \in [p]} j A_{(j)k}^q \leq s$. However, in light of \eqref{eq:Bing-sparsity} where we observe that $\|A\|_0 \geq \tilde{p}$, there exists at least one column of $A$ with at least $\lfloor\tilde{p}/K\rfloor$ nonzero entries, and so $s$ in \eqref{eq: lq-sparsity-def} cannot be much smaller than $p$ if we impose hard sparsity ($q = 0$) on \textit{all} columns of $A$. Therefore, the $\ell_q$-sparsity assumption \eqref{eq: lq-sparsity-def} gives us more flexibility as it allows for the possibility that most entries of $A$ are small but nonzero. When $q \approx 0$, we can approximate the assumption of hard sparsity on all columns of $A$, whereas when $q$ is close to 1, then \eqref{eq: lq-sparsity-def} with $s = O(1)$ corresponds to Zipf's law, which is the empirical observation that word frequency in text data is often inversely proportional to word rank. 

    The restriction that $q \in (0,1)$ is primarily due to the fact that we use the $\ell_1$ loss $\|\hat{A}-A\|_1$. Since the columns of $A$ sum up to 1, the columns of $A$ already satisfy $\ell_q$-sparsity with $q = s = 1$, but this alone is not sufficient to control the error term $\|A_{J^c}\|_1$ resulting from our thresholding step. }
    \end{remark}

    We are now ready to discuss our main theoretical results. Let $\hat{\Xi} = [\hat{\xi}_1, \dots, \hat{\xi}_K] \in \mathbb{R}^{|J| \times K}$ contains the first $K$ eigenvectors of $G_{JJ}$ where $G$ is defined as in \eqref{def-of-G}. Recall its oracle counterpart $\Xi = [\xi_1, \dots, \xi_K] \in \mathbb{R}^{|J|\times K}$ which contains the first $K$ left singular vectors of $[D_0]_{J*}$. Let $\{\Xi_j: j \in J\}$ and $\{\hat{\Xi}_j: j \in J\}$ denote the rows of $\Xi$ and $\hat{\Xi}$ respectively. 

    \begin{lemma}[Row-wise error bounds for $\hat{\Xi}$] For all $j \in [p]$, let $h_j := \sum_{k=1}^K A_{jk}$. With probability $1 - o(p_n^{-1})$, there exist $\omega \in \{\pm 1\}$ and a $(K-1)\times (K-1)$ orthonormal matrix $\Omega^* $ such that, if we define $\Omega := \diag(\omega, \Omega^*) \in \mathbb{R}^{K \times K}$, we have 
    \begin{equation}\label{result:row-wise-eigenvector}
        \|\Omega \hat{\Xi}_j - \Xi_j\|_2 \leq  C\sqrt{\frac{h_j \log p_n}{nN}} \quad \text{for all } j \in J 
    \end{equation}
    \end{lemma}
    The proof can be found in Lemma \ref{lem:D8} and is an application of the well-known Davis-Kahan theorem (more specifically, we need to use the row-wise perturbation version of the theorem as proven in Lemma F.1 of \cite{ke2022using}). We note here that the bound \eqref{result:row-wise-eigenvector} depends on $p$ only via the log term, and the $h_j$'s, which indicates how frequently one may encounter word $j$ in the corpus, determines the magnitude of the bound \eqref{result:row-wise-eigenvector}. 

    As a consequence of the above lemma, one can provide error bounds for the point cloud obtained from our procedure. Again, recall that $\{r_j: j \in J\}$ is the oracle point cloud from Step 3 of Definition \ref{def: oracle-procedure}, and $\{\hat{r}_j: j \in J\}$ is the point cloud from Step 4 of Definition \ref{actual-procedure}. 
    \begin{corollary}[Error bounds for the point cloud]\label{corollary:pt-cloud-error} With probability $1 - o(p_n^{-1})$, there exists a $(K-1) \times (K-1)$ orthonormal matrix $\Omega^*$ such that 
    \begin{equation}\label{eq:pt-cloud-error}
        \max_{j \in J} \|\Omega^* \hat{r}_j - r_j\|_2 \leq C\left(\frac{\log p_n}{nN}\right)^{1/4}
    \end{equation}
        \end{corollary}
    The proof can be found in Lemma \ref{lem: D.9}. To elaborate further on \eqref{eq:pt-cloud-error}, we can show that with high probability, 
    \begin{equation}\label{eq: pt-cloud-error2}
        \|\Omega^* \hat{r}_j - r_j\|_2 \leq C\sqrt{\frac{\log p_n}{h_jnN}} \quad \text{ for all } j \in J
    \end{equation}
    Observe that unlike \eqref{result:row-wise-eigenvector}, the bound \eqref{eq: pt-cloud-error2} is inversely proportional to $\sqrt{h_j}$ due to the fact that the point cloud is obtained from the elementwise division step. Since we do not restrict how small $\min_{1 \leq j \leq p}h_j$ can be, the error bound \eqref{eq: pt-cloud-error2} may be uncontrollable without appropriate thresholding of infrequent words. However, with the choice of $J$ as in \eqref{J-def}, one can show $\min_{j \in J} h_j \geq c\sqrt{\frac{\log p_n}{nN}}$ with high probability, which when combined with \eqref{eq: pt-cloud-error2} leads to \eqref{eq:pt-cloud-error}. 

    From \eqref{eq: pt-cloud-error2}, we can also obtain bounds on how much the probabilistic weights $\{\hat{\pi}_j: j \in J\}$ from Step 5 of Definition \ref{actual-procedure} deviate from the oracle weights $\{\pi_j: j \in J\}$ from Step 4 of Definition \ref{def: oracle-procedure}). The proof of the following corollary can be found in Lemma \ref{lem: pi-error} of the appendix. 
    \begin{corollary}[Error bounds for $\hat{\Pi}$] With probability $1 - o(p_n^{-1})$, 
    \begin{equation}\label{eq: pi-error}
        \max_{j \in J}\|\hat{\pi}_j - \pi_j\|_1 \leq C\left(\frac{\log p_n}{nN}\right)^{1/4}
    \end{equation}
    \end{corollary}
    Note that while $\{\Xi_j: j \in J\}$ and $\{r_j: j \in J\}$ can be recovered only up to an orthonormal transformation $\Omega^*$, the bound \eqref{eq: pi-error} does not depend on $\Omega^*$. We also note that the bounds \eqref{result:row-wise-eigenvector}, \eqref{eq:pt-cloud-error} and \eqref{eq: pi-error} are derived without using the $\ell_q$-sparsity assumption (Assumption \ref{ass:sparsity}).  

    The next theorem is our main result, which provides the error rate for estimating $A$ using the $\ell_1$ loss $\|\hat{A}-A\|_1$. Recall the definition of $\tilde{A}$ in Lemma \ref{lem: tilde-A}. 

    \begin{theorem}[Estimation error for $\hat{A}$]\label{theorem: main-result} Suppose Assumptions 1-4 are satisfied. Then with probability $1- o(p_n^{-1})$, 
    \begin{equation}\label{eq: est-error}
        \|\hat{A}_{J*} - \tilde{A}_{J*}\|_1 \leq C\left(\frac{\log p_n}{nN}\right)^{1/4}
    \end{equation}
    If we further assume the $\ell_q$-sparsity assumption (Assumption \ref{ass:sparsity}) and $s\left(\frac{\log p_n}{nN}\right)^{\frac{1-q}{2}} = o(1)$, we also have with probability $1- o(p_n^{-1})$,
    \begin{equation}\label{eq: approx-error}
        \|\tilde{A}_{J*} - A_{J*}\|_1 = \|A_{J^c*}\|_1 \leq Cs\left(\frac{\log p_n}{nN}\right)^{\frac{1-q}{2}}
    \end{equation}
    and therefore with probability $1- o(p_n^{-1})$, 
    \begin{equation}\label{eq: final-est-error}
        \|\hat{A}-A\|_1 \leq C\left[\left(\frac{\log p_n}{nN}\right)^{1/4} + s\left(\frac{\log p_n}{nN}\right)^{\frac{1-q}{2}}\right] 
    \end{equation}
    for some constant $C$ that may depend on $K$ and $q$. 
    \end{theorem}
    The proof of the above statements can be found in Appendix \ref{appendix: E}. 
    \begin{remark}
        {\rm The bounds \eqref{eq: est-error} and \eqref{eq: approx-error} can be interpreted as the estimation error and the approximation error respectively for using an estimator of $A$ whose row support is contained in the set $J$. Note that the approximation error \eqref{eq: approx-error} is smaller if $q$ is closer to $0$; here we assume $s$ does not grow too quickly relative to $nN$. In the most favorable setting where $s = O(1)$ and $0 < q < 1/2$ (strong sparsity regime), the aggregate error \eqref{eq: final-est-error} is of the order $\left(\frac{\log p_n}{nN}\right)^{1/4}$, which clearly converges to zero as $nN \to \infty$. On the other hand, if $s \geq 1$ and $1/2 < q < 1$ (weak sparsity regime), the bound \eqref{eq: final-est-error} is dominated by the term $s \left(\frac{\log p_n}{nN}\right)^{\frac{1-q}{2}}$.}
    \end{remark}
    \begin{remark}{\rm
        We note again that the bound \eqref{eq: final-est-error}, which does not depend on $p$ except for log terms, is valid for all parameter regimes and in particular for the high-dimensional setting where $p \gg \max(n, N)$. This justifies the use of our method for many text datasets where the number of unique words observed across all documents is extremely large. Also, the bound \eqref{eq: final-est-error} does not depend on $\max_{j \in [p]}h_j$ or $\min_{j \in [p]}h_j$ and is thus completely unaffected by variations in word frequencies. In these regards, our result improves upon the theoretical guarantees presented in prior works such as \cite{ke2022using}, \cite{bing2020fast}, \cite{arora2012learning} and \cite{wu2022sparse}.}  
    \end{remark}
\subsection{Relaxation of the separability condition}\label{sec: relax}

    Our main result (Theorem \ref{theorem: main-result}) may also hold under alternative identifiability assumptions on $A$ if we use a suitable vertex hunting procedure that is effective even for non-separable point clouds. Recall $v_1^*, \dots, v_K^*$ are the simplex vertices from the oracle point cloud $\{r_j: j \in J\}$ in Definition \ref{def: oracle-procedure} and $\hat{v}_1^*, \dots, \hat{v}_K^*$ are the estimated vertices based on the point cloud $\{\hat{r}_j: j \in J\}$ in Definition \ref{actual-procedure}. The assumptions we made concerning separability and vertex hunting efficiency, namely Assumptions \ref{ass: separability} and \ref{ass:VH}, 
    are only useful in our analysis insofar as they allow the following bound to hold with high probability:
    \begin{equation}\label{eq: 26}
        \max_{k \in [K]}\|\hat{v}_k^* - v_k^*\|_2 \leq \max_{j\in J}\|\hat{r}_j - r_j\|_2
    \end{equation}

            \begin{figure}[h!]
    \centering
    \includegraphics[width=0.7\textwidth]{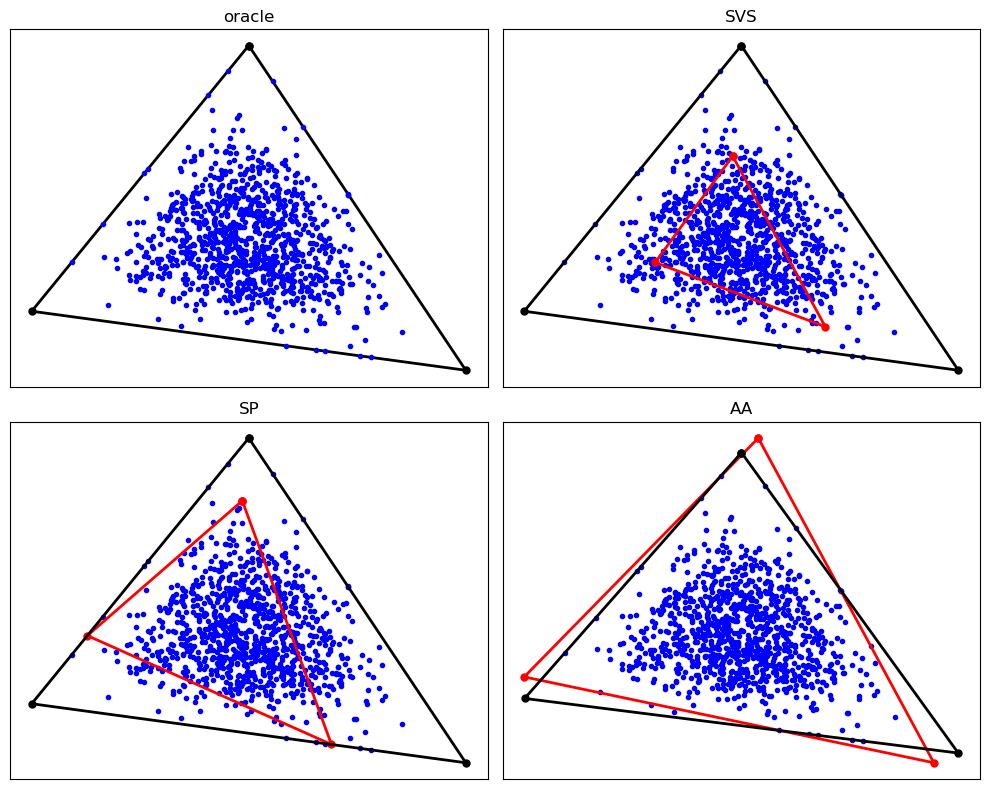}
    \caption{Top left: a non-separable point cloud (blue) contained in a simplex (black) with 3 vertices. Top right: Estimated vertices from SVS (red). Bottom left: Estimated vertices from SP (red). Bottom right: Estimated vertices from AA (red).}
    \label{fig:toy_example}
\end{figure}

However, this bound may also hold if we adopt the identifiability assumption and the Archetype Analysis (AA) vertex hunting procedure proposed in \cite{javadi2020nonnegative}. Figure \ref{fig:toy_example} provides an example of a non-separable point cloud where AA recovers the simplex vertices much more effectively than SP and SVS, which only search for possible vertices within the point cloud itself or its convex hull. Appendix \ref{appendix: Javadi} summarizes important results from \cite{javadi2020nonnegative} that are relevant to our paper. In our estimation procedure for $A$, once we obtain the matrix $\hat{R}$ whose rows are $\{\hat{r}_j: j \in J\}$ from Step 3 of Definition \ref{actual-procedure}, the estimated simplex vertices can be obtained via AA by solving the following minimization problem: 

    \begin{equation}\label{javadi-VH-main-text}
    \text{minimize } \mathscr{D}(V; \hat{R}) \text{ over $V$ s.t. } \mathscr{D}(\hat{r}_j; V) \leq \delta^2 \text{ for all } j \in J
    \end{equation}

Here the rows of $V$ represent the simplex vertices; see Appendix \ref{appendix: Javadi} for the definition of the distance function $\mathscr{D}(\cdot, \cdot)$. The main theoretical result of \cite{javadi2020nonnegative} (Theorem \ref{theorem: F2}) is that the AA algorithm is robust to noise in the point cloud under certain conditions. In particular, if we replace Assumptions \ref{ass: separability} and \ref{ass:VH} by the following assumptions: 
    \begin{enumerate}[(i)]
        \item The matrix $R$ from Step 2 of the oracle procedure (Definition \ref{def: oracle-procedure})
        satisfies $\alpha$-uniqueness for some absolute constant $\alpha > 0$. Here, $\alpha$-uniqueness (described in Definition \ref{def: alpha-uniqueness}) is an identifiability assumption on the simplex vertices that is more general than separability. 
        \item The convex hull of the rows of $R$ contains a $(K-1)$-dimensional ball of radius $\mu > 0$
        \item The vertex hunting procedure $\mathcal{V}(\cdot)$ is defined by \eqref{javadi-VH-main-text} with $\delta \asymp \left(\frac{\log p_n}{nN}\right)^{1/4}$. This value of $\delta$ is chosen based on Corollary \ref{corollary:pt-cloud-error} and Theorem \ref{theorem: F2}. 
    \end{enumerate}
    then, in light of Theorem \ref{theorem: F2}, \eqref{eq: 26} continues to hold and our main result, Theorem \ref{theorem: main-result}, remains valid. Alternatively, if we do not wish to use the $\alpha$-uniqueness condition for identifiability, we can also assume that the distance from the oracle simplex vertices $\{v_1^*, \dots, v_K^*\}$ to the convex hull of the oracle point cloud $\{r_j: j \in J\}$ is not larger than $\delta$. In light of Theorem \ref{theorem: F3}, this assumption can also be used to obtain \eqref{eq: 26}.  

    Beside from \cite{javadi2020nonnegative}, \cite{ge2015intersecting} also discusses an alternative identifiability assumption called \textit{subset separability}. This notion can be illustrated by the point cloud in Figure \ref{fig:toy_example} (top left), with $K = 3$. The point cloud (in blue) is contained in a triangle but is not separable as none of the triangle's vertices belongs to the point cloud. However, each edge of the triangle contains several blue points and thus can clearly be identified from the point cloud. The vertices can then be identified by taking intersections of the edges. \cite{ge2015intersecting} also provides a vertex hunting procedure which, under subset separability and additional regularity assumptions, can also be shown to be robust to noise in the point cloud, in the sense of \eqref{eq: 26}. 

        In terms of computation, \cite{javadi2020nonnegative} describes two algorithms to solve the following Lagrangian variant of \eqref{javadi-VH-main-text}: 
    \begin{equation}\label{javadi-Lagrangian-objective}
        \hat{V}_\lambda = \argmin_{V}[\mathscr{D}(\hat{R}; V) + \lambda \mathscr{D}(V;\hat{R})]
    \end{equation}
     Note that the objective function in \eqref{javadi-Lagrangian-objective} is non-convex and thus may have multiple minima. While AA may significantly reduce statistical error in the vertex hunting step when separability is not applicable, the trade-off is that its computational cost may be higher than that of the SP algorithm for separable point clouds.

    \subsection{Estimation of \texorpdfstring{$K$}{K}}\label{sec: est-of-K}
    Our discussion so far assumes $K$ is known. When $K$ needs to be estimated, it is natural to examine the spectrum of any matrix that should be of rank $K$ under the pLSI model.  

    Recall the definition of $G$ in \eqref{def-of-G}, and define $G_0 := \left(1-\frac{1}{N}\right)D_0D_0^T$. From Lemma \ref{lem: D3-GG0}, with probability $1 - o(p_n^{-1})$ we have (here $C^*$ is a numerical constant not dependent on unobserved constants but may depend on the choice of $\alpha$)
    \begin{equation}\label{eq:31}
        \|(G-G_0)_{JJ}\|_\text{op} \leq C^*K\sqrt{K} \sqrt{\frac{n\log p_n}{N}}
    \end{equation}
    Furthermore one can show $[G_0]_{JJ}$ has rank $K$ with high probability. By a simple application of Weyl's inequality, we then obtain the estimator \eqref{K-hat-def} for $K$. 

    \begin{lemma}\label{lem: K-hat}
            Let $g_n$ be a quantity satisfying 
    \begin{equation}\label{eq: g_n-def-main-text}
        c\sqrt{\frac{nN}{\log p_n}}\geq g_n \geq C^*K\sqrt{K}
    \end{equation}
    where $C^*$ in \eqref{eq: g_n-def-main-text} is the constant from \eqref{eq:31} and $c$ is another constant that may depend on $K$. If
    \begin{equation}\label{K-hat-def}
        \hat{K} := \max\left\{k:\lambda_k(G_{JJ}) > g_n\sqrt{\frac{n\log p_n}{N}}\right\}
    \end{equation}
    then $\hat{K} = K$ with probability $1 - o(p_n^{-1})$. 
    \end{lemma}
    The proof can be found in Corollary \ref{cor:D4} of the appendix. In \eqref{eq: g_n-def-main-text}, the quantity $g_n$ needs to be chosen to override the term $C^* K\sqrt{K}$ but cannot converge to $+\infty$ too quickly. Without any prior information on $K$, one can choose $g_n$ to be a quantity that slowly converges to $+\infty$, such as $g_n = 8\log p_n$. If one has prior knowledge on an upper bound for $K$ (for example if $K \leq 30$), the quantity $g_n$ can be determined more specifically. 

    The estimator \eqref{K-hat-def} is based on the bound \eqref{eq:31}, which depends on $K$ and so we need to assume $g_n \geq C^*K\sqrt{K}$. However, one can also show that with probability at least $1 - \frac{1}{n+p}$, 
    \begin{equation}\label{klopp-K-hat-bound}
        \|(D-D_0)_{J*}\|_\op \leq 4\sqrt{\frac{n \log (n+p)}{N}} 
    \end{equation}
    (see Lemma 4 of \cite{klopp2021assigning}). This bound does not depend on $K$. Under similar assumptions on $\sigma_K(A)$ and $\sigma_K(W)$, we can consider the following estimator
    \begin{equation}\label{klopp-K-hat}
        \hat{K}' := \max\left\{k: \sigma_k(D_{J*}) > 4\sqrt{\frac{n \log(p+n)}{N}}\right\}
    \end{equation}
    and also show that, based on \eqref{klopp-K-hat-bound}, $\hat{K}' = K$ with high probability. The advantage of \eqref{K-hat-def} over \eqref{klopp-K-hat} is computational: both Step 2 of Definition \ref{actual-procedure} and \eqref{K-hat-def} use the eigendecomposition of $G_{JJ}$, whereas \eqref{klopp-K-hat} requires us to additionally perform SVD on $D_{J*}$.

    \begin{figure}[h!]
    \centering
    \includegraphics[width=0.9\textwidth]{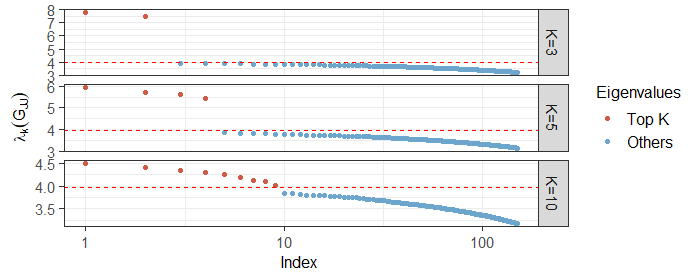}
    \caption{Scree plots of the eigenvalues of $G_{JJ}$ for three synthetic datasets, with $K \in \{3, 5, 10\}$, $n=N = 500$ and $p=5,000$. The x-axis is log-scaled. The red dots represent the largest $K$ eigenvalues (excluding the largest one), while the blue dots represent all other eigenvalues.}
    \label{fig:Estimation K}
\end{figure}

    There are many choices of the quantity $g_n$ that may satisfy \eqref{eq: g_n-def-main-text} when $nN$ is sufficiently large. In practice, the estimation of $K$ may be sensitive to the choice of the eigenvalue cutoff, and moreover real datasets may not always adhere to our assumptions. As Lemma \ref{lem: K-hat} suggests the spectrum of $G_{JJ}$ is useful for estimating $K$, we note that it is often possible to determine the eigenvalue cutoff by inspecting the scree plot of $G_{JJ}$'s eigenvalues. Figure \ref{fig:Estimation K} displays the scree plots for several synthetic datasets with different values of $K$. In some situations, the top $K$ eigenvalues of $G_{JJ}$ are separated from the other eigenvalues by a discernible gap, thus helping one to visually determine $K$. When such a gap is unavailable, one can use the Kneedle algorithm \citep{satopaa2011finding} to find the point of maximum curvature of the scree plot; this is a common technique to determine the number of principal components in principal component analysis.

\section{Experiments with synthetic data}\label{sec:experiments}
In this section, we assess the empirical performance of our estimator through a series of synthetic experiments\footnote{The code for our method and all the experiments presented in this section can be found on Github at the following link: \url{https://github.com/yatingliu2548/topic-modeling}}. The controlled environment provided by these experiments allows us to better understand the behavior of our method in different parameter regimes. 

Throughout this section, we benchmark our estimator's performance against the following well-established methods: (a) Latent Dirichlet Allocation \citep{blei2003latent}; (b) the anchor word recovery (AWR) approach in \cite{arora2012learning}, a procedure based on the non-negative factorization of the second-order moment $DD^T$; (c) the Topic-SCORE procedure in \cite{ke2022using}; and (d) the Sparse Topic Model solver proposed in \cite{bing2020optimal}. We note the following regarding the procedure in \cite{bing2020optimal}:
\begin{itemize}
    \item This procedure removes infrequently occurring words in the same manner as ours, but with the threshold $\alpha \sqrt{\frac{\log p_n}{nN}}$ in \eqref{J-def} replaced by $\frac{7\log p_n}{nN}$. This threshold is lower than ours if $\frac{\log p_n}{nN}$ is sufficiently small. In practice, however, the constant 7 used in their threshold is quite large and thus leads to excessive thresholding in some of our datasets, especially when the word frequencies decay according to Zipf's law. 
    \item This procedure requires a list of anchor words for each topic $k \in [K]$ as input, rather than just the number of topics $K$. We therefore need to estimate a partition of anchor words using a special procedure which is included in their original implementation. Clearly, whether the anchor words are estimated and partitioned correctly has an impact on the overall estimation of $A$. 
\end{itemize}
We therefore caution the reader that these factors put the Sparse Topic Model solver of \cite{bing2020optimal} at a comparative disadvantage in our experiments.\\


\xhdr{Data generation mechanism}
For simplicity, we ensure all documents are of the same length $N$. For each experiment, we create a document-to-topic matrix $W \in \mathbb{R}^{K \times n}$ by independently drawing the columns $W_{*i} \in \mathbb{R}^K, i = 1, \dots, n$ from the Dirichlet distribution with parameter $\boldsymbol{\alpha}_{W} = \mathbf{1}_K$.  We generate the matrix $A \in \mathbb{R}^{p\times K}$ either without anchor words or with 5 anchor words per topic, in which case whenever word $j$ is an anchor word for topic $k$, we set $A_{jk} = \delta_\text{anchor}$ where $\delta_\text{anchor} \in \{0.0001, 0.001, 0.01\}$. In order to mimic the behavior of real text data, the entries of column $k$ of $A$ corresponding to non-anchor words are then chosen such that they decay according to Zipf's law. This means for each column $k$ of $A$, we ensure that the frequency $f_{(j)}$ of the $j^\text{th}$ most frequent non-anchor word follows the pattern
\begin{equation}\label{eq:zipf}
    f_{(j)} \propto \frac{1}{( j + b_{\text{zipf}})^{a_{\text{zipf}}}}
\end{equation} 
where $a_\text{zipf} = 1, b_\text{zipf} = 2.7$. Each column of $A$ is subsequently normalized to unit $\ell_1$-norm. The pattern \eqref{eq:zipf} has indeed been empirically shown to hold approximatively for word frequencies in real datasets; see \cite{zipf2013psycho} and \cite{piantadosi2014zipf}. Figure \ref{fig:zipfs} in Appendix \ref{appendix:experiment} illustrates the distribution of word frequencies generated under our data generation mechanism.  

Having specified both $A$ and $W$, the observation matrix $D$ is then generated according to the pLSI model described in Section \ref{section: pLSI}. We fit our method and the four benchmarks while varying the values of $n, p, N,$ and $K$. In all of our experiments, unless otherwise specified, the constant $\alpha$ in the threshold  \eqref{J-def} is fixed at $\alpha=0.005$. We evaluate the estimation error of all methods relative to the true underlying $A$ by computing the $\ell_1$ loss per topic 
   $$ \mathcal{L}_1(\hat{A}, A) = \min_{\Pi \in \mathcal{P}} \frac{1}{K}\| \hat{A}\Pi - A \|_1 $$
   where $\mathcal{P}$ denotes the set of all $K \times K$ permutation matrices. \\

\xhdr{Varying $(p, N, K)$} We first provide a snapshot of our method's relative performance in different parameter regimes by fixing $n = 500$ and varying $(p,N, K)$. Here we specify 5 anchor words per topic and set the anchor word frequency to $\delta_\text{anchor} = 10^{-3}$. The median $ \mathcal{L}_1(\hat{A}, A)$-errors over 50 trials are plotted in Figure \ref{fig:l1-loss-all}. As Figure \ref{fig:l1-loss-all} shows, our method (in blue) outperforms all other methods in most parameter regimes considered here.
Interestingly, 
the estimation errors of AWR and LDA often appear constant as a function of document length $N$. As $N$ increases, the errors from both Topic-SCORE and our method display a clearer pattern of consistency relative to AWR and LDA; this observation is also made by \cite{ke2022using} in a similar experimental setup. However, our method's errors decay to zero much faster than all other benchmarks when the vocabulary size is large ($p \in \{5000, 10000\}$). 

\begin{figure}[ht!]
    \centering
    \includegraphics[width=\textwidth]{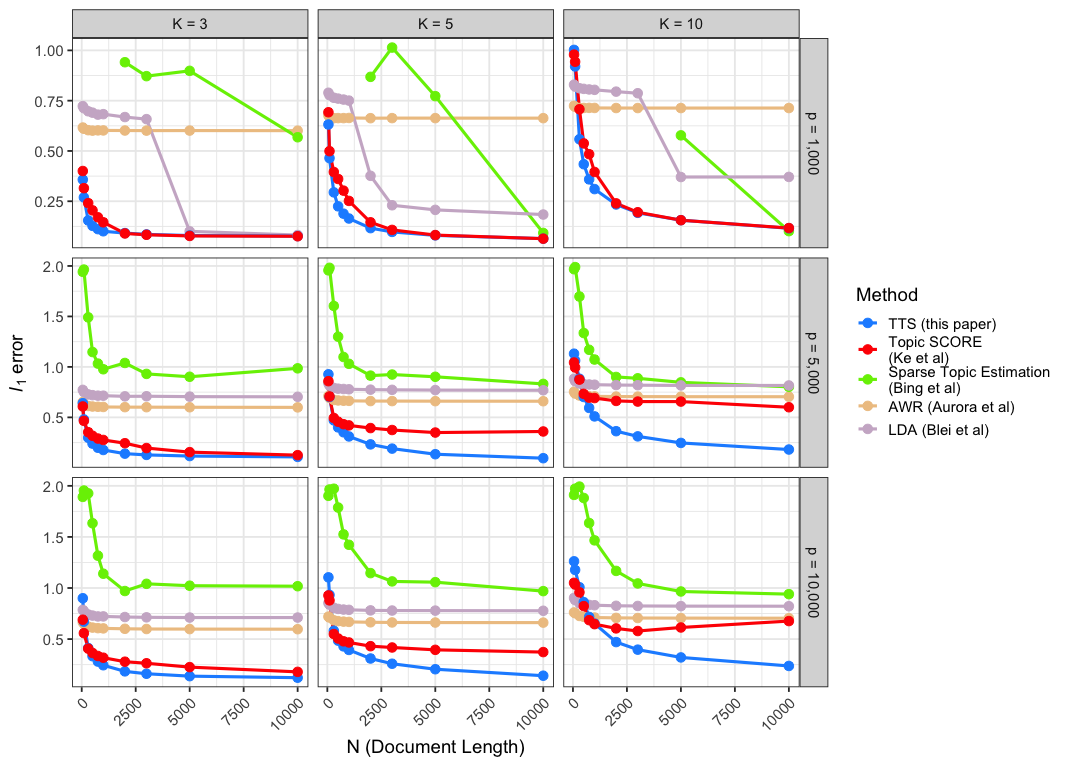} \caption{Median $\mathcal{L}_1(\hat{A}, A)$ errors for all methods based on 50 independent trials. Here the number of documents is fixed ($n = 500$). In each panel, the errors are plotted as a function of document length $N$ (log-scaled on the x-axis). The panels display results for different values of $(p,K)$, as specified by row and column labels. }
    \label{fig:l1-loss-all}
\end{figure}
We note that in these experiments, the approach proposed by \cite{bing2020optimal} does not perform very well. In particular, for small $p$ and small $N$, the number of topics returned by this method is smaller than the expected number of topics $K$, which prevents us from comparing its results with all four other methods. On inspection, we find that this is due to over-thresholding of the vocabulary, which leaves too few words to reliably estimate the matrix $A$. To provide a fair comparison with \cite{bing2020optimal}, we also compare all five methods using the data generation mechanism proposed 
in \cite{bing2020optimal}. This means that the non-anchor entries of each column of $A$ no longer display the Zipf's law pattern \eqref{eq:zipf}, but instead are generated from a Uniform distribution. We note that this data generation mechanism ensures all the non-anchor words for each topic are of roughly equal frequency and is thus also favorable to Topic-SCORE \citep{ke2022using}, which assumes $\min_{j\in [p]} h_j \geq c\bar{h}$ where $\bar{h}:= \frac{1}{p} \sum_{j=1}^p h_j$. The results are displayed in Figure~\ref{fig:enter-label} of Appendix \ref{appendix:experiment}. Under this uniform data generation mechanism, our method (with $\alpha = 0.005$) 
displays identical performance relative to Topic-SCORE, and both SCORE-based methods still perform well relative to other benchmarks in most parameter regimes. 
As expected, we also find that fewer words are removed by thresholding, in comparison with the Zipf's law setting where our $\ell_q$-sparsity assumption \eqref{eq: lq-sparsity-def} is more likely to hold with small $s$ and many more words occur infrequently. These experiments empirically suggest that 1) TTS improves upon the performance of Topic-SCORE when the columns of $A$ exhibit a Zipf's law (or $\ell_q$-sparsity) decay pattern, and 2) our procedure's performance remains reasonable and is similar to that of Topic-SCORE when the $\ell_q$-sparsity assumption \eqref{eq: lq-sparsity-def} is violated. \\

\xhdr{Varying the number of documents $n$} We now focus on the effect of varying $n$ on the estimation error.  Fixing this time $N=500$ and $p=10,000$, the $\mathcal{L}_1(\hat{A}, A)$-errors are presented in Figure~\ref{fig:function-of-n} with $K=5$ and $K=10$. Our method (in blue) consistently outperforms other methods and also displays a clear trend of consistency as $n$ increases. When $K$ increases, the estimation problem becomes more difficult due to the larger number of parameters, and so more documents are needed to achieve a reasonable performance. Nonetheless, our method still performs well when $K = 10$ and $n$ is reasonably large, whereas the error from Topic-SCORE decays to zero very slowly with this larger value of $K$.\\


\begin{figure}[ht!]
    \centering
    \includegraphics[width=\textwidth]{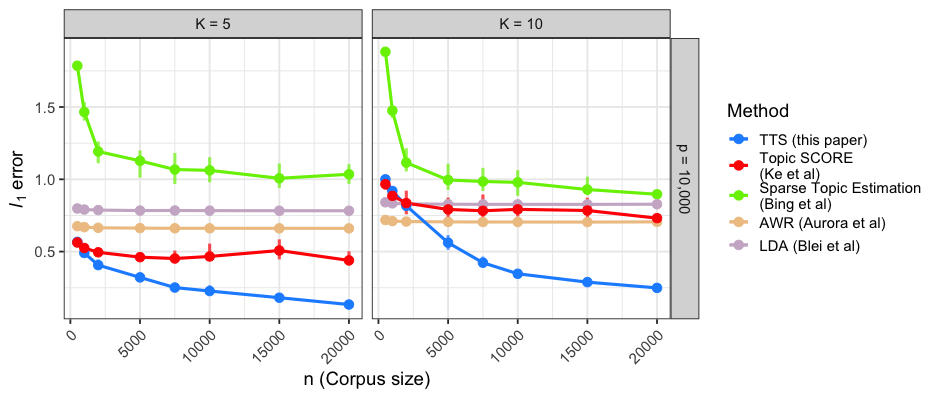}
    \caption{$\mathcal{L}_1(\hat{A}, A)$-errors from all methods as a function of $n$, for $K \in \{5, 10\}$ with $p$ and $N$ fixed. Vertical error bars centered about the median errors indicate the errors' interquartile ranges computed based on $50$ independent trials.}
    \label{fig:function-of-n}
\end{figure}

\xhdr{Varying the dictionary size $p$} Figure \ref{fig:function-of-p} shows how the $\mathcal{L}_1(\hat{A}, A)$-errors vary as the vocabulary size $p$ increases, with $N = 500$, $K = 5$ and $n \in \{1000, 5000, 10000\}$. We do not include the errors from the procedure in \cite{bing2020optimal} as they are higher than those of LDA. As expected, the errors for all methods increase with the dictionary size $p$. However, our method mostly outperforms the other benchmarks, even in some high-dimensional parameter regimes where $p > \max(n, N)$. The performance of Topic-SCORE only converges to ours when $p$ is too large relative to $n$, a setting which is challenging for all methods.

Additionally, our method also outperforms most other benchmarks in terms of computational runtime when $p$ is large. We provide in Figure~\ref{fig:time} of Appendix \ref{appendix:experiment} a visualization of how the runtimes for all methods scale with $p$. Our method's runtime is similar to that of AWR and is consistently better than that of Topic-SCORE, primarily due to our thresholding of infrequent words before performing eigendecomposition.\\ 

\begin{figure}[ht!]
\begin{subfigure}[t]{0.48\textwidth}
      \centering
    \includegraphics[width=\textwidth]{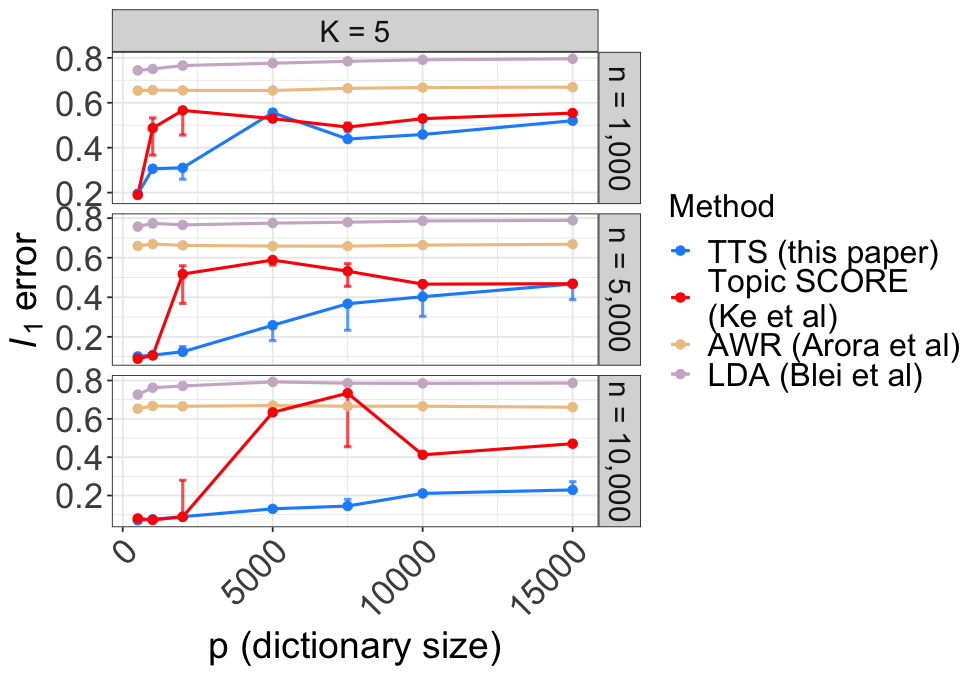}
    \caption{$\mathcal{L}_1(\hat{A}, A)$-errors as a function of $p$, with $K =5$ and $N = 500$. Results are obtained based on $15$ independent trials.}
    \label{fig:function-of-p}  
\end{subfigure}
\hspace{0.5cm}
\begin{subfigure}[t]{0.48\textwidth}
    \centering
    \includegraphics[width=\textwidth, height=5.2cm]{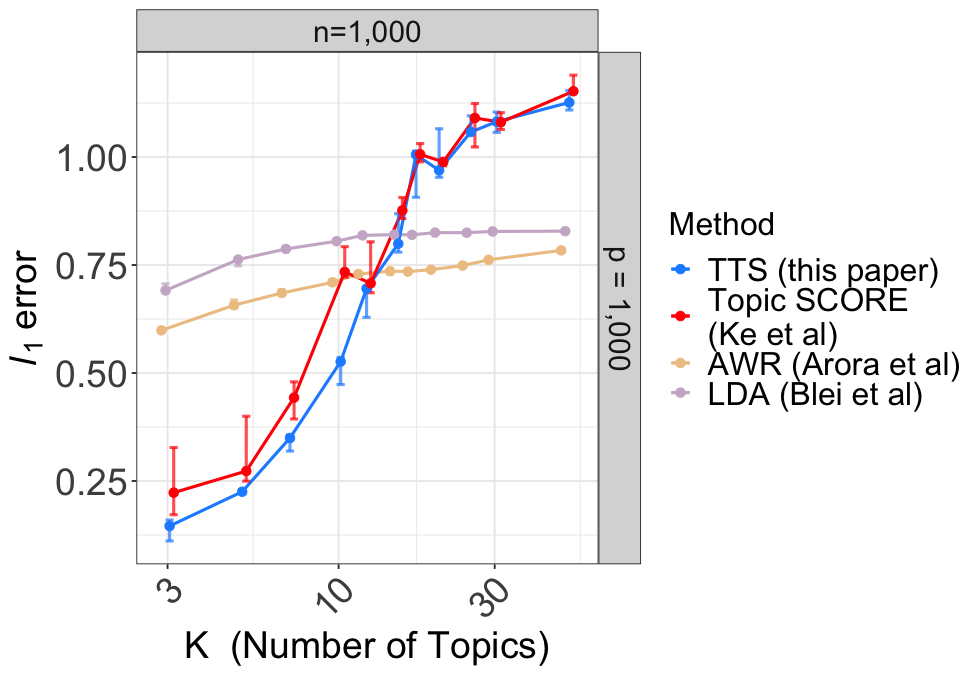}
    \caption{$\mathcal{L}_1(\hat{A}, A)$-errors as a function of $K$, with $n= p =10^3$ and $N=500$. Results are obtained based on $50$ independent trials.}
    \label{fig:function_of_K}
\end{subfigure}
\caption{$\mathcal{L}_1(\hat{A}, A)$-errors as a function of the dictionary size $p$ (left) and the number of topics $K$ (right). Vertical bars around median errors indicate interquartile ranges.}
\end{figure}

 \xhdr{Varying the number of topics $K$} Figure \ref{fig:function_of_K} shows how the $\mathcal{L}_1(\hat{A}, A)$-errors vary as $K$ increases, with $n = p = 1000$ and $N = 500$. The main observation here is that LDA and AWR may be preferable to our method if $K$ is \textit{a priori} known to be large while the dataset we possess is relatively small. As Figure \ref{fig:function_of_K} illustrates, the SCORE-based methods perform worse than LDA and AWR when $K > 15$, but this is because the number of documents is quite small in this experiment ($n = 1000$). If $n$ and $N$ are large enough, one can expect our method to accommodate a larger number of topics; see Figure \ref{fig:function-of-n} for an illustration.\\

\xhdr{Relaxation of the separability assumption} Section \ref{sec: relax} suggests that the vertex hunting algorithm from \cite{javadi2020nonnegative} may reduce the vertex hunting error in some situations when separability fails to hold. Figure \ref{fig:VHMethod} compares the overall $\mathcal{L}_1(\hat{A}, A)$-errors as a function of $n$ when we use Successive Projection (SP), Sketched Vertex Analysis (SVS) and Archetype Analysis (AA) in the vertex hunting step of TTS. As expected, when there are no anchor words, using the AA algorithm rather than SP/SVS can significantly improve the estimation of $\hat{A}$, especially when $K$ is large. Again, this is because SP and SVS are not designed for non-separable point clouds and also perform better with small $K$. In fact, the AA algorithm also often works well under separability, since the $\alpha$-uniqueness condition in \cite{javadi2020nonnegative} is satisfied. The main trade-off for this stronger statistical performance is the computational cost of solving the non-convex optimization problem required by AA. Nonetheless, the fact that our method accommodates non-separable datasets makes TTS more widely applicable compared to methods based on anchor words identification, such as those proposed in \cite{bing2020optimal} and \cite{arora2012learning}. \\

\begin{figure}
    \centering
\includegraphics[width=1.0\textwidth]{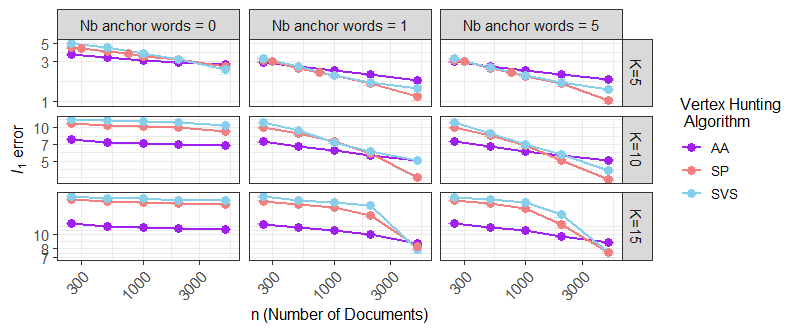}
    \caption{$\mathcal{L}_1(\hat{A}, A)$-errors as a function of $n$ when we use three different vertex hunting algorithms in the vertex hunting step of TTS. Here, $p  = 10^4$ and $N = 500$ are fixed, and $K \in \{5,10, 15\}$. The number of topics per document is either 0, 1 or 5. Results are averaged over 50 independent experiments.}
    \label{fig:VHMethod}
\end{figure}


\xhdr{The importance of appropriate thresholding} Figure \ref{fig:alpha_vs_l1} shows how the $\mathcal{L}_1(\hat{A}, A)$-error varies as the threshold level in \eqref{J-def} increases from zero, and Figure \ref{fig:alpha_vs_percent} shows the corresponding percentage of words removed. For this dataset, the performance of our method when $\alpha = 0$ (no thresholding) is not too different from Topic-SCORE. As the threshold level increases, infrequent words that contribute noise to the point cloud are removed, thus leading to an improvement in the estimation of $A_{J*}$. However, an excessively high threshold means we set too many rows of $A$ to zero, and so the error from estimating $A_{J^c*}$ becomes higher. This explains the pattern observed in Figure \ref{fig:alpha_vs_l1}, which demonstrates the importance of choosing a balanced threshold in our procedure. 

\begin{figure}[ht!]
    \centering
         \begin{subfigure}[t]{0.5\textwidth}
\centering
    \includegraphics[width=\textwidth]{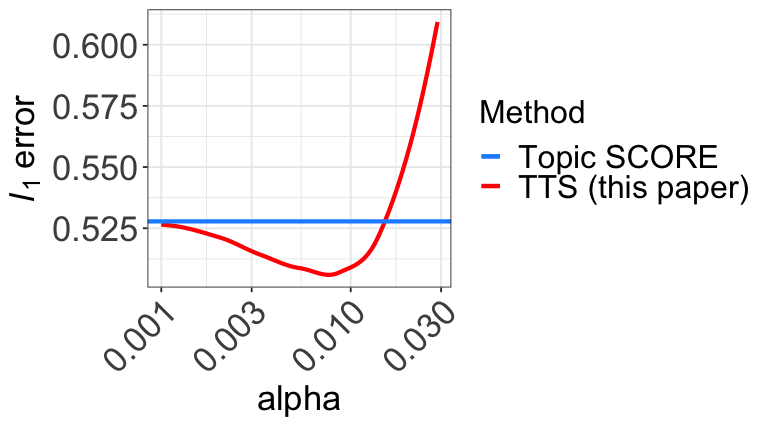}
    \captionsetup{width=0.8\textwidth}
         \caption{ Average $\mathcal{L}_1(\hat{A}, A)$-error as a function of the threshold parameter $\alpha$. }
         \label{fig:alpha_vs_l1}
     \end{subfigure}
\begin{subfigure}[t]{0.48\textwidth}
\centering
    \includegraphics[width=0.8\textwidth, height=4.2cm]{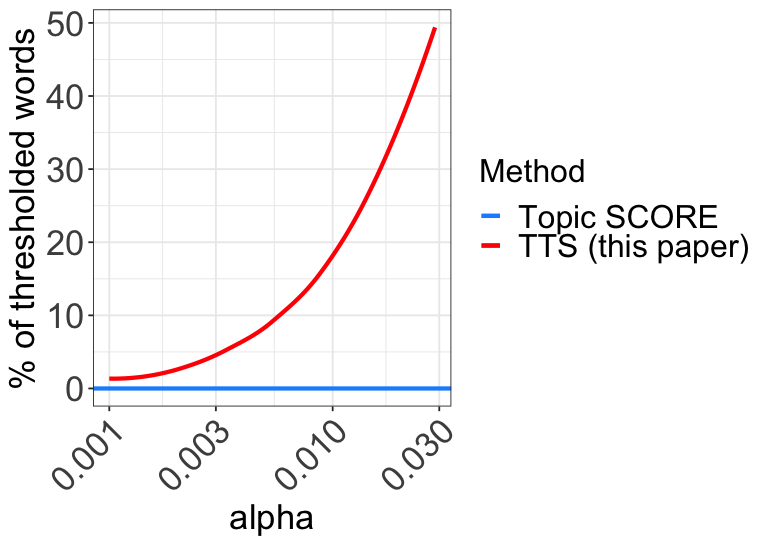}
    \captionsetup{width=0.8\textwidth}
         \caption{Corresponding percentage of the words discarded by thresholding as a function of $\alpha$.}
         \label{fig:alpha_vs_percent}
     \end{subfigure}
    \caption{$\mathcal{L}_1(\hat{A}, A)$-error averaged over $20$ independent trials and the percentage of words removed as $\alpha$ increases, for a synthetic dataset with $p = 5000, n= N = 500$.}
    \label{fig:thresholds}
\end{figure}

As we mentioned, the universal parameter $\alpha$ should be independent of $(p,n,N, K)$. Our recommended value of $\alpha = 0.005$ is obtained based on numerous such experiments with synthetic data where we vary the values of $(p,n, N, K)$. This choice of $\alpha$ also works well in all real data applications of Section \ref{sec: real-data}, where several parameter regimes are involved.\\

\xhdr{Additional experiments and conclusion} We also evaluate the impact of other aspects of the data generation mechanism on our estimator's performance. We find that changing $\delta_{\text{anchor}}$, which controls the frequency of anchor words, does not significantly impact the overall performance of TTS. This is an advantage of SCORE-based methods over methods that rely on anchor words identification, which are often affected by the frequency of anchor words both in theory and in practice. Additionally, when we increase the parameter $a_\text{zipf}$ in \eqref{eq:zipf}, we find that our estimator's performance improves significantly. This is not surprising as a larger $a_\text{zipf}$ means the ordered entries of $A$'s columns decay to zero faster, and our theoretical results also show that a strong sparsity regime (when $q$ is close to 0 in Assumption \ref{ass:sparsity}) is favorable to our method. Further details about these experiments are deferred to Appendix \ref{appendix:experiment}. Finally,  we check the performance of our method on 
 a set of semi-synthetic experiments based on the Associated Press dataset 
(included in the R package \texttt{tm} \citep{feinerer2015package}), thereby allowing us to test a different data generating mechanism. 
 The results are also presented in Appendix \ref{appendix:experiment}.

Overall, we have illustrated that our method {\it(a) performs well in a wide variety of parameter regimes, and notably in the high-dimensional setting where $p$ is large}, and {\it(b) performs well even if our sparsity assumption is violated} (see the discussion on the uniform data generation mechanism, and also note that we use a weak sparsity regime with $a_\text{zipf} \approx 1$ in most of our experiments). This makes our method applicable to the vast majority of real-world text datasets, which often are high-dimensional and exhibit Zipf's law decay. However, alternative methods such as LDA and AWR may still be competitive in some settings, especially when the pLSI model fails to hold or if the number of documents $n$ and the document length $N$ are unusually small relative to the number of topics $K$.

\section{Practical applications in text analysis and beyond}\label{sec: real-data}

In this section, we deploy our method on real-world datasets. Given the results of the previous section, we focus here on the comparison of our method with Topic-SCORE \citep{ke2022using} and LDA \citep{blei2003latent}. 

Real datasets seldom have ground truth for $A$, and some may even lack an obvious choice for the number of topics $K$. Consequently, in this section we evaluate the estimators' performance using, when appropriate, the following metrics:
\begin{enumerate}[(a)]
    \item {\it Topic Resolution} as a measure of topic consistency. We fit each estimator on two disjoint halves of the data and report the cosine similarity between estimated topics (after an appropriate permutation of the columns of $A$). Mathematically, letting $\hat{A}^{(i)}, i \in \{1, 2\}$ denote the estimated topic-word matrices obtained for each half of the data, we define the ``average topic resolution'' $\eta$ as the mean cosine similarity (a classical similarity metric in natural language processing) between aligned topics:

    \begin{equation}\label{equation: topic resolution}
        \eta =  \max_{\sigma \in \Pi_K} \frac{1}{K}\sum_{k=1}^K \frac{\hat{A}^{(1)\top}_{*k} \hat{A}^{(2)}_{*
        \sigma(k)}}{\| \hat{A}^{(1)}_{*k}\|_2\| \hat{A}^{(2)}_{*\sigma(k)}\|_2 },
    \end{equation}
        where $\Pi_K$ denotes the set of all permutations of $[K]$.
    Thus, higher resolution indicates better-defined and more consistent topic vectors (although this does not necessarily mean better $\ell_1$-error).
    \item {\textit{Multiscale Topic Refinement and Coherence} (\cite{fukuyama2021multiscale}):} In the absence of an obvious number of topics $K$,  we fit the method for multiple values of $K$ and analyze the resulting topic hierarchy to check the stability of our estimator. 
    We follow in particular the methodology of \cite{fukuyama2021multiscale}, which was developed to guide the choice of an appropriate number of topics $K$ for LDA \citep{blei2003latent} by investigating the relationships among topics of increasing granularity. Given a hierarchy of topics, the method evaluates which topics consistently appear, constantly split, or are merely transient. 
    We use these tools here (and its associated package \texttt{alto} \citep{fukuyama2021multiscale}) to analyze our estimator. The method of \cite{fukuyama2021multiscale} starts by computing the alignment of topics across the hierarchy using the transport distance: for each $K$, this method computes how the mass of topic $j \in \{1, \cdots, K\}$ is split amongst the $K+1$ topics at the next level of the hierarchy. We refer the reader to the original work by \cite{fukuyama2021multiscale} for a more detailed explanation of topic transport alignment. Once the relationships between consecutive topic models have been established, the method of \cite{fukuyama2021multiscale} allows visualization of (a) topic refinement (i.e., whether topics increase in granularity, as indicated by a small number of ancestors in the hierarchy; or conversely, whether topics are perpetually recombined from one level of the hierarchy to the next); and (b) topic coherence (whether a topic appears across multiple values of $K$). We choose here to favour methods with improved topic coherence and topic refinement, since there are markers of topic stability.
\end{enumerate}
We explore the comparison between our method, LDA and Topic-SCORE under diverse parameter  regimes (with varying $n$, $N$ and $p$).

\subsection{Research articles (high \texorpdfstring{$p$}{p}, high \texorpdfstring{$n$}{n}, low \texorpdfstring{$N$}{N})}
For our first experiment, we consider a corpus of 20,140 research abstracts belonging to (at least) one of four categories: Computer Science, Mathematics, Physics and Statistics\footnote{The data is available on Kaggle at \hyperlink{https://www.kaggle.com/datasets/blessondensil294/topic-modeling-for-research-articles/code?resource=download }{this link}. Although the original data set comprises six topics (with the addition of Quantitative Biology and Finance), due to the low representation of these last two topics ($<4\%$ of the data), we drop them from our analysis.}. 
After pre-processing of the data (including the removal of standard stop words, numbers, and punctuation), our dataset involves a dictionary of size $p=81,649$  and $n=20,140$ documents with an average document size of $N=157$ words.

 \par We first evaluate the topic consistency of all methods in estimating the topic-word matrix $A$ using the mean topic resolution defined in equation \eqref{equation: topic resolution}. Table~\ref{tab:nlp topic resolution} displays  the average topic resolution over 25 random splits of the data.
        \begin{table}[ht]
            \centering
            \begin{tabular}{c|cc}
            \hline
               Methods  & Average Topic Resolution($\eta$) & Interquantile range \\
               \hline LDA (Blei et al)& 0.304 &(0.270,0.330) \\
                 TTS (this paper)& 0.332 &(0.310,0.360) \\
                 Topic-SCORE (Ke et al) &0.145& (0.093,0.179)\\
               \hline
            \end{tabular}
            \caption{Average Topic Resolution on research article data. The interquartile range for the average topic resolution was computed over 25 random splits of the data.}
            \label{tab:nlp topic resolution}
        \end{table}
As highlighted in the introductory paragraph to this section, topic resolution can be taken as an indicator of the stability of the estimator of $\hat{A}$ between two separate portions of the data. A method that produces higher topic resolution with a narrower interquartile range indicates a more stable estimation of the topic-word matrix $A$. As shown in Table \ref{tab:nlp topic resolution}, our approach consistently outperforms LDA and Topic-SCORE on this metric; it offers the highest average topic resolution score. Topic-SCORE's performance exhibits more significant fluctuations, as indicated by its larger interquartile range. 
\par Taking a closer look at the estimation of $A$, we consider the 10 most representative words generated by each of the three methods for every topic (obtained by selecting the top 10 largest entries in each column of $\hat{A}$). The results are presented in Tables \ref{table: common words ours}, \ref{table: common words LDA}, and \ref{table: common words tracy}. For the topics of Computer Science and Statistics, the top 10 most representative words produced by our method agree with 70\% of LDA's most representative words in the corresponding topics. There is much less agreement for the topic of Physics, but upon closer inspection we find that some of the words produced by our method in that category (such as `magnetic', `energy') are more indicative of the topic of Physics, whereas all of the top 10 words for Physics produced by LDA are generic words that can appear in other categories. 

In contrast, the results of Topic-SCORE (Table \ref{table: common words tracy}) seem to diverge substantially from those of LDA and our method. It appears that the top 10 most representative words for Physics, Mathematics and Statistics from Topic-SCORE are dominated by infrequently occuring words and foreign words; the foreign words can be traced back to a few rare abstracts written in English and followed by a foreign language translation. This supports our hypothesis that Topic-SCORE amplifies the effects of infrequent words, unless significant \textit{ad hoc} data pre-processing (removal or merger of rare words, and removal of documents with significant numbers of rare words) is applied. 

\begin{table}[ht]
\centering
\begin{tabular}{c|l}
  \hline
 & Top 10 most representative words per topic  \\ 
  \hline
Computer Science & ``learning" ``network"  ``networks" ``model"    "can"      ``neural"  \\ &
 ``deep"     ``using"    "models"   ``data"   \\ 
 
  Physics & ``model"    ``can"      ``system"   ``field"    ``energy"   ``systems" \\ &
 ``magnetic" ``models"   ``using"    ``phase"   
 \\ 
  Mathematics & ``problem"   ``can"       ``algorithm" ``show"      ``method"   
 ``paper" \\&    ``results"   ``also"      ``time"      ``using"   \\ 
  Statistics & ``data"     ``model"    ``can"      ``learning" ``using"    ``models"  \\&
 ``method"   ``approach" ``based"    ``paper"  \\ 
   \hline
\end{tabular}
\caption{Most common words found by our method}
\label{table: common words ours}
\end{table}
\begin{table}[ht]
\centering
\begin{tabular}{c|l}
  \hline
 & Top 10 most representative words per topic   \\ 
  \hline
Computer Science & ``data"     ``network"  ``learning" ``networks" ``can"      ``model"   
 \\ &``using"    ``new"      ``paper"    ``based"    \\ 
 
  Physics & ``show"        ``data"        ``analysis"    ``two"         ``can"        
 ``problem" \\&    ``results"     ``field"       ``system"      ``performance"
 \\ 
  Mathematics & ``can"      ``used"     ``models"   ``using"    ``model"    ``paper"   
 \\& ``number"   ``method"   ``proposed" ``approach"  \\ 
  Statistics &  ``model"    ``results"  ``show"     ``can"      ``learning" ``method"  
\\& ``using"    ``based"    ``data"     ``also" \\ 
   \hline
\end{tabular}
\caption{Most common words found by LDA}
\label{table: common words LDA}
\end{table}
\begin{table}[ht]
\centering
\begin{tabular}{c|l}
  \hline
 & Top 10 most representative words per topic   \\ 
  \hline
Computer Science & ``data"     ``can"      ``model"    ``using"    "learning" ``show"    
\\& ``results"  ``method"   ``paper"    ``also"     \\ 
 
  Physics & ``della"      ``quantum"    ``theory"      ``del"         ``year"       
 ``teoria"  \\&    ``quantistica" ``per"         ``nel"        ``delle"  
 \\ 
  Mathematics &  ``die"         ``der"         ``collectors"  ``problem"     ``able"       
 ``assumptions" ``coupon"   \\&   ``wir"         ``based"      ``one"        \\ 
  Statistics & ``der"             ``und"             ``music"          
 ``automatischen"   ``learning"        ``sheet"          
\\&``die"             ``musikverfolgung" ``deep"   
``algorithms"  \\ 
   \hline
\end{tabular}
\caption{Most common words found by Topic SCORE}
\label{table: common words tracy}
\end{table}

 In order to further investigate the performance gap between TTS and Topic-SCORE, we visualize the point cloud from both methods in Figure \ref{fig:point-cloud-nlp}. As expected, we observe that the Topic-SCORE point cloud is severely stretched by a set of low-frequency words that include several foreign words.
Again, with the presence of many rare words in the dataset, 
the lack of thresholding and the use of the pre-SVD multiplication step in Topic-SCORE contribute to a significant distortion of the point cloud. In comparison, the thresholding approach we adopt yields a more compact point cloud. As demonstrated in Figure \ref{fig:point-cloud-nlp}b, our method effectively recaptures the essential vertices of the point cloud simplex.  A closer look at the words surrounding each vertex, as shown in Figure \ref{fig:sub2nlppointcloud}, allows us to easily identify which simplex vertex belongs to which topic (Physics, Math, Computer Science and Statistics when moving in the anticlockwise direction). Under this ``large $p$'' regime and in the presence of a myriad of rare words that may introduce significant noise, our method not only distinguishes words effectively but also clusters them into well-defined topics. 
\begin{figure}[htbp]
    \centering
    \begin{subfigure}[b]{0.45\textwidth}
        \includegraphics[width=\textwidth]{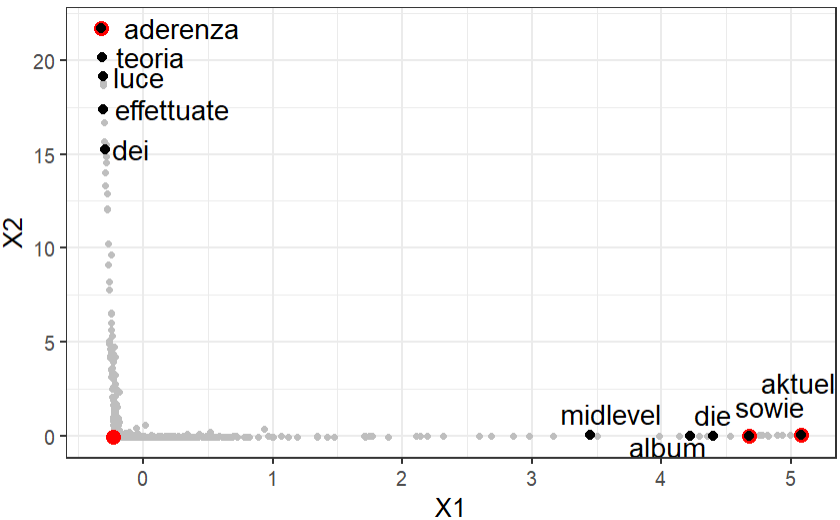}
        \caption{Point cloud for $K=4$ from Topic-SCORE}
        \label{fig:sub1nlppointcloud}
    \end{subfigure}
    \hfill
    \begin{subfigure}[b]{0.45\textwidth}
        \includegraphics[width=\textwidth]{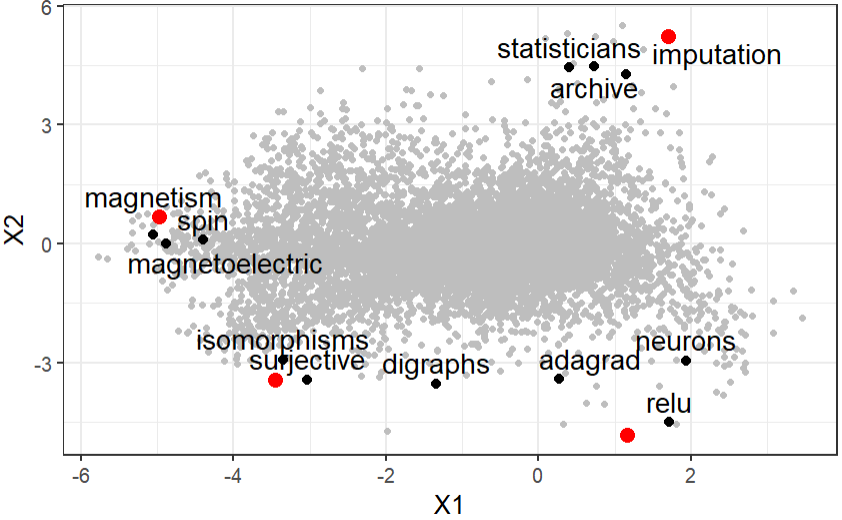}
        \caption{Point cloud for $K=4$ from our method}
        \label{fig:sub2nlppointcloud}
    \end{subfigure}
    \caption{Comparison of the 3-dimensional point clouds from TTS (right) and Topic-SCORE (left), projected on the first two axes for visualization. Estimated vertices are colored red, and the point clouds are represented by gray dots. Most  outlying words in Topic-SCORE's point cloud are thresholded away by TTS, thus contributing to higher point cloud stability for our method.}
    \label{fig:point-cloud-nlp}
\end{figure}

\par We note that this dataset comes with manually curated topic labels for each document. As a final verification, we analyze the performance of the different methods when used for recovering the ground truth labels for each document. Having estimated $A$, it is quite natural in light of the pLSI model to perform regression of $D$ against $\hat{A}$ in order to yield an estimator of $W$. To this end, we use the estimation procedure for $W$ in \cite{ke2022using}, where the problem of estimating $W$ given $\hat{A}$ is reduced to a weighted constrained linear regression problem:
\begin{equation}\label{est-procedure-for-W-Ke}
    \forall i \in [n], \quad  \hat{W}_{*i} = \text{argmin}_{\omega \in [0,1]^{K}} \frac{1}{p} \sum_{j=1}^p \frac{1}{M_{jj}}( D_{ji} - \sum_{k=1}^K\hat{A}_{jk}\omega_{ki})^2
\end{equation}
We strongly emphasize that the aim of this experiment is to evaluate the estimation of $A$, and we do not claim here that our method provides state-of-the-art results in the estimation of $W$. Other potentially better estimation procedures are available for $W$, many of which do not require estimating $A$ first. Rather, as topic labels are available for this dataset, we use this simple estimation procedure for $W$ via $\hat{A}$ as another way of comparing the quality of $\hat{A}$ obtained from TTS, Topic-SCORE and LDA. Since the $\hat{W}$ obtained from \eqref{est-procedure-for-W-Ke} depends on $\hat{A}$ as input, it stands to reason that a better estimation procedure for $A$ may be reflected in a better agreement between $\hat{W}$ and the provided topic labels for each document, if we use \eqref{est-procedure-for-W-Ke} to estimate $W$. 

Let $y_{ki} = 1$ if document $i$ is labeled as belonging to topic $k$ (and $y_{ki}=0$ otherwise). We compute the average $l_1$ distance $ \mathcal{D}(\hat W, y)$ and cosine similarity $ S_k$ between the permuted matrix $\hat W$ and the provided labels $y$ for each topic $k$ as follows: 
    \begin{equation}
        \mathcal{D}(\hat W, y):=\min_{\sigma \in \Pi_K}\frac{1}{nK}\sum_{ki}|\hat W_{\sigma(k)i}-y_{ki}|, \quad S_k =\max_{\sigma\in \Pi_K} \frac{\sum_{i=1}^n\hat{W}_{\sigma(k)i}y_{ki}}{\|\hat{W}_{\sigma(k)*}\|_2 \|y_{k*}\|_2}
    \end{equation}
  Here, a smaller value of  the $l_1$ distance or a larger value of the cosine similarity score between $y$ and $\hat{W}$ indicate greater alignment with the provided topic labels.  The results are displayed in Table~\ref{table: summary of nlp}.

\begin{table}[ht]
\centering
\begin{tabular}{c|ccccc|c}
  \hline
 Methods & $S_{\mathsf{CS}}$  & $S_{\mathsf{Phys}}$& $S_{\mathsf{Math}}$& $S_{\mathsf{Stat}}$ & $\bar S$ &
$\mathcal{D}(\hat W, y)$ \\ 
  \hline
LDA(Blei et al) &0.671& 0.576&0.534& 0.493 &0.569& 0.403\\ 
 
 TTS(this paper) &0.610& 0.748&0.636&0.494&0.622&0.305\\ 
 
 Topic SCORE(Ke et al) &0.670& 0.545&0.588&0.373&0.544& 0.348\\

   \hline
\end{tabular}
\caption{The evaluation of $\hat W$ obtained via estimating $A$ first by using the three methods. $\bar S$ is the average cosine similarity across all $K$ topics}
\label{table: summary of nlp}
\end{table}

Table \ref{table: summary of nlp} indicates that our method improves the estimation of $W$ overall and provides the best topic alignment on average, when using \eqref{est-procedure-for-W-Ke} to estimate $W$. This suggests that our procedure yields a more accurate estimator of $A$.  

\subsection{Single cell analysis (low \texorpdfstring{$p$}{p}, high \texorpdfstring{$n$}{n}, low \texorpdfstring{$N$}{N})}

In this subsection, we consider a different application area for our methodology: the analysis of single-cell data. We revisit the mouse spleen dataset presented by \cite{goltsev2018deep}. This dataset consists of a set of images from both healthy and diseased mouse spleens. Each sample undergoes staining with 30 different antibodies via the CODEX process, as detailed in \cite{goltsev2018deep}. In \cite{chen2020modeling}, each spleen sample is divided into a set of non-overlapping Voronoi bins, and the count of immune cell types is recorded in each bin.  In this framework, each bin can be viewed as a document and cell types correspond to words. It is of interest to determine appropriate groupings of cell types (topics), as this may help one study the interactions between cells. 

\par Since this dataset does not come with ground-truth labels, we sample two disjoint sets of  size $n=10,000$ out of the 100,840 Voronoi tessellations across all spleen samples (where 10,000 is a number chosen to be large enough to ensure a ``high $n$'' regime while still allowing all methods to have reasonable computational runtimes). On the contrary, there are only 24 different cell types ($p=24$), while the average ``document'' length is $N=11.2$ with an interquartile range of $(6, 16)$.  While \cite{chen2020modeling} focus on evaluating estimators of the matrix $W$, here we repurpose the use of this dataset to study our estimator of $A$.
In this dataset, the precise number of topics $K$ is unknown. We thus apply the three methods for different values of $K$ and use the metrics introduced at the beginning of this section (topic resolution, topic coherence and refinement) to compare the three methods. The results are presented in Figures~\ref{fig:res-spleen} and \ref{fig:spleen-refinement}.\\

\xhdr{Discussion of the results}  Due to the structured nature of this dataset, all methods perform remarkably well, exhibiting an average topic similarity above 0.95. Going into more details, we see that our method outperforms Topic-SCORE in terms of topic resolution. In particular, Topic-SCORE (in red) appears to have more variable performance, as reflected in its larger interquartile ranges and its jittery resolution as a function of $K$. Interestingly, in this specific instance, LDA seems to score higher on topic resolution (although we again emphasize that all methods perform very well on this metric). Additionally, Figure~\ref{fig:refinement-spleen-TTS}  shows the refinement and coherence of the topics for our method as $K$ increases, in contrast to those of LDA in Figure~\ref{fig:refinement-spleen-lda}. In this data example, our method seems to provide topics with higher refinement (fewer ancestors per topic) and higher coherence (note in particular the stability of topic 1, 2, and 18) compared to LDA. In Figure \ref{fig:refinement-spleen-lda}, it can be observed that topics 1, 2, and 18 are dispersed across different branches within the refinement plot as $K$ varies.

\begin{figure}[htbp!]
    \centering
    \includegraphics[width=\textwidth]{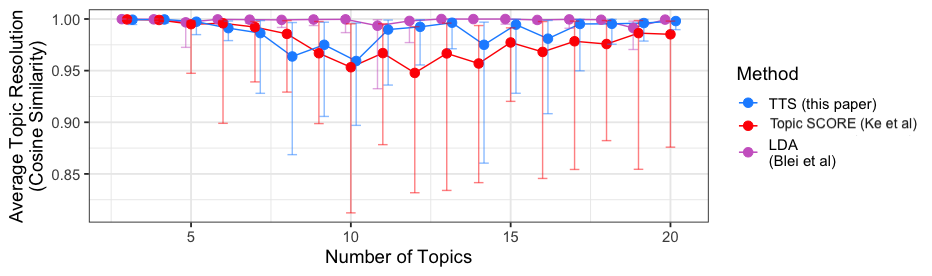}
    \caption{Median Topic Resolution as a function of $K$ on the Mouse Spleen Data \citep{goltsev2018deep, chen2020modeling}. Vertical error bars represent the interquartile range for the average topic resolution scores over 25 trials.}
    \label{fig:res-spleen}
\end{figure}


\begin{figure}[htbp]
    \centering
    \begin{subfigure}[b]{0.48\textwidth}
    \includegraphics[width=\textwidth]{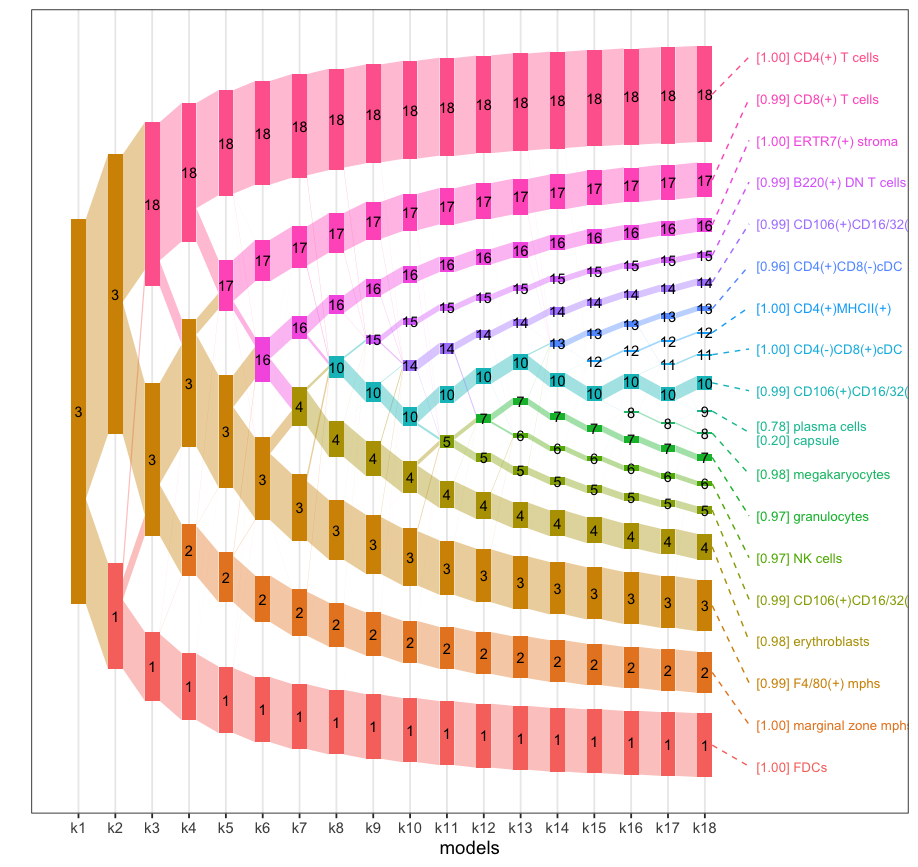}
    \caption{Topic refinement for our method as $K$ varies, provided by the package \texttt{alto} \citep{fukuyama2021multiscale}.}
    \label{fig:refinement-spleen-TTS}
    \end{subfigure}
    \hspace{0.2cm}
    \begin{subfigure}[b]{0.48\textwidth}
    \includegraphics[width=\textwidth]{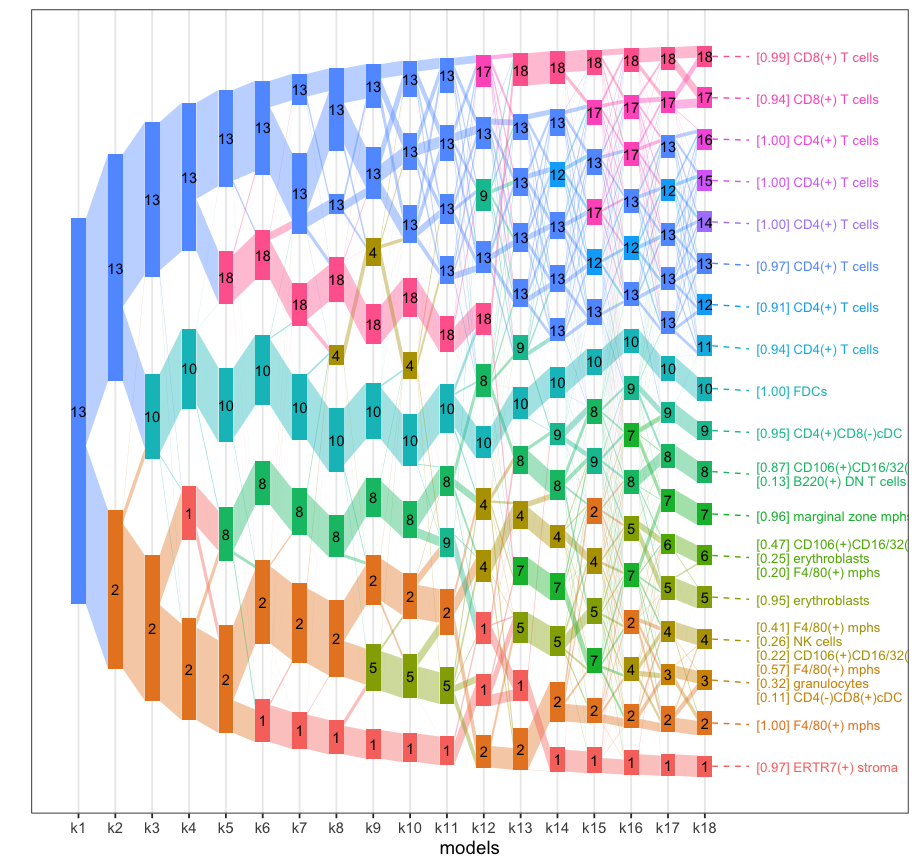}
    \caption{Topic refinement for LDA \citep{blei2003latent} as $K$ varies, provided by the package \texttt{alto} \citep{fukuyama2021multiscale}.}
    \label{fig:refinement-spleen-lda}
    \end{subfigure}
    \caption{Comparison of the refinement and coherence of topics recovered using our method (left) and LDA (right).}
    \label{fig:spleen-refinement}
\end{figure}



\subsection{Microbiome examples (low \texorpdfstring{$p$}{p}, low \texorpdfstring{$n$}{n}, high \texorpdfstring{$N$}{N})}

We finish our discussion with an application of our method to microbiome data analysis. In particular, we reanalyze two datasets that have been previously analyzed through topic modeling: the colon dataset of \cite{yachida2019metagenomic} and the vaginal microbiome example of \cite{callahan2017replication}, which was re-analyzed in \cite{fukuyama2021multiscale} using LDA. Microbiome data are represented in the form of a count matrix. In this matrix, each column corresponds to a different sample, while each row represents various taxa of bacteria. The entries within the matrix represent the abundance of each bacteria in a given sample. Taking samples to be documents and bacteria as words, topic modeling offers an interesting way of exploring communities of bacteria (``topics'') \citep{sankaran2019latent}. For the sake of conciseness, we present the results here for the colon dataset of \cite{yachida2019metagenomic}, and refere the reader to Appendix \ref{app: real data experiment} for the results on the other dataset.

\par After pre-processing and eliminating species with a relative abundance below 0.001\%, this dataset contains microbiome counts for $p=541$ distinct taxa from $n=503$ samples.  In contrast, the length of each ``document'' is extremely high, with around $N=43$ million bacteria per sample. We test all three methods for different values of $K$ and display the average topic resolution in Figure~\ref{fig:resolution-microbiome-cosine}. On this metric, our method exhibits significantly better results than both LDA and  Topic SCORE for up to 15 topics. After 15 topics, LDA outperforms all SCORE-based methods in terms of topic resolutions. However, this comes at a much higher computational cost: while each of the SCORE methods in this example could be fitted in under a minute, each of the LDA fits took on the order of tens of minutes. Note that LDA's high topic resolution could also be due to the higher weight of the prior in the estimation of the topic-word matrix $A$, which, due to the relatively small size of the dataset, could have a stabilizing effect on estimation. On the other hand, the performance of Topic-SCORE quickly drops to 0.65 as $K$ increases, before reaching a plateau at around $K\approx 10$. By contrast, for small $K$, our method exhibits a resolution up to 40\% higher than Topic Score (for $K=10$) before also decreasing as the number of topics increases.

\begin{figure}[htbp!]
    \centering
    \includegraphics[width=0.8\textwidth]{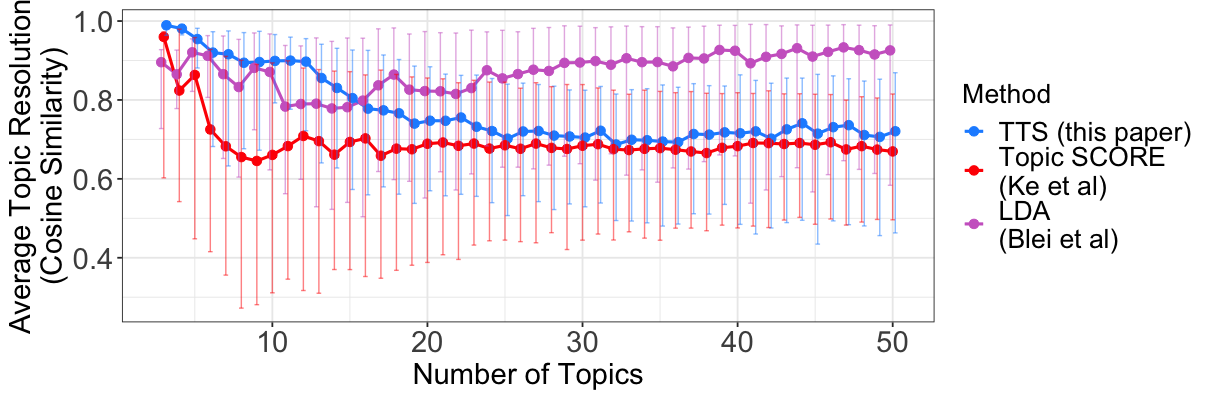}
    \caption{Topic resolution (measured by the average cosine similarity between halves of the data) of our method (in blue) and Topic-SCORE (red) on the microbiome dataset of \cite{yachida2019metagenomic}. Topic resolution is averaged over 25 random splits of the data. }
    \label{fig:resolution-microbiome-cosine}
\end{figure}

\par To understand the gap in performance between Topic-SCORE and our method, we again visualize the point clouds obtained by both methods. The visualization can be found in Figure~\ref{fig:point-cloud-mic} in  Appendix~\ref{app: real data experiment}. Similarly to our first example with text analysis, we observe that the point cloud of Topic-SCORE is heavily distorted; in contrast, ours is more compact. 

\section{Conclusion and future works} 
In this paper, we introduce \textit{Thresholded Topic-SCORE} (TTS), a new estimation procedure for the word-topic matrix $A$ that is based on eigenvalue decomposition and thresholding. Our procedure is shown to perform well under the column-wise $\ell_q$-sparsity assumption \eqref{eq: lq-sparsity-def}, which is exhibited by many real-world text datasets but to our knowledge has not been considered in prior works. TTS also accommodates non-separable data, simply by adopting a suitable vertex hunting algorithm such as Archetype Analysis from \cite{javadi2020nonnegative}. Empirical results show that our method is competitive for a diverse range of parameter regimes, especially in the ``large $p$'' setting where many words occur infrequently in the corpus. Overall, TTS is a compelling alternative to existing methods when $p$ is large and the number of topics $K$ is relatively small. 

Based on this paper, some potential research directions can be suggested. First, the estimation of $W$ is also of interest in applications, and the minimax-optimal $\ell_1$-error rate for estimating $W$ has been established (see for example \cite{klopp2021assigning}). However, the problem of estimating $W$ essentially involves $n$ independent sub-problems (one for each column of $W$), and consequently the minimax-optimal error rate for $W$ scales significantly with $n$. One may consider imposing additional structural assumptions on how the $n$ documents are related to one another, in order to design an estimation procedure whose error decays to zero as $n \to \infty$. 

Second, the minimax-optimal $\ell_1$-error rate given the $\ell_q$-assumption \eqref{eq: lq-sparsity-def} remains an open problem. The minimax lower bound arguments of \cite{ke2022using} are not directly applicable, as they are based on constructing hypotheses of $A$ whose columns contain entries that are of roughly equal magnitudes (i.e. the columns do not exhibit $\ell_q$-decay). 

Third, our proposed method only makes use of the word counts for the documents, as the underlying pLSI model disregards the order of words in a document. Other language models, such as the multi-gram topic model, make use of word orders. An extension of our method using tensor factorization may be possible in this setting, as the corpus is stored using a multi-way tensor \citep{zheng2016topic}.

\acks{The authors would like to acknowledge support for this project
from the National Science Foundation (NSF grant IIS-2238616). This work was completed in part with resources provided by the University of Chicago’s Research Computing Center. }

\bibliography{citation} 

\begin{thebibliography}{36}
\providecommand{\natexlab}[1]{#1}
\providecommand{\url}[1]{\texttt{#1}}
\expandafter\ifx\csname urlstyle\endcsname\relax
  \providecommand{\doi}[1]{doi: #1}\else
  \providecommand{\doi}{doi: \begingroup \urlstyle{rm}\Url}\fi

\bibitem[Ara{\'u}jo et~al.(2001)Ara{\'u}jo, Saldanha, Galvao, Yoneyama, Chame,
  and Visani]{araujo2001successive}
M.~C.~U. Ara{\'u}jo, T.~C.~B. Saldanha, R.~K.~H. Galvao, T.~Yoneyama, H.~C.
  Chame, and V.~Visani.
\newblock The successive projections algorithm for variable selection in
  spectroscopic multicomponent analysis.
\newblock \emph{Chemometrics and Intelligent Laboratory Systems}, 57\penalty0
  (2):\penalty0 65--73, 2001.

\bibitem[Arora et~al.(2012)Arora, Ge, and Moitra]{arora2012learning}
S.~Arora, R.~Ge, and A.~Moitra.
\newblock Learning topic models--going beyond svd.
\newblock In \emph{2012 IEEE 53rd Annual Symposium on Foundations of Computer
  Science}, pages 1--10. IEEE, 2012.

\bibitem[Bicego et~al.(2012)Bicego, Lovato, Perina, Fasoli, Delledonne,
  Pezzotti, Polverari, and Murino]{bicego2012investigating}
M.~Bicego, P.~Lovato, A.~Perina, M.~Fasoli, M.~Delledonne, M.~Pezzotti,
  A.~Polverari, and V.~Murino.
\newblock Investigating topic models' capabilities in expression microarray
  data classification.
\newblock \emph{IEEE/ACM Transactions on Computational Biology and
  Bioinformatics}, 9\penalty0 (6):\penalty0 1831--1836, 2012.

\bibitem[Bing et~al.(2020{\natexlab{a}})Bing, Bunea, and Wegkamp]{bing2020fast}
X.~Bing, F.~Bunea, and M.~Wegkamp.
\newblock A fast algorithm with minimax optimal guarantees for topic models
  with an unknown number of topics.
\newblock \emph{Bernoulli}, 2020{\natexlab{a}}.

\bibitem[Bing et~al.(2020{\natexlab{b}})Bing, Bunea, and
  Wegkamp]{bing2020optimal}
X.~Bing, F.~Bunea, and M.~Wegkamp.
\newblock Optimal estimation of sparse topic models.
\newblock \emph{The Journal of Machine Learning Research}, 21\penalty0
  (1):\penalty0 7189--7233, 2020{\natexlab{b}}.

\bibitem[Blei et~al.(2003)Blei, Ng, and Jordan]{blei2003latent}
D.~M. Blei, A.~Y. Ng, and M.~I. Jordan.
\newblock Latent dirichlet allocation.
\newblock \emph{Journal of Machine Learning Research}, 3\penalty0
  (Jan):\penalty0 993--1022, 2003.

\bibitem[Cai and Zhou(2012)]{cai2012optimal}
T.~T. Cai and H.~H. Zhou.
\newblock Optimal rates of convergence for sparse covariance matrix estimation.
\newblock \emph{The Annals of Statistics}, 2012.

\bibitem[Callahan et~al.(2017)Callahan, DiGiulio, Goltsman, Sun, Costello,
  Jeganathan, Biggio, Wong, Druzin, Shaw, et~al.]{callahan2017replication}
B.~J. Callahan, D.~B. DiGiulio, D.~S.~A. Goltsman, C.~L. Sun, E.~K. Costello,
  P.~Jeganathan, J.~R. Biggio, R.~J. Wong, M.~L. Druzin, G.~M. Shaw, et~al.
\newblock Replication and refinement of a vaginal microbial signature of
  preterm birth in two racially distinct cohorts of us women.
\newblock \emph{Proceedings of the National Academy of Sciences}, 114\penalty0
  (37):\penalty0 9966--9971, 2017.

\bibitem[Chen et~al.(2020)Chen, Soifer, Hilton, Keren, and
  Jojic]{chen2020modeling}
Z.~Chen, I.~Soifer, H.~Hilton, L.~Keren, and V.~Jojic.
\newblock Modeling multiplexed images with spatial-lda reveals novel tissue
  microenvironments.
\newblock \emph{Journal of Computational Biology}, 27\penalty0 (8):\penalty0
  1204--1218, 2020.

\bibitem[Corral et~al.(2015)Corral, Boleda, and Ferrer-i
  Cancho]{corral2015zipf}
{\'A}.~Corral, G.~Boleda, and R.~Ferrer-i Cancho.
\newblock Zipf’s law for word frequencies: Word forms versus lemmas in long
  texts.
\newblock \emph{PloS one}, 10\penalty0 (7):\penalty0 e0129031, 2015.

\bibitem[Curiskis et~al.(2020)Curiskis, Drake, Osborn, and
  Kennedy]{curiskis2020evaluation}
S.~A. Curiskis, B.~Drake, T.~R. Osborn, and P.~J. Kennedy.
\newblock An evaluation of document clustering and topic modelling in two
  online social networks: Twitter and reddit.
\newblock \emph{Information Processing \& Management}, 57\penalty0
  (2):\penalty0 102034, 2020.

\bibitem[Donoho and Stodden(2003)]{donoho2003does}
D.~Donoho and V.~Stodden.
\newblock When does non-negative matrix factorization give a correct
  decomposition into parts?
\newblock \emph{Advances in Neural Information Processing Systems}, 16, 2003.

\bibitem[Feinerer et~al.(2015)Feinerer, Hornik, and
  Feinerer]{feinerer2015package}
I.~Feinerer, K.~Hornik, and M.~I. Feinerer.
\newblock Package ‘tm’.
\newblock \emph{Corpus}, 10\penalty0 (1), 2015.

\bibitem[Fukuyama et~al.(2021)Fukuyama, Sankaran, and
  Symul]{fukuyama2021multiscale}
J.~Fukuyama, K.~Sankaran, and L.~Symul.
\newblock Multiscale analysis of count data through topic alignment.
\newblock \emph{arXiv preprint arXiv:2109.05541}, 2021.

\bibitem[Ge and Zou(2015)]{ge2015intersecting}
R.~Ge and J.~Zou.
\newblock Intersecting faces: Non-negative matrix factorization with new
  guarantees.
\newblock In \emph{International Conference on Machine Learning}, pages
  2295--2303. PMLR, 2015.

\bibitem[Gillis and Vavasis(2013)]{gillis2013fast}
N.~Gillis and S.~A. Vavasis.
\newblock Fast and robust recursive algorithmsfor separable nonnegative matrix
  factorization.
\newblock \emph{IEEE Transactions on Pattern Analysis and Machine
  Intelligence}, 36\penalty0 (4):\penalty0 698--714, 2013.

\bibitem[Goltsev et~al.(2018)Goltsev, Samusik, Kennedy-Darling, Bhate, Hale,
  Vazquez, Black, and Nolan]{goltsev2018deep}
Y.~Goltsev, N.~Samusik, J.~Kennedy-Darling, S.~Bhate, M.~Hale, G.~Vazquez,
  S.~Black, and G.~P. Nolan.
\newblock Deep profiling of mouse splenic architecture with codex multiplexed
  imaging.
\newblock \emph{Cell}, 174\penalty0 (4):\penalty0 968--981, 2018.

\bibitem[Greenbaum et~al.(2020)Greenbaum, Li, and Overton]{greenbaum2020first}
A.~Greenbaum, R.-c. Li, and M.~L. Overton.
\newblock First-order perturbation theory for eigenvalues and eigenvectors.
\newblock \emph{SIAM review}, 62\penalty0 (2):\penalty0 463--482, 2020.

\bibitem[Hofmann(1999)]{hofmann1999probabilistic}
T.~Hofmann.
\newblock Probabilistic latent semantic indexing.
\newblock In \emph{Proceedings of the 22nd annual international ACM SIGIR
  Conference on Research and Development in Information Retrieval}, pages
  50--57, 1999.

\bibitem[Horn and Johnson(2012)]{horn2012matrix}
R.~A. Horn and C.~R. Johnson.
\newblock \emph{Matrix analysis}.
\newblock Cambridge University Press, 2012.

\bibitem[Javadi and Montanari(2020)]{javadi2020nonnegative}
H.~Javadi and A.~Montanari.
\newblock Nonnegative matrix factorization via archetypal analysis.
\newblock \emph{Journal of the American Statistical Association}, 115\penalty0
  (530):\penalty0 896--907, 2020.

\bibitem[Jin(2015)]{jin2015fast}
J.~Jin.
\newblock Fast community detection by score.
\newblock \emph{The Annals of Statistics}, 2015.

\bibitem[Jin et~al.(2017)Jin, Ke, and Luo]{jin2017estimating}
J.~Jin, Z.~T. Ke, and S.~Luo.
\newblock Estimating network memberships by simplex vertex hunting.
\newblock \emph{arXiv preprint arXiv:1708.07852}, 12, 2017.

\bibitem[Ke and Wang(2022)]{ke2022using}
Z.~T. Ke and M.~Wang.
\newblock Using svd for topic modeling.
\newblock \emph{Journal of the American Statistical Association}, pages 1--16,
  2022.

\bibitem[Klopp et~al.(2021)Klopp, Panov, Sigalla, and
  Tsybakov]{klopp2021assigning}
O.~Klopp, M.~Panov, S.~Sigalla, and A.~Tsybakov.
\newblock Assigning topics to documents by successive projections.
\newblock \emph{arXiv preprint arXiv:2107.03684}, 2021.

\bibitem[Li et~al.(2010)Li, Wang, Lim, Blei, and Fei-Fei]{li2010building}
L.-J. Li, C.~Wang, Y.~Lim, D.~M. Blei, and L.~Fei-Fei.
\newblock Building and using a semantivisual image hierarchy.
\newblock In \emph{2010 IEEE Computer Society Conference on Computer Vision and
  Pattern Recognition}, pages 3336--3343. IEEE, 2010.

\bibitem[Ma(2013)]{ma2013sparse}
Z.~Ma.
\newblock Sparse principal component analysis and iterative thresholding.
\newblock \emph{The Annals of Statistics}, 2013.

\bibitem[Piantadosi(2014)]{piantadosi2014zipf}
S.~T. Piantadosi.
\newblock Zipf’s word frequency law in natural language: A critical review
  and future directions.
\newblock \emph{Psychonomic Bulletin \& Review}, 21:\penalty0 1112--1130, 2014.

\bibitem[Pritchard et~al.(2000)Pritchard, Stephens, and
  Donnelly]{pritchard2000inference}
J.~K. Pritchard, M.~Stephens, and P.~Donnelly.
\newblock Inference of population structure using multilocus genotype data.
\newblock \emph{Genetics}, 155\penalty0 (2):\penalty0 945--959, 2000.

\bibitem[Rudin et~al.(1976)]{rudin1976principles}
W.~Rudin et~al.
\newblock \emph{Principles of mathematical analysis}, volume~3.
\newblock McGraw Hill, New York, 1976.

\bibitem[Sankaran and Holmes(2019)]{sankaran2019latent}
K.~Sankaran and S.~P. Holmes.
\newblock Latent variable modeling for the microbiome.
\newblock \emph{Biostatistics}, 20\penalty0 (4):\penalty0 599--614, 2019.

\bibitem[Satopaa et~al.(2011)Satopaa, Albrecht, Irwin, and
  Raghavan]{satopaa2011finding}
V.~Satopaa, J.~Albrecht, D.~Irwin, and B.~Raghavan.
\newblock Finding a" kneedle" in a haystack: Detecting knee points in system
  behavior.
\newblock In \emph{2011 31st International Conference on Distributed Computing
  Systems Workshops}, pages 166--171. IEEE, 2011.

\bibitem[Wu et~al.(2022)Wu, Zhang, and Tony~Cai]{wu2022sparse}
R.~Wu, L.~Zhang, and T.~Tony~Cai.
\newblock Sparse topic modeling: Computational efficiency, near-optimal
  algorithms, and statistical inference.
\newblock \emph{Journal of the American Statistical Association}, pages 1--13,
  2022.

\bibitem[Yachida et~al.(2019)Yachida, Mizutani, Shiroma, Shiba, Nakajima,
  Sakamoto, Watanabe, Masuda, Nishimoto, Kubo, et~al.]{yachida2019metagenomic}
S.~Yachida, S.~Mizutani, H.~Shiroma, S.~Shiba, T.~Nakajima, T.~Sakamoto,
  H.~Watanabe, K.~Masuda, Y.~Nishimoto, M.~Kubo, et~al.
\newblock Metagenomic and metabolomic analyses reveal distinct stage-specific
  phenotypes of the gut microbiota in colorectal cancer.
\newblock \emph{Nature Medicine}, 25\penalty0 (6):\penalty0 968--976, 2019.

\bibitem[Zheng et~al.(2016)Zheng, Ding, Lin, and Chen]{zheng2016topic}
X.~Zheng, W.~Ding, Z.~Lin, and C.~Chen.
\newblock Topic tensor factorization for recommender system.
\newblock \emph{Information Sciences}, 372:\penalty0 276--293, 2016.

\bibitem[Zipf(1936)]{zipf2013psycho}
G.~K. Zipf.
\newblock \emph{The Psycho-biology of Language: An Introduction to Dynamic
  Philology}.
\newblock Houghton-Mifflin, 1936.

\end{thebibliography}

\newpage

\begin{appendix}
All proofs in this appendix make use of notations described in Section \ref{section: notations}. Assumptions 1-4 (which include separability) are assumed in Appendix A-D, whereas the sparsity assumption (Assumption \ref{ass:sparsity}) is further imposed in Appendix E. 
\section{Properties of the set \texorpdfstring{$J$}{J}}
    \begin{lemma}[Weak sparsity of $A$]\label{lem:A1} Order the $\ell_2$ row norms of $A$ so that 
    $$\|A_{(1)*}\|_2 \geq \dots \geq \|A_{(p)*}\|_2$$
    Then the matrix $A$ satisfies $\max_{j\in [p]}j\|A_{(j)*}\|_2 \leq K$.
    \end{lemma}
    \begin{proof}
        Observe that for any $j \in [p]$,  
        $$j\|A_{(j)*}\|_2 \leq \sum_{l=1}^p \|A_{l*}\|_2  \leq \sum_{l=1}^p \|A_{l*}\|_1 = K$$
        since $A$ contains only non-negative entries and each column sums up to $1$.
    \end{proof}

    \begin{lemma}\label{lem:A2}
        If $M_0 := \text{diag}(n^{-1}D_0\One_n)$ and $h_j := \|A_{j*}\|_1$, then for any $j \in [p]$, 
        $$\sigma_K(\Sigma_W)h_j \leq M_0(j,j) \leq h_j$$
    \end{lemma}
    \begin{proof}
        Note that 
        $$M_0(j,j) = \frac{1}{n}\sum_{i=1}^n [D_0]_{ji} = \frac{1}{n} \sum_{i=1}^n \sum_{k=1}^K A_{jk}W_{ki} = \sum_{k=1}^K A_{jk}\left(\frac{1}{n}\sum_{i=1}^n W_{ki}\right)$$
        and observe that $h_j := \sum_{k=1}^K A_{jk}$ and for each $k \in [K]$ (recall $\Sigma_W := \frac{1}{n}WW^T$),
        $$\sigma_K(\Sigma_W) \leq \Sigma_{W}(k,k) = \frac{1}{n}\sum_{i=1}^n W_{ki}^2\leq \frac{1}{n}\sum_{i=1}^n W_{ki} \leq 1$$
    \end{proof}
    \begin{theorem}\label{lem:A3}
        Define $p_n := p\vee n$, $\tau_n := \sqrt \frac{\log p_n}{Nn}$. Let 
        $$J := \left\{j \in [p]: M(j,j) > \alpha \tau_n\right\}
        , \quad \quad J_\pm := \{j \in [p]: M_0(j,j) > \alpha_\pm \alpha \tau_n\} $$
        for some suitably chosen $\alpha > 0$ and $0 < \alpha_+ < 1 < \alpha_-$. The following statements hold:
        \begin{enumerate}[(a)]
            \item For a fixed $j \in [p]$, we have 
            $$\bP(|M(j,j) - M_0(j,j)| \geq t) \leq 2\exp\left(-nNt^2/2\right) $$
            \item The event 
            $$\cE := \{J_- \subseteq J \subseteq J_+\}$$ occurs with probability at least $1-o(p_n^{-1})$. Here, we can select $\alpha = 8, \alpha_+ = \frac{1}{2}, \alpha_- = 2$. Note that this implies that on $\cE$, $\min_{j \in J} h_j > \alpha_+\alpha \tau_n$.
            \item We have  
            $$\|A_{J_-^c*}\|_F^2 \leq \left[2K(\beta\tau_n)^{-1}\wedge p\right](\beta\tau_n)^2 = o(1) $$ where $\beta := \frac{\alpha_-\alpha}{\sigma_K(\Sigma_W)} = \frac{16}{\sigma_K(\Sigma_W)} \leq C$, and the same is true of $\|A_{J^c*}\|_F^2$ on event $\cE$.
            \item $|J_+| \leq \frac{K\tau_n^{-1}}{\alpha\alpha_+} \wedge p = \frac{K\tau_n^{-1}}{4} \wedge p$, and the same is true of $|J|$ on event $\cE$. 
            \item $\sigma_K(A_{J*}) > c\sqrt K$ for some absolute constant $c > 0$ on event $\mathcal{E}$. This implies $|J| \geq K$ and $[D_0]_{J*}$ have rank $K$ on $\cE$. 
            \item There exists an absolute constant $c\in (0,1)$ such that the entries of $A_{J*}^TA_{J*}$ are all greater than $c$ on event $\cE$. 
        \end{enumerate}
    \end{theorem}
    \begin{proof}
        \begin{enumerate}[(a)]
            \item Denote $Z := D - D_0$. We introduce the set of $p$-dimensional one-hot vectors 
            $$\{T_{im}: 1 \leq i \leq n, 1 \leq m \leq N \}$$ 
            for each word in the dataset; note that $T_{im} \sim \text{Multinomial}(1, [D_0]_{*i})$ and these one-hot vectors are mutually independent. It follows that each column of $Z$ satisfies 
            \begin{equation}\label{eq:Z-expression} 
                [Z]_{*i} = \frac{1}{N} \sum_{m=1}^{N}(T_{im} - \bE[T_{im}])   
            \end{equation}
            
            Note that for a given $j \in [p]$: 
            \begin{equation}\label{eq:M1}
                M(j,j) - M_0(j,j) = \frac{1}{n}\sum_{i=1}^n Z_{ji} = \frac{1}{nN}\sum_{i=1}^n \sum_{m=1}^{N}(T_{im}(j)-\bE[T_{im}(j)])
            \end{equation}
            and since $|T_{im}(j) - \bE[T_{im}(j)]| \leq 1$, we can apply Hoeffding's inequality to conclude
            $$\bP(|M(j,j) - M_0(j,j)| \geq t) \leq 2\exp\left(-nNt^2/2\right) $$
            \item  Note that $\alpha_- > 1$. We have
            \begin{align*}
\bP(J_- \not\subseteq J) &= \bP\left(\cup_{j\in J_-}\left\{M(j,j)\leq \alpha \tau_n\right\}\right) \\ 
&\leq \sum_{j\in J_-}\bP\left(M(j,j) \leq \alpha \tau_n\right)\\ 
&\leq \sum_{j\in J_-}\bP\left(M(j,j)-M_0(j,j) \leq \alpha \tau_n - \alpha_- \alpha \tau_n\right)\\
&\leq \sum_{j\in J_-}\bP\left(|M(j,j) -M_0(j,j)|\geq (\alpha_--1)\alpha\tau_n\right)\\ 
&\leq\sum_{j\in J_-}2\exp\left(-Nn(\alpha_--1)^2\alpha^2\tau_n^2/2\right) \\ 
&\leq 2p_n^{1-(\alpha_--1)^2\alpha^2/2}
\end{align*}
where in the last step we used $|J_-| \leq p$. We want to choose $1 - \frac{(\alpha_- - 1)^2\alpha^2}{2} < -1$ or equivalently $(\alpha_- - 1)\alpha > 2$. 

    Note that $0 < \alpha_+ < 1$. We further have
        \begin{align*}
\bP(J \not\subseteq J_+) &= \bP\left(\cup_{j\in J_+^c}\left\{M(j,j) > \alpha \tau_n\right\}\right) \\
&\leq \sum_{j\in J_+^c}\bP\left(M(j,j) - M_0(j,j) > (\alpha - \alpha \alpha_+)\tau_n\right)  \\
&\leq \sum_{j\in J_+^c}\bP\left(|M(j,j)-M_0(j,j)| > (1-\alpha_+)\alpha \sqrt{\frac{\log p_n}{nN}}\right)\\ 
&\leq \sum_{j\in J_+^c} 2\exp\left(-\frac{Nn(1-\alpha_+)^2\alpha^2\log p_n}{2Nn}\right) \\ 
&\leq 2p_n^{1-(1-\alpha_+)^2\alpha^2/2}\end{align*}
Again, we want to choose $1 - \frac{(1-\alpha_+)^2\alpha^2}{2} < -1$ or equivalently $(1-\alpha_+) \alpha > 2$. A suitable choice is $\alpha = 8, \alpha_+ = \frac{1}{2}, \alpha_- = 2$.
\item From Lemma \ref{lem:A2}, we have 
$$ M_0(j,j) \geq \sigma_K(\Sigma_W) \|A_{j*}\|_1 \geq \sigma_K(\Sigma_W) \|A_{j*}\|_2$$
and so if we define $$L := \{j \in [p]: \|A_{j*}\|_2 > \beta \tau_n\} \quad \text{where } \beta := \frac{\alpha_-\alpha}{\sigma_K(\Sigma_W)} = \frac{16}{\sigma_K(\Sigma_W)}$$ then $L \subseteq J_-$ by definition of $J_-$, and thus $\|A_{J^c_-*}\|_F \leq \|A_{L^c*}\|_F$. Now, if we order the $\ell_2$ row norms $\|A_{(1)*}\|_2 \geq \dots \geq \|A_{(p)*}\|_2$ and apply Lemma \ref{lem:A1},
\begin{align*}
\|A_{L^c*}\|_F^2 &= \sum_{j\not\in L}\|A_{j*}\|_2^2 = \sum_{j\not\in L}\min(\|A_{j*}\|_2^2, \beta^2\tau_n^2) \\
&\leq \sum_{j=1}^p \min(\|A_{(j)*}\|_2^2, \beta^2\tau_n^2) \leq \sum_{j=1}^p \min\left(\frac{K^2}{j^2}, \beta^2\tau_n^2\right)\\ 
&\leq \int_0^\infty \min(\beta^2\tau_n^2, K^2t^{-2})dt
\end{align*}
Let $t_0$ satisfies $\beta^2\tau_n^2 = K^2t^{-2}$ or $t_0 = \frac{K}{\beta\tau_n}$. We continue: 
\begin{align*}
\|A_{L^c*}\|_F^2 &\leq  t_0 \beta^2\tau_n^2 + K^2\int_{t_0}^\infty t^{-2}dt \\ 
&= t_0\beta^2\tau_n^2 + K^2t_0^{-1} = 2t_0\beta^2\tau_n^2 \\
&= 2K\beta\tau_n = 2K\beta \sqrt{\frac{\log p_n}{Nn}} = o(1)
\end{align*}
given our assumption that $\sigma_K(\Sigma_W) > c$ for some absolute constant $c > 0$. Moreover, it is also clear from the definition of $L$ that 
$$\|A_{L^c*}\|_F^2 \leq p(\beta\tau_n)^2 $$
\item For all $j \in J_+ := \{j \in [p]: M_0(j,j) > \alpha\alpha_+\tau_n\}$, note that $h_j \geq M_0(j,j) > \alpha_+ \alpha \tau_n $. Then observe that 
$$K = \sum_{j=1}^p h_j \geq \sum_{j \in J_+} h_j \geq |J_+|\alpha\alpha_+\tau_n = 4|J_+|\tau_n$$
\item Here we use the assumption that $\sigma_K(A) > c\sqrt K$ for some absolute constant $c > 0$. Observe that by Weyl's inequality for singular values, on event $\cE$ we have
\begin{align*}
    \sigma_K(A_{J*}) &\geq \sigma_K(A) - \|A_{J^c*}\|_\op \\ 
    &\geq \sigma_K(A) - \|A_{J^c*}\|_F\\ 
    &\geq c\sqrt{K} - o(1) \geq c\sqrt K/2
\end{align*} when $nN$ is sufficiently large, since in part (c) we have shown $\|A_{J_-^c*}\|_F \leq C\left(\frac{\log p_n}{Nn}\right)^{1/4}$. Hence, $A_{J*}$ has rank $K$ on $\cE$, and by Sylvester's rank inequality, 
$$K = \text{rank}(A_{J*}) + \text{rank}(W) - K \leq \text{rank}([D_0]_{J*}) \leq \text{rank}(A_{J*}) = K $$
\item Here we use the assumption that the entries of $A^TA$ are bounded below by an absolute constant. For any $k, l \in [K]$, since $A^TA = A_{J*}^TA_{J*} + A_{J^c*}^TA_{J^c*}$, on event $\cE$ the $(k,l)$-entry of $A_{J*}^TA_{J*}$ satisfies 
\begin{align*}
    (A_{J*}^TA_{J*})(k,l) &= (A^TA)(k,l) - \sum_{j \not \in J} A_{jk}A_{jl} \\ 
    &\geq c - \|A_{J^ck}\|_2 \|A_{J^cl}\|_2 \\ 
    &\geq c - \|A_{J^c*}\|_F^2 = c-o(1) \geq c/2
\end{align*}
when $nN$ is sufficiently large.
        \end{enumerate}
            \end{proof}
    
\section{Properties of unobserved quantities}
    \begin{lemma}\label{lem:B1}
        The following statements are true: 
        \begin{enumerate}[(a)]
            \item $\sigma_1(A) \leq \sqrt{K}$ and $\sigma_1(\Sigma_W)\leq 1$, where $\Sigma_W := \frac{1}{n}WW^T$.
            \item If $\Xi \in \mathbb{R}^{|J|\times K}$ contains the first $K$ left singular vectors of $[D_0]_{J*}$, then $\Xi^T A_{J*}$ is invertible. If $V := (\Xi^TA_{J*})^{-1} \in \mathbb{R}^{K \times K}$ then $V$ satisfies the following: 
            \begin{enumerate}[(i)]
                \item $\Xi = A_{J*}V$
                \item The singular values of $V$ are the inverses of the singular values of $A_{J*}$
                \item The columns $V_1, \dots, V_K$ of $V$ are eigenvectors of the matrix $\Theta := \Sigma_W A_{J*}^TA_{J*}$, associated with the eigenvalues
                $$\lambda_k(\Theta) = \frac{\sigma_k^2([D_0]_{J*})}{n}\quad\text{ for } 1 \leq k \leq K$$
            \end{enumerate}
            \item The matrix $\Theta_0 := \Sigma_W A^TA \in \mathbb{R}^{K\times K}$ satisfies the following: 
            \begin{enumerate}[(i)]
                \item The entries of $\Theta_0$ are all positive and bounded below by an absolute constant $c_1 > 0$.
                \item The gap between its first two eigenvalues is bounded below by an absolute constant $c_2 > 0$.
                \item The entries of the unit-norm leading positive eigenvector of $\Theta_0$ are all bounded below by an absolute constant $c_3 > 0$.
            \end{enumerate}
        \item On event $\cE$, the results of part (c) also apply to $\Theta$, possibly with smaller absolute constants $c_1, c_2, c_3 > 0$. 
         \item There exist absolute constants $c, C > 0$ such that on $\cE$, the entries of the first column of $V$ satisfy
         $$ \frac{c}{\sqrt{K}} \leq \min_{k\in [K]} V_1(k) \leq \max_{k \in [K]}V_1(k) \leq \frac{C}{\sqrt{K}}$$
         and if $\xi_1, \dots, \xi_K$ are the columns of $\Xi$, then for any $j \in J$, its first column satisfies
         $$\frac{ch_j}{\sqrt{K}} \leq \xi_1(j) \leq \frac{Ch_j}{\sqrt{K}}$$
        \item Let $Q \in \mathbb{R}^{K \times K}$ be defined by $Q^T = [\text{\diag}(V_1)]^{-1}V$, and note that the entries of the first row of $Q$ are all equal to 1. If $v_1^*, \dots, v_K^* \in \mathbb{R}^{K-1}$ are defined by the relation 
        $$ Q = \begin{pmatrix} 1 & \dots & 1 \\ v_1^* &\dots &v_K^*\end{pmatrix}$$
        then we have
        $$c \leq \sigma_K(Q) \leq \sigma_1(Q) \leq C$$ 
        Consequently, $v_1^*, \dots, v_K^*$ are affinely independent (which means the simplex defined by their convex hull is non-degenerate) and $\max_{k\in [K]} \|v_k^*\|_2 \leq C$. 
        \end{enumerate}
    \end{lemma}
\begin{proof}
    \begin{enumerate}[(a)]
        \item The $k^\text{th}$ diagonal entry of $A^TA$ is $\|A_{*k}\|_2^2 \leq \|A_{*k}\|_1 = 1$, so $\text{tr}(A^TA) \leq K$ which implies $\sigma_1(A) \leq \sqrt{K}$. Similarly, 
        $$\sigma_1(\Sigma_W) = \frac{\sigma_1(W^TW)}{n} \leq \frac{\text{tr}(W^TW)}{n} = \frac{\sum_{i=1}^n \|W_{*i}\|_2^2}{n} \leq \frac{\sum_{i=1}^n \|W_{*i}\|_1}{n} = 1  $$
        \item  By singular value decomposition, we have 
        $$[D_0]_\Js  = \Xi \Lambda B^T$$
        where $\Lambda = \diag(\sigma_1, \dots, \sigma_K) \in \mathbb{R}^{K \times K}$ contains the singular values of $[D_0]_{J*}$ and $B \in \mathbb{R}^{n\times K}$ contains its right singular vectors. Here $\Xi^T \Xi = B^TB = I_K$. Then 
        $$\Xi = \Xi \Lambda B^T B\Lambda^{-1} = [D_0]_\Js B\Lambda^{-1} = A_{J*}WB\Lambda^{-1}$$
        
        If we let $V = WB\Lambda^{-1} \in \mathbb{R}^{K \times K}$, then $\Xi = A_{J*} V$. Furthermore, since 
        $$\Xi^T \Xi = \Xi^T A_{J*} V = I_K $$
        we can see that $V$ can be defined as the inverse of $\Xi^T A_{J*}$, thus proving (i). Also, since 
        $$\Xi^T \Xi = V^TA_\Js^T A_\Js V = I_K $$
        we have $VV^TA^T_{J*}A_\Js VV^T = VV^T$, which implies $VV^T = (A_{J*}^TA_{J*})^{-1}$ and (ii) follows. 

        Now let $\xi_1, \dots, \xi_K$ be the columns of $\Xi$, and let $V_1, \dots, V_K$ be the columns of $V$. Note that $\xi_k = A_{J*} V_k$. Since $[D_0]_{J*}[D_0]_{J*}^T \xi_k = \sigma_k^2\xi_k$ and $\Sigma_W := \frac{1}{n} WW^T$, we have 
        $$A_{J*}\Sigma_W A_{J*}^T A_{J*} V_k = A_{J*}\Sigma_W A_{\Js}^T\xi_k = \frac{1}{n} [D_0]_{J*}[D_0]_{J*}^T\xi_k =  \frac{\sigma_k^2}{n}\xi_k = \frac{\sigma_k^2}{n}A_{J*}V_k$$

        Multiplying both sides by $(A_{J*}^TA_{J*})^{-1}A_{J*}^T$ on the left, we have 
        \begin{equation*}
            \Sigma_W A_{J*}^TA_{J*}V_k = \frac{\sigma_k^2}{n} V_k
        \end{equation*}
        and so $V_1, \dots, V_K$ are eigenvectors (not necessarily orthonormal) of $\Theta := \Sigma_W A_{J*}^TA_{J*}$, associated with eigenvalues $\sigma^2_k/n$ for $k = 1, \dots, K$. This proves (iii). 
        \item For any $1 \leq k, l \leq K$, (i) follows from our assumptions:
            \begin{align*}
            \Theta_0(k,l) &= \sum_{s=1}^K \Sigma_W(k,s)\cdot (A^TA)(s,l)\\ 
            &\geq \min_{t,u\in [K]}(A^TA)(t,u) \cdot \sum_{s=1}^K \Sigma_W(k,s) \\ 
            &\geq \min_{t,u\in [K]}(A^TA)(t,u)\cdot \Sigma_W(k,k) \\ 
            &\geq \min_{t,u\in [K]}(A^TA)(t,u)\cdot \sigma_K(\Sigma_W) > c
        \end{align*}
        Let $\gamma(\Theta_0) := \lambda_1(\Theta_0) - \lambda_2(\Theta_0) \geq 0$ denote the gap between the first two eigenvalues of $\Theta_0$. The proof of (ii) is an asymptotic argument. If we consider a sequence $\{\Theta_0^{(n)}\}$ that varies with $n$ as $n \to \infty$,  then (ii) follows if we can establish that 
        $$\liminf_{n\to \infty} \gamma(\Theta_0^{(n)}) > 0$$
        Assume to the contrary that $\liminf_{n\to \infty} \gamma(\Theta_0^{(n)}) = 0$. Then there exists a subsequence $\{\Theta_0^{(n_m)}\}_{m=1}^\infty$ such that the gap between the first two eigenvalues decays to zero. Since $$\|\Theta_0^{(n)}\|_\op \leq \|\Sigma_W^{(n)}\|_\op\|A^{(n)}\|_\op ^2 \leq K$$
        and $K$ is fixed as $n$ varies, there must exist a further subsequence that converges to some matrix $\Theta_0^{(\infty)}$. By part (i), this matrix $\Theta_0^{(\infty)}$ has entries that are bounded below by some absolute constant $c > 0$, and yet its first two eigenvalues are equal (by eigenvalue continuity). By Perron's theorem (see Section 8.2 of \cite{horn2012matrix} for a reference), such a matrix $\Theta_0^{(\infty)}$ cannot exist.

        (iii) is also proven in a similar manner. Let $\eta_0^{(n)} \in \mathbb{R}^K$ denote the leading unit-norm positive eigenvector of $\Theta_0^{(n)}$; its entries are all positive by Perron's theorem. Suppose there exists some $k \in [K]$ such that 
        $$\liminf_{n\to \infty} \eta_0^{(n)}(k) = 0$$
        Note that the mapping from a matrix in $\mathbb{R}^{K \times K}$ with strictly positive entries to its leading unit-norm positive eigenvector is continuous (this will be further elaborated in part (d)). Again, this implies that there exists a subsequence $\{\Theta_0^{(n_m)}\}$ that converges to some $\Theta_0^{(\infty)}$ having strictly positive entries, and yet its leading eigenvector contains a zero entry. This contradicts Perron's theorem. 
        \item In light of Theorem \ref{lem:A3}(f) which shows $A_{J*}^TA_{J*}$ has entries bounded below by $c>0$ on $\cE$, (i) is proven similarly as in part (c). 

        We will first show (iii). Note that we refrain from applying the asymptotic arguments of part (c) directly to $\Theta$ since, unlike $\Theta_0$, $\Theta$ depends on $J$ which is random. Also, the $\sin\theta$ theorem is not applicable to eigenvectors of $\Theta$ and $\Theta_0$ as these matrices are not symmetric. Hence, we opt for the approach presented below. 
        
        Define the open domain 
        $$E = \{\Psi \in \mathbb{R}^{K \times K}: \Psi(k,l) > 0 \text{ for all } k, l \in [K] \} $$
        and define $\mathbf{f}: E \to \mathbb{R}^K$ as the function mapping a matrix in $E$ to its leading unit-norm positive eigenvector. Also, fix $\Psi_0 \in E$ and $1 \leq k,l \leq K$. For any real-valued $t$ in a neighborhood of zero, consider the function 
        $$f_{kl}^{\Psi_0}(t) := \mathbf{f}(\Psi_0 + t\mathbf{e}_k\mathbf{e}_l^T)$$
        where $\mathbf{e}_k$ and $\mathbf{e}_l$ are the $k^\text{th}$ and $l^\text{th}$ canonical basis vectors of $\mathbb{R}^k$ respectively. 
        
        Since the algebraic multiplicity of the first eigenvalue of any matrix in $E$ is 1 (Perron's theorem), by Theorem 2 of \cite{greenbaum2020first}, for any $\Psi_0 \in E$ and any $k,l \in [K]$, the function $f_{kl}^{\Psi_0}(\cdot)$ is continuously differentiable around 0 (more specifically, one can write $f_{kl}^{\Psi_0}(t) = \frac{x(t)}{\|x(t)\|_2}$ for some eigenvector function $x(t)$ that is analytic in a neighborhood of 0). Therefore, the function $\mathbf{f}$ itself is continuously differentiable on $E$, and we can define its derivative $\mathbf{f}'(\Psi)$ as a matrix in $\mathbb{R}^{K^2 \times K}$ containing all the partial derivatives of $\mathbf{f}$ at $\Psi \in E$. Since these partial derivatives are all continuous, $\mathbf{f}': E \to \mathbb{R}^{K^2 \times K}$ is a continuous function. 
        
        Now if $c > 0$ is an absolute constant such that all the entries of $\Theta_0$ and $\Theta$ are greater than $c$ (the latter on event $\cE$), then $\Theta$ and $\Theta_0$ belong to the set 
        $$E' = \{\Psi \in \mathbb{R}^{K \times K}: \Psi(k,l) \geq c \text{ for all } k,l \in [K] \text{ and } \|\Psi\|_\op \leq K\} $$
        which is a compact subset of $E$. Let $\eta$ and $\eta_0$ be the unit-norm positive first eigenvectors of $\Theta$ and $\Theta_0$ respectively. On event $\cE$, by Theorem 9.19 of \cite{rudin1976principles}, 
        \begin{align*}
            \|\eta-\eta_0\|_2 = \|\mathbf{f}(\Theta) - \mathbf{f}(\Theta_0)\|_2 &\leq \max_{\Psi \in E'} \|\mathbf{f}'(\Psi)\|_\op \|\Theta-\Theta_0\|_F  \\ 
            &\leq C\|\Sigma_W\|_\op \|A^TA - A_{J*}^TA_{J*}\|_F \\ 
            &= C \|\Sigma_W\|_\op\|A_{J^c*}^TA_{J^c*}\|_F \leq C \|A_{J^c*}\|_F^2 = o(1)
        \end{align*}
        where we note that $A^TA = A_{J*}^TA_{J*} + A_{J^c*}^TA_{J^c*}$. Hence, for any $k \in [K]$, 
        $$ \eta(k) \geq \eta_0(k) - \|\eta - \eta_0\|_2 \geq c - o(1) > c/2$$
        if $nN$ is sufficiently large. We have shown $\min_k \eta(k) \geq c/2 > 0$ on $\cE$. 

        As for (ii), we have shown in (b)(iii) for $\Theta$ (and the proof is similar for $\Theta_0$) that
        $$\lambda_k(\Theta) = \frac{\sigma_k^2([D_0]_{J*})}{n}, \quad \lambda_k(\Theta_0) = \frac{\sigma_k^2(D_0)}{n} $$
        Note that since $\|A\|_\op \leq \sqrt{K}$ and $\|W\|_\op \leq \sqrt{n}$, 
        $$\max[\sigma_k([D_0]_{J*}), \sigma_k(D_0)] \leq \|A\|_\op \|W\|_\op \leq \sqrt{Kn} $$
        and by Weyl's inequality for singular values (which can be applied after appending zero rows to the matrix $A_{J*}$ so as to match the dimension of $A$),
        \begin{align*}
            |\lambda_k(\Theta) - \lambda_k(\Theta_0)| &\leq \frac{|\sigma_k([D_0]_{J*}) - \sigma_k(D_0)||\sigma_k([D_0]_{J*}) + \sigma_k(D_0)|}{n}\\ 
            &\leq \frac{\|A_{J^c*}\|_\op\|W\|_\op(2\sqrt{Kn})}{n} \\ 
            &\leq 2\sqrt{K}\|A_{J^c*}\|_\op= o(1)
        \end{align*}
        on event $\cE$, so 
        \begin{align*}
            |\lambda_1(\Theta) - \lambda_2(\Theta)| &\geq |\lambda_1(\Theta_0) - \lambda_2(\Theta_0)| - o(1)\\ 
 &\geq c - o(1) \geq c/2
        \end{align*}
        if $nN$ is sufficiently large, for some absolute constant $c > 0$. 
        \item Since we assume $\sigma_K(A) \geq c\sqrt{K}$ for some $c \in (0,1)$, 
        $$\max_{k\in [K]}V_1(k) \leq \|V_1\|_2 \leq \sigma_1(V) = \sigma_K^{-1}(A_{J*}) \leq \sigma_K^{-1}(A) \leq \frac{C}{\sqrt{K}} $$
        and since 
        $$\|V_1\|_2 \geq \sigma_K(V) = \sigma_1^{-1}(A_{J*}) \geq \sigma_1^{-1}(A) \geq \frac{1}{\sqrt{K}}$$ and $\frac{1}{\|V_1\|_2} V_1$ is the unit-norm leading positive eigenvector of $\Theta$, on event $\cE$ we have
        $$\min_{k \in [K]}V_1(k) = \|V_1\|_2 \min_{k \in [K]}\left\{\frac{V_1(k)}{\|V_1\|_2}\right\} \geq \frac{c}{\sqrt{K}}  $$
        Since $\xi_1 = A_{J*}V_1$, it follows that on event $\cE$, for any $j \in J$, 
        $$\frac{ch_j}{\sqrt{K}}\leq \xi_1(j) \leq \frac{Ch_j}{\sqrt{K}} $$
        \item It can be seen by the definition of $Q$ that 
        $$\sigma_K(Q) \geq \frac{\sigma_K(V)}{\max_{k\in [K]}V_1(k)} \geq c > 0 $$
        and
        $$\sigma_1(Q) \leq \frac{\sigma_1(V)}{\min_{k\in [K]}V_1(k)} \leq C$$
        for some $c, C >0$. Thus, $\max_{k\in [K]}\|v_k^*\|_2 \leq C$ and $Q$ has independent columns, which implies $v_1^*, \dots, v_K^*$ are affinely independent. 
    \end{enumerate}
\end{proof}
\begin{lemma}\label{lem:B2}
    Let $A_{J1}, \dots, A_{JK}$ be the columns of $A_{J*}$. Under Assumption \ref{ass:VH} on the vertex hunting function $\mathcal{V}(\cdot)$, the oracle procedure in Definition \eqref{def: oracle-procedure} returns 
    \begin{equation}\label{eq:oracle}
        \tilde{A}_{J*} = A_{J*} \cdot \diag(\|A_{J1}\|_1^{-1}, \dots, \|A_{JK}\|_1^{-1})
    \end{equation}
on event $\cE$. 
\end{lemma}
\begin{proof}
    Note that from Lemma \ref{lem:B1}(b)(i), we have $\Xi = A_{J*}V$. Let $\One_J$ be the vector of size $|J|$ with entries all equal to 1. Now, by the definition of $R$, 
    $$[\One_J, R] = [\text{diag}(\xi_1)]^{-1} \Xi = [\text{diag}(\xi_1)]^{-1} A_{J*}V$$
    Recall from Lemma \ref{lem:B1}(f) the definition $Q^T := [\text{diag}(V_1)]^{-1}V = \begin{pmatrix}
        1 & \dots & 1 \\ v_1^* & \dots & v_K^*
    \end{pmatrix}^T$. Then
    \begin{equation}\label{eq:Pi-1}
        [\One_J, R] = [\text{diag}(\xi_1)]^{-1}A_{J*}\cdot\text{diag}(V_1) Q^T  = \Pi \begin{pmatrix}
        1 & \dots & 1 \\ v_1^* & \dots & v_K^*
    \end{pmatrix}^T
    \end{equation}
    where $\Pi$ is defined as follows:
    \begin{equation}\label{eq: Pi-def}
        \Pi := [\text{diag}(\xi_1)]^{-1}A_{J*}\cdot\text{diag}(V_1) \in \mathbb{R}^{|J| \times K}
    \end{equation}
    
    From \eqref{eq:Pi-1} and \eqref{eq: Pi-def}, we can see that $\Pi$ contains only non-negative entries and the rows of $\Pi$ sum up to 1. This means the rows of $R$ (the point cloud) lie inside the convex hull of simplex vertices $\{v_1^*, \dots, v_K^*\} \subseteq \mathbb{R}^{K-1}$. 

    By Assumption \ref{ass: separability}, for each topic there exists at least an anchor word for that topic in the set $J$ on event $\mathcal{E}$. This means that the point cloud contains at least one point on each vertex $v_1^*, \dots , v_K^*$. By Assumption \ref{ass:VH}, the vertex hunting procedure $\mathcal{V}(\cdot)$ returns precisely the vertices $v_1^*, \dots, v_K^*$. Now let $\{\pi_j: j \in J\} \subseteq{R}^K$ denote the rows of $\Pi$. From taking the transpose of \eqref{eq:Pi-1}, $\Pi$ is then estimated correctly by solving
    $$\begin{pmatrix}
            1 & \dots & 1 \\ v_1^* &\dots & v_K^*            
        \end{pmatrix} \pi_j = \begin{pmatrix}
            1 \\ r_j
        \end{pmatrix} $$
    Now, by the definition of $\Pi$ in \eqref{eq: Pi-def}, 
    \begin{equation}\label{lem:B2-result}
        \diag(\xi_1)\cdot\Pi = A_{J*}\cdot \diag(V_1)
    \end{equation}
    and thus \eqref{eq:oracle} follows if we normalize $\diag(\xi_1)\cdot\Pi$ to ensure its columns sum up to 1. 
\end{proof}

\section{Concentration inequalities involving \texorpdfstring{$Z = D - D_0$}{Z}}\label{section: C}
\begin{remark}\label{rem: C1}{\rm 
    This section contains all the concentration inequalities necessary for our analysis, and is comparable to Section E in the appendix of \cite{ke2022using}. 
    
    Lemma \ref{lem:E1} and Lemma \ref{lem:E2} are similar to Lemmas E.1 and E.2 of \cite{ke2022using} in that they are simple applications of Bernstein's inequality. However, it is crucial to note that our results are applicable even when $\min_{j \in [p]} h_j$ is extremely small, as we only restrict our attention to $j \in J_+$ (where $J_+$ is defined in Section A). In contrast, Lemmas E.1 and E.2 of \cite{ke2022using} require $\min_{j\in [p]} h_j \geq cK/p$ (or at least $\min_{j \in [p]} h_j \gg (Nn)^{-1}\log n$). 

    Lemma \ref{lem: C4} in our paper is based on standard techniques for deriving concentration inequalities for U-statistics. Our results here can be compared to Lemmas E.3-E.6 of \cite{ke2022using}, which use a truncation argument and the fact that the product of two sub-Gaussian variables is sub-exponential. Our bounds do not depend on $p$ except for log factors and are applicable to all parameter regimes (in particular when $ p \gg n \vee N$), whereas the bounds in Lemmas E.3-E.6 \cite{ke2022using} depend heavily on $p$ and $\min_{j \in [p]} h_j$.}
\end{remark}
\begin{lemma}[Bernstein's inequality] Let $X_1, \dots, X_n$ be independent random variables with $\mathbb{E}(X_i) = 0$ and $\text{Var}(X_i) \leq \sigma_i^2$ for all $i$. Let $\sigma^2 := n^{-1}\sum_{i=1}^n \sigma_i^2$. Then for any $t > 0$, 
$$ \bP\left(n^{-1}\left|\sum_{i=1}^n X_i\right| \geq t \right) \leq 2 \exp\left(-\frac{nt^2/2}{\sigma^2 + bt/3}\right)$$
    
\end{lemma}
\begin{lemma}\label{lem:E1} Denote $\tilde{h}_j := h_j \wedge 1$. With probability at least $1 - o(p_n^{-1})$, 
\begin{equation}\label{eq: M-error}
|M(j,j) - M_0(j,j)| \leq C^* \sqrt{\frac{\tilde{h}_j \log p_n}{nN}} \quad \text{ for all } j \in J_+
\end{equation}
\end{lemma}
\begin{proof}
    Similar to \eqref{eq:M1}, for a fixed $j \in J_+$ we have
    $$M(j,j) - M_0(j,j) = \frac{1}{n}\sum_{i=1}^n Z_{ji} = \frac{1}{nN}\sum_{i=1}^n \sum_{m=1}^{N}(T_{im}(j)-\bE[T_{im}(j)])$$
Note that since $T_{im}(j) \sim \text{Bernoulli}(D_0(j,i))$, $|T_{im}(j)- \bE[T_{im}(j)]| \leq 1$ and 
\begin{equation}\label{eq:h_j}
    \text{Var}(T_{im}(j)) \leq D_0(j,i) = \sum_{k=1}^K A_{jk}W_{ki} \leq \sum_{k=1}^KA_{jk} = h_j
\end{equation}
(and also $\text{Var}(T_{im}(j)) \leq 1$). We apply Bernstein's inequality to conclude for any $t > 0$:
$$\bP\left(|M(j,j) - M_0(j,j)| \geq t\right) \leq 2\exp\left(-\frac{nNt^2/2}{\tilde{h}_j + t/3}\right) $$
One can choose $t = C^*\sqrt{\frac{\tilde{h}_j \log p_n}{nN}}$ or $t = \frac{C^*\log p_n}{nN}$ depending on whether $\tilde{h}_j \geq \frac{\log p_n}{nN}$ holds. Thus with probability at least $1- o(p_n^{-2})$, 
\begin{align*}
    |M(j,j) - M_0(j,j)| &\leq C^*\max\left(\sqrt{\frac{\tilde{h}_j\log p_n}{nN}}, \frac{\log p_n}{nN}\right) \\ 
    &\leq C^*\sqrt{\frac{\tilde{h}_j \log p_n}{nN}}
\end{align*}
since if $j\in J_+$, then $\tilde{h}_j > \alpha_+\alpha \sqrt{\frac{\log p_n}{nN}} \geq \frac{c^*\log p_n}{nN}$ when $nN$ is sufficiently large so that $\frac{\log p_n}{nN} \leq 1$. We then take union bound over $j \in J_+$. 
\end{proof}
\begin{lemma}\label{lem:E2} Denote $\{Z_j: j \in J_+\} \subseteq \mathbb{R}^n$ as the rows of $Z$ in $J_+$, and $\{W_k: k \in [K]\} \subseteq \mathbb{R}^n$ as the rows of $W$. With probability at least $1 - o(p_n^{-1})$, 
\begin{equation}\label{eq:E.2}
    |Z_j^T W_k| \leq C^*\sqrt{\frac{n\tilde{h}_j\log p_n}{N}} \quad \text{for all } j \in J_+ \text{ and } k\in [K]
\end{equation}
\end{lemma}
\begin{proof}
    Note that 
    \begin{equation}\label{eq:small_z}
    Z_{ji} = \frac{1}{N}\sum_{m=1}^{N} (T_{im}(j) - \bE[T_{im}(j)])  
    \end{equation}
    and so for any $j \in J_+$ and $ k \in [K]$,
    $$Z_j^TW_k = \sum_{i=1}^nZ_{ji}W_{ki} = \frac{1}{nN}\sum_{i=1}^n\sum_{m=1}^{N} nW_{ki}(T_{im}(j)-\mathbb{E}[T_{im}(j)])$$
    We note that $|nW_{ki}(T_{im}(j) - \mathbb{E}[T_{im}(j)])| \leq n$ and $\text{Var}[nW_{ki}(T_{im}(j)-\bE[T_{im}(j)])] \leq n^2 \tilde{h}_j $, so by Bernstein inequality, for any $t > 0$, 
    $$\bP(|Z_j^TW_k| > t) \leq 2\exp\left(-\frac{nNt^2/2}{n^2\tilde{h}_j + nt/3}\right) $$
    We can let $t = C^*\sqrt{\frac{n\tilde{h}_j\log p_n}{N}}$ and again by noting that $\tilde{h}_j \geq \alpha_+\alpha \sqrt{\frac{\log p_n}{nN}} \geq c^* \frac{\log p_n}{nN}$ if $j \in J_+$, we obtain \eqref{eq:E.2}.  
    
\end{proof}
\begin{lemma}\label{lem: C4}
    With probability at least $1-o(p_n^{-1})$, 
    \begin{equation}\label{eq:7}
        |Z_j^TZ_l - \bE(Z_j^TZ_l)| \leq C^*\sqrt{\frac{n\tilde{h}_j\tilde{h}_l \log p_n}{N}} \quad \text{ for all } j,l \in J_+ \text{ with } j \neq l
    \end{equation}
    \begin{equation}\label{eq:8}
        |Z_j^TZ_j - \bE(Z_j^TZ_j)| \leq C^*\sqrt{\frac{n\tilde{h}_j^2\log p_n}{N}} + \frac{C^*}{N}\sqrt{\frac{n\tilde{h}_j \log p_n}{N}} \text{ for all } j \in J_+
    \end{equation}
\end{lemma}
\begin{proof}
    Denote $X_{im}(j) := T_{im}(j) - \mathbb{E}[T_{im}(j)]$. Fix $j, l \in J_+$. By \eqref{eq:small_z}, note that 
    \begin{align*}
    Z_j^T Z_l &= \sum_{i=1}^n Z_{ji}Z_{li} =  \frac{1}{N^2}\sum_{i=1}^n  \sum_{m=1}^{N}\sum_{s=1}^{N} X_{im}(j) X_{is}(l) \\ 
    &= \frac{1}{N^2}\sum_{i=1}^n \sum_{m=1}^{N} X_{im}(j)X_{im}(l) + \frac{1}{N^2} \sum_{i=1}^n \sum_{\substack{1 \leq m, s \leq N \\ m\neq s}}X_{im}(j)X_{is}(l)\\ 
    &= \frac{n}{N}V_1 + \frac{N-1}{N}V_2 
    \end{align*}
    where we define 
    $$V_1 := \frac{1}{nN}\sum_{i=1}^n \sum_{m=1}^{N} X_{im}(j)X_{im}(l)$$
    $$V_2 := \frac{1}{N(N-1)}\sum_{i=1}^n \sum_{\substack{1 \leq m, s \leq N \\ m\neq s}}X_{im}(j)X_{is}(l) $$
    Note that $\bE(V_2) = 0$, and we need an upper bound on $|V_1 -\bE(V_1)|$ and  $|V_2|$. We will deal with $V_2$ first. Define $S_{N}$ as the set of permutations on $\{1, \dots, N\}$ and $N' := \lfloor N / 2\rfloor$. Also define 
    $$W_i(X_{i1} ,\dots, X_{iN}) := \frac{1}{N'} \sum_{m=1}^{N'} X_{i, 2m-1}(j) X_{i, 2m}(l) $$
    Then by symmetry (note that the inner sum over $m,s$ in the definition of $V_2$ has $N(N-1)$ summands),
    $$V_2 = \frac{\sum_{i=1}^n\sum_{\pi \in S_{N}}W_i(X_{i,\pi(1)}, \dots, X_{i,\pi(N)})}{N!} $$
    Define, for a given $\pi \in S_{N}$,  
    $$Q_\pi := \sum_{i=1}^nN' W_i(X_{\pi(1)},\dots, X_{\pi(N)}) $$
    so that $N' V_2  = \frac{1}{N!}\sum_{\pi\in S_{N}}Q_\pi$.
    For arbitrary $t, s > 0$, by Markov's inequality and the convexity of the exponential function, 
    $$\bP(N'V_2 \geq t) \leq e^{-st} \bE(e^{sN'V_2}) \leq e^{-st}\frac{\sum_{\pi\in S_{N}}\bE(e^{sQ_\pi})}{N!} $$
    Also, define $Q = Q_\pi$ for $\pi$ being the identity permutation. Observe that 
    $$Q = \sum_{i=1}^n \sum_{m=1}^{N'}Q_{im} \quad \text{where } Q_{im} = X_{i,2m-1}(j) X_{i,2m}(l)$$ so $Q$ is a (double) summation of mutually independent variables. We have $|Q_{im}| \leq 1$, $\bE(Q_{im}) = 0$ and $\bE(Q^2_{im}) \leq \tilde{h}_j\tilde{h}_l$. The rest of the proof for $V_2$ is similar to the standard proof for the usual Bernstein's inequality and one can skip to the conclusion \eqref{eq:Bern}. 

    If we denote $G(x) = \frac{e^x - 1-  x}{x^2}$, observe $G(x)$ is increasing. Hence, 
    \begin{align*}
    \bE(e^{sQ_{im}}) &= \bE\left(1+ sQ_{im} + \frac{s^2Q_{im}^2}{2} +\dots \right)  \\ 
    &= \bE[1 + s^2Q_{im}^2G(sQ_{im})] \\
    &\leq \bE[1 + s^2Q_{im}^2G(s)] \\ 
    &\leq 1 + s^2\tilde{h}_j\tilde{h}_lG(s) \leq e^{s^2\tilde{h}_j\tilde{h}_lG(s)}
    \end{align*}
    Hence, 
    $$e^{-st}\bE(e^{sQ}) = \exp(-st + N'n\tilde{h}_j\tilde{h}_l s^2G(s)) $$
    Since this bound is applicable to all the other $Q_\pi$ and not just $\pi$ being equal to the identity permutation, we have 
    $$\bP(N'V_2 \geq t) \leq \exp(-st + N'n \tilde{h}_j \tilde{h}_ls^2G(s)) = \exp\left(-st + N'n\tilde{h}_j\tilde{h}_l(e^s - 1 - s)\right) $$
    Now we choose $s = \log \left(1 + \frac{t}{N'n\tilde{h}_j\tilde{h}_l}\right) > 0$. Then 
    \begin{align*}
        \bP(N'V_2 \geq t) &\leq \exp\left[-t\log \left(1 + \frac{t}{N'n\tilde{h}_j\tilde{h}_l}\right) + N' n\tilde{h}_j\tilde{h}_l\left(\frac{t}{N'n\tilde{h}_j\tilde{h}_l} - \log\left(1 + \frac{t}{N'n\tilde{h}_j\tilde{h}_l}\right)\right)\right] \\  
        &= \exp\left[-N'n\tilde{h}_j\tilde{h}_lH\left(\frac{t}{N'n\tilde{h}_j\tilde{h}_l}\right)\right]
    \end{align*}
    where we define the function $H(x)= (1+x)\log (1+x) - x$. Note that we have the inequality
    $$H(x) \geq \frac{3x^2}{6+2x} $$
    for all $x > 0$. Hence, 
    $$\bP\left(N'V_2 \geq t\right) \leq \exp\left(-\frac{t^2/2}{N'n\tilde{h}_j\tilde{h}_l + t/3}\right) $$
    or by rescaling, 
    \begin{equation}\label{eq:Bern}
    \bP(N'V_2 \geq N'nt) \leq \exp\left(-\frac{N'nt^2/2}{\tilde{h}_j\tilde{h}_l +t/3}\right)  
    \end{equation}
    We can choose $t^2 = \frac{C^*\tilde{h}_j\tilde{h}_l}{N'n}\log p_n$ and note that $\tilde{h}_j\tilde{h}_l \geq (\alpha_+ \alpha)^2\frac{\log p_n}{nN}$ if $j,l \in J_+$. Hence, with probability $1-o(p_n^{-1})$ (even after taking union bound over $j, l \in J_+$),
    $$|V_2| \leq C^*\sqrt{\frac{n\tilde{h}_j\tilde{h}_l\log p_n}{N}} $$
    As for $V_1$, we can just apply the usual Bernstein's inequality. Let $\mu_{ij} = \bE[T_{im}(j)] = [D_0]_{ji}$ and define $\mu_{il}$ similarly; note $\mu_{ij} \leq \tilde{h}_j$. Since $X_{im}(j) = T_{im}(j) - \mu_{ij}$,
    \begin{equation}\label{eq:XX}
    X_{im}(j) X_{im}(l) = T_{im}(j) T_{im}(l) - \mu_{ij} T_{im}(l)- \mu_{il}T_{im}(j) + \mu_{ij}\mu_{il}  
    \end{equation}
    If $j \neq l$ then $T_{im}(j)T_{im}(l) = 0$ and so 
    \begin{align*}
        \text{Var}[X_{im}(j) X_{im}(l)] &= \text{Var}\left[\mu_{ij}T_{im}(l) + \mu_{il}T_{im}(j)\right] \\
        &\leq \bE[\mu_{ij}T_{im}(l) + \mu_{il}T_{im}(j)]^2 \\ 
        &= \mu_{ij}^2 \mu_{il} + \mu_{il}^2\mu_{ij} = \mu_{ij}\mu_{il}(\mu_{ij}+\mu_{il}) \\ 
        &\leq \mu_{ij}\mu_{il} \leq \tilde{h}_j \tilde{h}_l
    \end{align*}
    since $\mu_{ij} + \mu_{il} \leq 1$. Hence, by Bernstein's inequality,
    $$\bP\left(|V_1 -\bE(V_1)| \geq t\right) \leq 2\exp\left(-\frac{-nNt^2/2}{\tilde{h}_j\tilde{h}_l + t/3}\right)$$
    which is similar to \eqref{eq:Bern}, so we obtain with probability $1-o(p_n^{-1})$ that 
    $$\frac{n}{N}|V_1 - \bE(V_1)| \leq \frac{C^*}{N}\sqrt{\frac{n\tilde{h}_j\tilde{h}_l\log p_n}{N}} \leq C^*\sqrt{\frac{n\tilde{h}_j\tilde{h}_l\log p_n}{N}}$$
    and \eqref{eq:7} is proven. 

    If $j=l$ then since $T_{im}^2(j) = T_{im}(j)$, \eqref{eq:XX} leads to
    $$X^2_{im}(j) = T_{im}(j)(1-2\mu_{ij}) + \mu_{ij}^2 $$
    and since $|1-2\mu_{ij}| \leq 1$ and $\text{Var}(T_{im}(j)) = \mu_{ij}(1-\mu_{ij})$,
    $$\text{Var}[X^2_{im}(j)] \leq \mu_{ij} \leq \tilde{h}_j $$
    and so we obtain \eqref{eq:8} since with probability $1 - o(p_n^{-1})$
    $$\frac{n}{N}|V_1-\bE(V_1)| \leq \frac{C^*}{N}\sqrt{\frac{n\tilde{h}_j \log p_n}{N}} $$
\end{proof}
\begin{corollary} With probability $1-o(p_n^{-1})$, the following statements hold:
\begin{equation}\label{eq:C.4-1}
\|[ZW_k]_{J_+}\|_2 \leq C^* \sqrt{\frac{nK \log p_n}{N}} \text{ for all } 1 \leq k \leq K  
\end{equation}
\begin{equation}\label{eq:C.4-2}
\|[ZZ^T - \bE(ZZ^T)]_{jJ_+}\|_2 \leq C^* \sqrt{\frac{n\tilde{h}_jK\log p_n}{N}} \text{ for all } j \in J_+    
\end{equation}
\begin{equation}\label{eq:C.4-3}
    \|[ZZ^T - \bE(ZZ^T)]_{J_+J_+}\|_F \leq C^*K\sqrt{\frac{n\log p_n}{N}}
\end{equation}
\end{corollary}
\begin{proof}
    This follows from \eqref{eq:E.2}, \eqref{eq:7} and \eqref{eq:8} after squaring the error bounds and summing them up. We note that $\sum_{j \in J_+} \tilde{h}_j \leq \sum_{j=1}^p h_j =  K$.
\end{proof}

\section{Estimation errors for singular vectors and the point cloud}
We will use the following theorem (a row-wise perturbation bound for eigenvectors) from \cite{ke2022using}. 
\begin{lemma}[Lemma F.1 of \cite{ke2022using}]\label{lem:D1}
    Let $B_0$ and $B$ be $m \times m$ symmetric matrices with $\text{rank}(B_0) = K$, and assume $B_0$ is positive semi-definite. For $1 \leq k \leq K$, let $\delta_k^0$ and $\delta_k$ be the $k^\text{th}$ largest eigenvalues of $B_0$ and $B$ respectively, and let $u_k^0$ and $u_k$ be the $k^\text{th}$ eigenvectors of $B_0$ and $B$. Fix $1 \leq s \leq k \leq K$. If for some $c \in (0,1)$, suppose (by default, if $s = 1$ then $\delta_{s-1}^0 - \delta_s^0 = \infty$)
    $$\min(\delta_{s-1}^0 - \delta_s^0, \delta_k^0 - \delta_{k+1}^0, \min_{l \in [K]} \delta_l^0) \geq c\|B_0\|_\op, \quad \|B-B_0\|_\op \leq (c/3)\|B_0\|_\op$$
    Write $U_0 = [u_s^0, \dots, u_k^0], U = [u_s, \dots u_k]$ and $\Xi = [u_1^0, \dots, u_K^0]$. There exists an orthonormal matrix $O$ such that for all $1 \leq j \leq p$, 
    $$\|(UO - U_0)_{j*}\|_2 \leq \frac{5}{c\|B_0\|_\op}(\|B-B_0\|_\op\|\Xi_{j*}\|_2 + \sqrt{K}\|(B-B_0)_{j*}\|_2) $$
\end{lemma}
If we define
$$G := DD^T - \frac{n}{N}M$$ 
$$G_0 := \left(1-\frac{1}{N}\right) D_0D_0^T$$
then the above lemma can be applied to the submatrices $G_{JJ}$ and $[G_0]_{JJ}$ (see Lemma \ref{lem:D8}). 
\begin{lemma}
    With probability $1 - o(p_n^{-1})$, we have
    $$ \|(G-G_0)_{J_+J_+}\|_\op \leq C^*K\sqrt{\frac{nK\log p_n}{N}} $$
    and for any $j \in J_+$, row $j$ of $(G-G_0)_{J_+J_+}$ has $\ell_2$ norm satisfying
    $$\|(G-G_0)_{jJ_+}\|_2 \leq C^* K\sqrt{\frac{nh_j\log p_n}{N}}$$
\end{lemma}
\begin{proof}
    From basic properties of the multinomial distribution, we can show that
    $$\bE(ZZ^T) = \sum_{i=1}^n \text{Cov}(Z_{*i}) = \sum_{i=1}^n \text{Cov}(D_{*i}) =  \frac{n}{N}M_0 - \frac{1}{N}D_0D_0^T$$
    and therefore 
    \begin{align*}
        G - G_0 &= DD^T - \frac{n}{N}M - \left(1-\frac{1}{N}\right)D_0D_0^T \\ 
                &= (D_0 + Z)(D_0 + Z)^T - \frac{n}{N}M - \left(1-\frac{1}{N}\right)D_0D_0^T \\ 
                &= ZD_0^T + D_0Z^T + ZZ^T - \frac{n}{N}M + \frac{1}{N}D_0D_0^T \\ 
                &= ZD_0^T + D_0Z^T + (ZZ^T - \bE[ZZ^T]) + \frac{n}{N}(M_0-M) \\ 
    \end{align*}
    and so we can write $(G-G_0)_{J_+J_+} = E_1 + E_2 + E_3$ where 
    $$E_1 := (ZD_0^T + D_0Z^T)_{J_+J_+}$$ 
    $$E_2 := (ZZ^T-\bE[ZZ^T])_{J_+J_+} $$
    $$E_3 := \frac{n}{N}(M_0-M)_{J_+J_+} $$
    We can deal with $E_3$ first. From \eqref{eq: M-error}, with probability $1-o(p_n^{-1})$ we have 
    $$\|E_3\|_\op \leq \frac{C^*n}{N}\sqrt{\frac{(\max_{j\in J_+} \tilde{h}_j )\log p_n}{nN}} \leq \frac{C^*}{N}\sqrt{\frac{n\log p_n}{N}}$$
    and for any $j \in J_+$,
    $$\|[E_3]_{j*}\|_2 = \frac{n}{N}|M(j,j) - M_0(j,j)| \leq \frac{C^*}{N}\sqrt{\frac{n\tilde{h}_j\log p_n}{N}} $$
    From \eqref{eq:C.4-2} and \eqref{eq:C.4-3}, with probability $1-o(p_n^{-1})$
    $$\|E_2\|_\op \leq \|E_2\|_F \leq C^*K\sqrt{\frac{n \log p_n}{N}} $$
    $$\|[E_2]_{j*}\|_2 \leq C^*\sqrt{\frac{n\tilde{h}_j K \log p_n}{N}}$$
    If we denote $A_1, \dots, A_K$ as the columns of $A$ and $W_1, \dots, W_k$ as the rows of $W$, 
    $$D_0 Z^T = \sum_{k=1}^K A_k(ZW_k)^T $$
    and so from \eqref{eq:C.4-1} and the fact that $\sum_{k=1}^K \|A_k\|_2 \leq \sum_{k=1}^K \|A_k\|_1 \leq K$, 
    $$\|E_1\|_\op \leq 2\|[D_0Z^T]_{J_+J_+}\|_\op \leq 2\sum_{k=1}^K \|A_k\|_2 \|Z_{J_+*}W_k\|_2\leq C^*K\sqrt{\frac{nK\log p_n}{N}}$$
    Let $Z_1, \dots, Z_p$ denote the rows of $Z$. From \eqref{eq:E.2}, \eqref{eq:C.4-1} and the fact that $\sum_{k=1}^K A_k(j) = h_j$ and $h_j \leq K$, for any $j \in J_+$:
    \begin{align*}
        \|[E_1]_{j*}\|_2 &\leq \sum_{k=1}^K A_k(j) \|Z_{J_+*}W_k\|_2 + \sum_{k=1}^K |Z_j^T W_k|\|A_k\|_2\\ 
        &\leq C^*h_j \sqrt{\frac{nK\log p_n}{N}} + C^*K \sqrt{\frac{n\tilde{h}_j \log p_n}{N}} \\ 
        &\leq C^*K\sqrt{\frac{nh_j\log p_n}{N}}
    \end{align*}
    Since the bounds for $\|E_1\|_\op$ and $\|[E_1]_{j*}\|_2$ dominate those for $E_2$ and $E_3$, our result follows. 
    
\end{proof}

\begin{lemma}\label{lem: D3-GG0}
    With probability $1-o(p_n^{-1})$, we also have 
    \begin{equation}\label{eq:GJJ-op}
        \|(G-G_0)_{JJ}\|_\op  \leq C^*K\sqrt{\frac{nK\log p_n}{N}}
    \end{equation}
    and for any $j \in J$, 
    \begin{equation}\label{eq:GjJ-2}
        \|(G-G_0)_{jJ}\|_2 \leq C^*K\sqrt{\frac{nh_j \log p_n}{N}} 
    \end{equation}
\end{lemma}
\begin{proof} This is simply a consequence of the previous lemma and the fact that $J \subseteq J_+$ with probability $1-o(p_n^{-1})$, which implies that $(G-G_0)_{JJ}$ is a submatrix of $(G-G_0)_{J_+J_+}$. Note that we refrain from applying the argument of the previous lemma directly to $(G-G_0)_{JJ}$, since $J$ and $Z$ are not independent (whereas $J_+$ is a non-random index set).
\end{proof}

\begin{corollary}\label{cor:D4}
    Let $g_n$ be a quantity satisfying 
    \begin{equation}\label{eq: g_n-def}
        c\sqrt{\frac{nN}{\log p_n}}\geq g_n \geq C^*K\sqrt{K}
    \end{equation}
    where $C^*$ in \eqref{eq: g_n-def} is the constant from \eqref{eq:GJJ-op} and $c$ is another constant to be determined. If
    $$\hat{K} := \max\left\{k:\lambda_k(G_{JJ}) > g_n\sqrt{\frac{n\log p_n}{N}}\right\}$$
    then $\hat{K} = K$ with probability $1 - o(p_n^{-1})$. 
\end{corollary}
\begin{proof}
    We have shown in Lemma A.3(e) that $[G_0]_{JJ}$ has rank $K$ on $\mathcal{E}$. By Weyl's inequality, 
    $$ \lambda_{K+1}(G_{JJ}) \leq \|(G-G_0)_{JJ}\|_\op \leq C^*K\sqrt{K}\sqrt{\frac{n\log p_n}{N}} \leq g_n \sqrt{\frac{n\log p_n}{N}}$$
    This implies $\hat{K} \leq K$. On the other hand, again by Weyl's inequality, 
    $$|\lambda_K(G_{JJ}) - \lambda_K([G_0]_{JJ})| \leq C^*K\sqrt{K} \sqrt{\frac{n \log p_n}{N}} \leq g_n\sqrt{\frac{n\log p_n}{N}}$$
    and since $G_0 := \left(1-\frac{1}{N}\right) D_0D_0^T$, by our assumption that $\sigma_K(A) \geq c\sqrt{K}$ and $\sigma_K(\Sigma_W) \geq c$,
    $$\lambda_K([G_0]_{JJ}) \geq \left(1-\frac{1}{N}\right)\sigma_K^2(A_{J*}) \sigma_K^2(W) \geq cKn > 2g_n \sqrt{\frac{n \log p_n}{N}}$$
    when $nN$ is sufficiently large and $c$ in \eqref{eq: g_n-def} is chosen appropriately. Hence, 
    $$\lambda_K(G_{JJ}) \geq  \lambda_K([G_0]_{JJ}) - |\lambda_K(G_{JJ}) - \lambda_K([G_0]_{JJ})| > g_n \sqrt{\frac{n \log p_n}{N}}$$
    and thus $K \leq \hat{K}$ with probability $1 - o(p_n^{-1})$. 
\end{proof}

    

Recall that $\hat{\Xi}$ contains the first $K$ eigenvectors of $G_{JJ}$ and $\Xi$ contains the first $K$ left singular vectors of $[D_0]_{J*}$, or equivalently the first $K$ eigenvectors of $[G_0]_{JJ}$. We will provide a coordinate-wise error bound for $\hat{\Xi}$ in Lemma \ref{lem:D8}. First we need a few lemmas. 

\begin{lemma} For any $j \in J$, $\|\Xi_{j*}\|_2 \leq Ch_j$. 
\end{lemma}
\begin{proof}
    Note $\Xi = A_{J*}V$ so $\Xi_{j*} = A_{j*}V$. Hence, 
    $$\|\Xi_{j*}\|_2 \leq \|V\|_\op \|A_{j*}\|_2 \leq \|V\|_\op \|A_{j*}\|_1 \leq \sigma_K^{-1}(A) h_j \leq Ch_j$$
    since we have shown before that the singular values of $V$ are just the inverses of the singular values of $A_{J*}$. 
\end{proof}
Note that on event $\cE$, $[D_0]_{J*}$ and hence $[G_0]_{JJ}$ has rank $K$. 

\begin{lemma}\label{lem: Lemma D.7}
    On event $\cE$, 
    $$cnK \leq \lambda_k([G_0]_{JJ}) \leq nK \text{ for all } k \in [K] \text{ and } \lambda_1([G_0]_{JJ}) \geq cn + \max_{2\leq k \leq K} \lambda_k([G_0]_{JJ}) $$
\end{lemma}
\begin{proof}
    We note that $[D_0D_0^T]_{JJ} = A_{J*}WW^TA_{J*}^T = nA_{J*}\Sigma_W A_{J*}^T$. Hence, 
    $$\lambda_1([G_0]_{JJ}) \leq n\|A\|_\op^2 \|\Sigma_W\|_\op \leq nK$$
    $$\lambda_K([G_0]_{JJ}) \geq n[\sigma_K(A_{J*})]^2\sigma_K(\Sigma_W)  \geq n[\sigma_K(A)]^2\sigma_K(\Sigma_W) \geq cnK$$
    We also note that for any two matrices $P$ and $Q$, the nonzero eigenvalues of $PQ$ are the same as those of $QP$. Thus the nonzero eigenvalues of $[D_0D_0^T]_{JJ}$ are the same as the nonzero eigenvalues of $WW^TA_{J*}^TA_{J*} =: n\Theta$. We have already shown in Lemma \ref{lem:B1}(d) that the gap between the first two eigenvalues of $\Theta$ are at least an absolute constant on $\cE$. Hence, our result follows. 
    
\end{proof}
\begin{lemma}[Row-wise estimation error for $\hat{\Xi}$]\label{lem:D8} Denote $\{\Xi_j: j \in J\}$ as the rows of $\Xi$ and $\{\hat{\Xi}_j: j \in J\}$ as the rows of $\hat{\Xi}$. With probability $1 - o(p_n^{-1})$, there exist $\omega \in \{\pm 1\}$ and an orthonormal matrix $\Omega^* \in \mathbb{R}^{(K-1)\times (K-1)}$ such that, if $\Omega := \text{diag}(\omega, \Omega^*) \in \mathbb{R}^{K \times K}$, we have
$$\|\Omega \hat{\Xi}_j - \Xi_j\|_2 \leq  C\sqrt{\frac{h_j \log p_n}{nN}} \quad \text{for all } j \in J$$
    
\end{lemma}
\begin{proof}
    Let $\hat{\xi}_1$ and $\xi_1$ be the first eigenvectors of $[G]_{JJ}$ and $[G_0]_{JJ}$ respectively. The gap between the first two eigenvalues of $[G_0]_{JJ}$ is at least $cn$, which is much greater than $C^*K\sqrt{\frac{nK\log p_n}{N}}$ (the high-probability bound on $\|(G-G_0)_{JJ}\|_\op$). By applying Lemma \ref{lem:D1}, there exists $\omega \in \{\pm 1\}$ such that with probability $1 - o(p_n^{-1})$, for all $j \in J$, 
    \begin{align*}
        |\omega \hat{\xi}_1(j) - \xi_1(j)| &\leq C\frac{h_j\|(G-G_0)_{JJ}\|_\op + \sqrt{K}\|(G-G_0)_{jJ}\|_2}{n}
        \\ 
        &\leq C\frac{h_j\sqrt{\frac{n \log p_n}{N}} + \sqrt{\frac{nh_j \log p_n}{N}}}{n} \\ 
        &\leq C\sqrt{\frac{h_j \log p_n}{nN}}
    \end{align*}
    where we applied $h_j \leq K$. 
    
    Let $\Xi^* = [\xi_2, \dots, \xi_K]$ contain the other $(K-1)$ eigenvectors of $[G_0]_{JJ}$, and define $\hat{\Xi}^*$ similarly. Again, since the smallest nonzero eigenvalue of $[G_0]_{JJ}$ is at least $cnK$, there exists an orthonormal matrix $\Omega^* \in \mathbb{R}^{(K-1)\times(K-1)}$ such that for all $j \in J$, 
    $$\|(\hat{\Xi}^*\Omega^* - \Xi^*)_{j*}\|_2 \leq C\sqrt{\frac{h_j \log p_n}{nN}} $$
    We then define $\Omega = \text{diag}(\omega, \Omega^*)$ and combine the above results. 
\end{proof}
\begin{lemma}[Estimation error for the point cloud]\label{lem: D.9}
   With probability $1- o(p_n^{-1})$, all entries of $\hat{\xi}_1$ have the same sign and there exists an orthonormal matrix $\Omega^* \in \mathbb{R}^{(K-1) \times (K-1)}$ such that 
    \begin{align*}
        \max_{j \in J}\|\Omega^* \hat{r}_j - r_j\|_2 \leq C \left(\frac{\log p_n}{nN}\right)^{1/4}
    \end{align*}
\end{lemma}
\begin{proof}
    First, we note that WLOG, we can assume $\omega = 1$. This is because from the previous lemma, for any $j \in J$, since $h_j \geq c\sqrt{\frac{\log p_n}{nN}}$,
    $$|\omega \hat{\xi}_1(j) - \xi_1(j)| \leq C\sqrt{\frac{h_j \log p_n}{nN}} \leq  C h_j \left(\frac{\log p_n}{nN}\right)^{1/4}$$
    whereas we also know from Lemma \ref{lem:B1}(e) that 
    $$\xi_1(j) > ch_j > 0$$
    We can see that $|\omega \hat{\xi}_1(j) - \xi_1(j)| \ll \xi_1(j)$ with high probability as $nN$ is sufficiently large, and this implies $\omega \hat{\xi}_1(j) \geq \xi_1(j)/2$. If $\hat{\xi}_1$ is defined such that the majority of its entries are positive (and in fact its entries are all of the same sign with high probability), we can simply assume $\omega = 1$ from now on. 

    Denote $\{\Xi_j: j \in J\}$ as the rows of $\Xi$ and $\{\hat{\Xi}_j: j \in J\}$ as the rows of $\hat{\Xi}$. Now, since by definition, 
    $$\begin{pmatrix} 1\\ r_j \end{pmatrix} = [\xi_1(j)]^{-1}\Xi_{j}, \quad \begin{pmatrix}
        1 \\ \Omega^* \hat{r}_j 
    \end{pmatrix}  = [\hat{\xi}_1(j)]^{-1} \Omega \hat{\Xi}_{j}$$
    it follows that 
    \begin{align*}
        \|\Omega^*\hat{r}_j - r_j\|_2  &= \left\|\frac{1}{\hat{\xi}_1(j)} \Omega \hat\Xi_j - \frac{1}{\xi_1(j)} \Xi_j\right\|_2  \\ 
        &= \left\|\frac{1}{\hat{\xi}_1(j)} (\Omega \hat{\Xi} - \Xi_j) - \frac{\hat{\xi}_1(j) - \xi_1(j)}{\hat{\xi}_1(j)}r_j\right\|_2 \\ 
        &\leq |\hat{\xi}_1(j)|^{-1}(\|\Omega \hat{\Xi}_j - \Xi_j\|_2 + \|r_j\|_2 |\hat{\xi}_1(j) - \xi_1(j)|)
    \end{align*}
    We have noted in Lemma \ref{lem:B2} that the point cloud $\{r_j: j \in J\}$ lies entirely in the convex hull of $v_1^*, \dots, v_K^*$, and Lemma \ref{lem:B1}(f) shows that $\max_{k \in [K]} \|v_k^*\|_2 \leq C$, so we also have $\max_{j \in J}\|r_j\|_2 \leq C$. We have also noted before that 
    $$\hat{\xi}_1(j) \geq \frac{\xi_1(j)}{2} > ch_j$$ with high probability. Therefore, with probability $1 - o(p_n^{-1})$, for all $j \in J$:
    \begin{align*}
        \|\Omega^* \hat{r}_j - r_j\|_2 &\leq C \sqrt{\frac{\log p_n}{h_jnN}} \\ 
        &\leq C \left(\frac{\log p_n}{nN}\right)^{1/4}
    \end{align*}
    since $\min_{j \in J} h_j \geq c\sqrt{\frac{\log p_n}{nN}}$ with probability $1-o(p_n^{-1})$. 
    \end{proof}

\begin{lemma}\label{lem: pi-error}
    If we denote the rows of $\hat{\Pi}$ from our proposed procedure as $\{\hat{\pi}_j: j \in J\}$ and the rows of $\Pi$ from the oracle procedure as $\{\pi_j: j \in J\}$, then with probability $1- o(p_n^{-1})$, 
    $$ \max_{j \in J}\|\hat{\pi}_j - \pi_j\|_1 \leq C \left(\frac{\log p_n}{nN}\right)^{1/4}$$
\end{lemma}
    
\begin{proof}
    Recall that $\hat{\pi}_j^\diamond \in \mathbb{R}^K$ is the unnormalized vector solving 
    \begin{equation*}
        \begin{pmatrix}
            1 & \dots & 1 \\ \hat{v}_1^* & \dots & \hat{v}_K^*
        \end{pmatrix} \hat{\pi}_j^\diamond = \begin{pmatrix}
            1 \\ \hat{r}_j
        \end{pmatrix} \iff \begin{pmatrix}
            1 & \dots & 1 \\ \Omega^* \hat{v}_1^* & \dots & \Omega^* \hat{v}_K^* 
        \end{pmatrix}\hat{\pi}_j^\diamond= \begin{pmatrix}
            1 \\ \Omega^* \hat{r}_j
        \end{pmatrix}
    \end{equation*}
    Therefore, 
    $$\hat{\pi}_j^\diamond = \hat{Q}^{-1} \begin{pmatrix}
        1 \\ \Omega^* \hat{r}_j
    \end{pmatrix} \quad \text{where} \quad \hat{Q} := \begin{pmatrix}
        1 & \dots & 1 \\ \Omega^* \hat{v}_1^* & \dots & \Omega^* \hat{v}_K^*
    \end{pmatrix}$$ 
    We also have $$\pi_j = Q^{-1} \begin{pmatrix}
        1 \\ r_j
    \end{pmatrix} \quad \text{where } Q = \begin{pmatrix}
        1 & \dots & 1 \\ v_1^* & \dots & v_K^* 
    \end{pmatrix}$$
    Consequently, 
    \begin{equation*}
        \|\hat{\pi}_j^\diamond - \pi_j\|_2 \leq \|\hat{Q}^{-1}\|_\op \|\Omega^* \hat{r}_j - r_j\|_2 + \|\hat{Q}^{-1} - Q^{-1}\|_\op \sqrt{\|r_j\|_2^2+1}
    \end{equation*} 
    Note that $\max_{j\in J}\|r_j\|_2 \leq C$ since the $r_j$'s are in the convex hull of $v_1^*, \dots, v_K^*$. Also, since $Q^T = [\diag(V_1)]^{-1}V$, we have $\|Q^{-1}\|_\op \leq C$ since 
    $$\max_{k \in [K]}V_1(k) \leq \frac{C}{\sqrt{K}} \quad \text{and} \quad \|V^{-1}\|_\op = \sigma_{1}(A)\leq \sqrt{K}$$
    Now, we note that with probability $1 - o(p_n^{-1})$,
    \begin{align*}
        \|\hat{Q} - Q\|_\op \leq \|\hat{Q} - Q\|_F &\leq \sqrt{K} \max_{k \in [K]}\|\Omega^*\hat{v}_k^* - v_k^*\|_2 \\ 
        &\leq \sqrt{K} \max_{j \in J}\|\Omega^* \hat{r}_j - r_j\|_2 \leq C \left( \frac{\log p_n}{nN}\right)^{1/4} = o(1)
    \end{align*}
    where we used Assumption \ref{ass:VH}, since in the oracle procedure the vertex hunting algorithm correctly returns $v_1^*, \dots, v_K^*$. Therefore, 
    $$\|\hat{Q}^{-1} - Q^{-1}\|_\op = \|\hat{Q}^{-1}(Q-\hat{Q}) Q^{-1}\|_\op \leq \|\hat{Q}^{-1}\|_\op \|\hat{Q}-Q\|_\op \|Q\|^{-1}_\op \leq C \left(\frac{\log p_n}{nN}\right)^{1/4} $$
    Here we note $\sigma_K(\hat{Q}) \geq \sigma_K(Q) - \|\hat{Q}- Q\|_\op \geq c - o(1) \geq c/2$ if $nN$ is large enough, so $\|\hat{Q}^{-1}\|_\op \leq C$. Therefore, we obtain 
    $$\|\hat{\pi}_j^\diamond - \pi_j\|_2 \leq C \left(\frac{\log p_n}{nN}\right)^{1/4}$$
    Now if we define $\hat{\pi}_j = \frac{\tilde{\pi}_j^\diamond}{\|\tilde{\pi}_j^\diamond\|_1}$ where $\tilde{\pi}_j^\diamond(k) = \max(\hat{\pi}_j^\diamond(k), 0)$, then since $\|\hat{\pi}_j\|_1 = \|\pi_j\|_1 = 1$, 
    \begin{align*}
        \|\hat{\pi}_j - \pi_j\|_1 &\leq \|\hat{\pi}_j - \tilde{\pi}_j^\diamond\|_1 + \|\tilde{\pi}_j^\diamond - \pi_j\|_1  \\ 
        &= |1-\|\tilde{\pi}_j^\diamond\|_1| \|\hat{\pi}_j\|_1 + \|\tilde{\pi}_j^\diamond - \pi_j\|_1 \\ 
        &= |\|\pi_j|_1 - \|\tilde{\pi}_j^\diamond\|_1| + \|\tilde{\pi}_j^\diamond - \pi_j\|_1 \\ 
        &\leq 2\|\pi_j - \tilde{\pi}_j^\diamond\|_1 \\ 
        &\leq 2\|\pi_j - \hat{\pi}_j^\diamond\|_1 \leq 2\sqrt{K} \|\hat{\pi}_j^\diamond - \pi_j\|_2 \\ 
        &\leq C\left(\frac{\log p_n}{nN}\right)^{1/4}
    \end{align*}
\end{proof}

\section{Estimation error of \texorpdfstring{$\hat{A}$}{hatA}}\label{appendix: E}
    In this section, we will additionally impose the $\ell_q$-sparsity assumption \eqref{eq: lq-sparsity-def} for $q \in (0,1)$. 

    \begin{lemma}
        Under Assumption \ref{ass:sparsity}, if $\beta := \frac{\alpha_-\alpha}{\sigma_K(\Sigma_W)}$ and $\tau_n := \sqrt{\frac{\log p_n}{nN}}$, on event $\cE$
        \begin{equation}\label{eq:1209}
        \|A_{J^c*}\|_1 \leq \frac{K}{1-q} s(\beta\tau_n)^{1-q}
        \end{equation}
    \end{lemma}
    \begin{remark}
                {\rm We assume from now on that $s$ does not grow too quickly relative to $nN$ so that the RHS of \eqref{eq:1209} is $o(1)$.}
    \end{remark}
    \begin{proof}
        On event $\mathcal{E}$ we have $J_- \subseteq J$, so $j \notin J$ implies $M_0(j,j) \leq \alpha_-\alpha\tau_n$ where $\tau_n := \sqrt{\frac{\log p_n}{nN}}$. Since $\sigma_K(\Sigma_W)h_j \leq M_0(j,j)$, $j \notin J$ implies $A_{jk} \leq h_j \leq \beta\tau_n$ for any $k \in [K]$ on $\mathcal{E}$. Then with probability $1 - o(p_n^{-1})$, for any $ k\in [K]$, 
        \begin{align*}
            \|A_{J^ck}\|_1 &= \sum_{j \not\in J}\min(A_{jk}, \beta \tau_n) \leq \sum_{j=1}^p \min(A_{(j)k}, \beta\tau_n) \\ 
            &\leq \sum_{j=1}^p \min (s^{1/q}j^{-1/q}, \beta \tau_n) \leq \int_0^\infty \min(s^{1/q}t^{-1/q}, \beta \tau_n)dt
        \end{align*}
        Now, let $t_0 := s(\beta\tau_n)^{-q}$ so that $ s^{1/q}t_0^{-1/q} = \beta\tau_n$. Then continuing from the above display, 
        \begin{align*}
            \|A_{J^ck}\|_1 &\leq t_0\beta\tau_n  + s^{1/q}\int_{t_0}^\infty t^{-1/q}dt \\ 
            &= t_0\beta\tau_n + \frac{q}{1-q} s^{1/q}t_0^{1-1/q} = \frac{1}{1-q} t_0\beta\tau_n \\ 
            &= \frac{1}{1-q} s(\beta\tau_n)^{1-q}
        \end{align*}
        and the result follows by summing up this bound across $k \in [K]$. Note that the assumption $q \in (0,1)$ ensures the integrals above converge. 
    \end{proof}
    \begin{lemma}
        On event $\cE$, if $\tilde{A}_{J*}$ is defined as in \eqref{eq:oracle}, 
        \begin{equation}\label{eq:1210}
        \|\tilde{A}_{J*} - A_{J*}\|_1 = \|A_{J^c*}\|_1 \leq \frac{K}{1-q}s(\beta\tau_n)^{1-q}  
        \end{equation}
    \end{lemma}
    \begin{proof}
        We note that the columns of $\tilde{A}_{J*}$ sum up to 1, the columns of $A$ sum up to 1, and as a result of the definition of $\tilde{A}_{J*}$ in \eqref{eq:oracle},
        $$\tilde{A}_{J*} - A_{J*} = \tilde{A}_{J*} \cdot \diag(\|A_{J^c1}\|_1, \dots, \|A_{J^cK}\|_1) $$
        Then $\|\tilde{A}_{J*} - A_{J*}\|_1 = \|A_{J^c*}\|_1$ and our result follows from the previous lemma.  
    \end{proof}
\begin{theorem}
With probability $1 - o(p_n^{-1})$, for some constant $C$ that may depend on $K$ and $q$, we have 
$$\|\hat{A} - A\|_1 \leq C\left[\left(\frac{\log p_n}{nN}\right)^{1/4} + s\left(\frac{\log p_n}{nN}\right)^{\frac{1-q}{2}}\right]  $$
\end{theorem}
    
\begin{proof}
    Consider the unnormalized matrices $$\hat{A}_{J*}^\diamond := \diag(\hat{\xi}_1)\hat{\Pi} \quad \text{and} \quad \tilde{A}_{J*}^\diamond := \diag(\xi_1) \Pi$$
    Then with probability $1 - o(p_n^{-1})$, for any $j \in J$, 
    \begin{align}
        \|(\hat{A}^\diamond - \tilde{A}^\diamond)_{j*}\|_1  &= \|\hat{\xi}_1(j) \hat{\pi}_j  - \xi_1(j) \pi_j\|_1 \nonumber\\
        &\leq |\hat{\xi}_1(j)|\|\hat{\pi}_j - \pi_j\|_1 + |\hat{\xi}_1(j)-\xi_1(j)|\|\pi_j\|_1 \nonumber\\ 
        &\leq C \left[h_j \left(\frac{\log p_n}{nN}\right)^{1/4} + \sqrt{\frac{h_j\log p_n}{nN}}\right] \nonumber\\ 
        &\leq C h_j \left(\frac{\log p_n}{nN}\right)^{1/4}\label{eq:1212121}
    \end{align}
    where again we note that on event $\cE$, $h_j > \alpha_+\alpha\sqrt{\frac{\log p_n}{nN}}$ if $j \in J$. Since $\sum_{j=1}^p h_j =K$, with probability $1 - o(p_n^{-1})$, 
    \begin{equation}\label{eq:1212122}
            \|\hat{A}_{J*}^\diamond - \tilde{A}_{J*}^\diamond\|_1 \leq C \left(\frac{\log p_n}{nN}\right)^{1/4} = o(1)
    \end{equation}

    Now, $\hat{A}_{J*}$ and $\tilde{A}_{J*}$ are defined by normalizing the columns of $\hat{A}_{J*}^\diamond$ and $\tilde{A}_{J*}^\diamond$, so we have for each $j \in J$ and $k \in [K]$
    $$\hat{A}_{jk} = \frac{\hat{A}_{jk}^\diamond}{\|\hat{A}_{Jk}^\diamond\|_1}  \quad \text{and} \quad \tilde{A}_{jk} = \frac{\tilde{A}_{jk}^\diamond}{\|\tilde{A}_{Jk}^\diamond\|_1} = \frac{A_{jk}}{\|A_{Jk}\|_1}$$
    Therefore, for each $j \in J$ and $k \in [K]$, 
    \begin{align}
        |\hat{A}_{jk} - \tilde{A}_{jk}| &= \left|\frac{\hat{A}_{jk}^\diamond}{\|\hat{A}_{Jk}^\diamond\|_1} - \frac{\tilde{A}_{jk}^\diamond}{\|\tilde{A}_{Jk}^\diamond\|_1}\right| \nonumber\\ 
        &\leq \frac{|\hat{A}_{jk}^\diamond - \tilde{A}_{jk}^\diamond|}{\|\hat{A}_{Jk}^\diamond\|_1} + \tilde{A}_{jk}^\diamond \left|\frac{1}{\|\hat{A}_{Jk}^\diamond\|_1} - \frac{1}{\|\tilde{A}_{Jk}^\diamond\|_1}\right| \nonumber\\ 
        &\leq \frac{|\hat{A}_{jk}^\diamond - \tilde{A}_{jk}^\diamond| + \tilde{A}_{jk}\|\hat{A}_{Jk}^\diamond - \tilde{A}_{Jk}^\diamond\|_1}{\|\hat{A}_{Jk}^\diamond\|_1} \nonumber\\ 
        &= \frac{|\hat{A}_{jk}^\diamond - \tilde{A}_{jk}^\diamond|}{\|\hat{A}_{Jk}^\diamond\|_1} + \frac{A_{jk}\|\hat{A}_{Jk} - \tilde{A}_{Jk}^\diamond\|_1}{\|A_{Jk}\|_1 \|\hat{A}_{Jk}^\diamond\|_1} \label{eq:1212}
    \end{align}
    Now, 
    $$\|A_{Jk}\|_1 = 1-\|A_{J^ck}\|_1 \geq 1-\frac{1}{1-q}s(\beta\tau_n)^{1-q} \geq c$$ for some absolute constant $c \in (0,1)$ as $nN$ becomes sufficiently large. Furthermore, since by definition of $\Pi$ we have $\tilde{A}_{Jk}^\diamond = \diag(\xi_1)\Pi = A_{J*} \cdot \diag(V_1)$ and $\min_{k\in [K]}V_1(k) \geq \frac{c}{\sqrt{K}}$, so 
    $$\|\tilde{A}_{Jk}^\diamond\|_1 = V_1(k) \|A_{Jk}\|_1 \geq c $$
    and thus 
    $$\|\hat{A}_{Jk}^\diamond\|_1 \geq \|\tilde{A}_{Jk}^\diamond\|_1 - \|\tilde{A}_{Jk}^\diamond - \hat{A}_{Jk}^\diamond\|_1 \geq c - C\left(\frac{\log p_n}{nN}\right)^{1/4} \geq c/2$$ as $nN$ becomes sufficiently large. Hence, we have from \eqref{eq:1212}, \eqref{eq:1212121} and \eqref{eq:1212122} that 
    $$|\hat{A}_{jk} - \tilde{A}_{jk}| \leq C h_j \left(\frac{\log p_n}{nN}\right)^{1/4}$$
    and so for any $j \in J$, 
    $$\|\hat{A}_{j*} - \tilde{A}_{j*}\|_1 \leq C h_j \left(\frac{\log p_n}{nN}\right)^{1/4} $$
    which, since $\sum_{j=1}^p h_j = K$, implies
    \begin{equation}\label{eq:1213}
        \|\hat{A}_{J*} - \tilde{A}_{J*}\|_1 \leq C \left(\frac{\log p_n}{nN}\right)^{1/4} 
    \end{equation}
    We combine \eqref{eq:1209}, \eqref{eq:1210} and \eqref{eq:1213} to obtain what we need to prove.
\end{proof}

\section{Results on Archetype Analysis \texorpdfstring{\citep{javadi2020nonnegative}}{(Javadi and Montanari, 2020)}}\label{appendix: Javadi}
To facilitate our discussion on relaxing the separability assumption, we summarize the results of \cite{javadi2020nonnegative} in this section. 

We first introduce the notations in this paper. For a point $u \in \mathbb{R}^d$ and a matrix $V \in \mathbb{R}^{m \times d}$, let 
    $$\mathscr{D}(u; V) := \min\{\|u - V^T\pi\|_2^2: \pi \in \Delta^m\}, \text{ where}$$
    $$\Delta^m := \{x \in \mathbb{R}^m: x^T \mathbf{1}_m = 1 \text{ and } x_j \geq 0 \text{ for all } 
    j \in [m]\}$$ 
    In words, $\mathscr{D}(u; V)$ is the square of the distance between $u$ and $\text{conv}(V)$, where $\text{conv}(V)$ denotes the convex hull of the rows of $V$. If $U \in \mathbb{R}^{p \times d}$ is a matrix with rows $u_1, \dots, u_p \in \mathbb{R}^d$, we generalized the above definition by letting 
    \begin{equation}
        \mathscr{D}(U; V) := \sum_{l=1}^p \mathscr{D}(u_l; V)
    \end{equation}
    Now, consider a factorization of the form $X_0 = W_0 H_0$, where the rows of $X_0 \in \mathbb{R}^{m \times (K-1)}$ form a point cloud,  $W_0 \in \mathbb{R}^{m \times K}$ is a matrix of weights whose rows are in $\Delta^{K}$, and the rows of $H_0 \in \mathbb{R}^{K \times (K-1)}$ are the $K$ simplex vertices. 
    \begin{definition}[$\alpha$-uniqueness]\label{def: alpha-uniqueness} We say that the point cloud $X_0 = W_0 H_0$ satisfies \textit{uniqueness} with parameter $\alpha > 0$ (or $\alpha$-uniqueness) if for all $H \in \mathbb{R}^{K\times (K-1)}$ with $\text{conv}(X_0) \subseteq \text{conv}(H)$, we have
    \begin{equation}\label{alpha-uniqueness-appendix}
        \mathscr{D}(H; X_0)^{1/2} \geq \mathscr{D}(H_0; X_0)^{1/2} + \alpha[\mathscr{D}(H; H_0)^{1/2} + \mathscr{D}(H_0; H)^{1/2}]
    \end{equation}
    \end{definition}

        The motivation behind this assumption is quite clear. Any $H$ with $\text{conv}(X_0) \subseteq \text{conv}(H)$ is a plausible explanation of the data. For $H_0$ to be identifiable, we want $\mathscr{D}(H; X_0) > \mathscr{D}(H_0; X_0)$ if $H \neq H_0$, and so \eqref{alpha-uniqueness-appendix} is a quantitative formulation of this requirement. Note that if $X_0 = W_0H_0$ is a separable factorization, then it always satisfies uniqueness with $\alpha = 1$. Indeed, whenever $\text{conv}(H_0) = \text{conv}(X_0)$, one has $\mathscr{D}(H; X_0) = \mathscr{D}(H; H_0)$ and $\mathscr{D}(H_0; X_0) = \mathscr{D}(H_0; H) = 0$.

            The vertex hunting procedure considered in \cite{javadi2020nonnegative} is as follows. Suppose we observe $X$ which is a noisy version of $X_0$: 
            \begin{equation}\label{eq:javadi-noisy-pt-cloud}
                X = X_0 + Z = W_0 H_0 + Z 
            \end{equation}
    Let $x_1, \dots, x_m$ be the rows of $X$. We can obtain an estimator $\hat{H}$ of $H_0$ by solving the following optimization problem (Archetype Analysis): 
    \begin{equation}\label{javadi-VH}
    \text{minimize } \mathscr{D}(H; X) \text{ s.t. } \mathscr{D}(x_i; H) \leq \delta^2 \text{ for all } i \in [m]  
    \end{equation}
    where $\delta \geq \max_{i\in [m]}\|Z_{i*}\|_2$. In light of Corollary \ref{corollary:pt-cloud-error}, we want to choose $\delta \geq C\left(\frac{\log p_n}{nN}\right)^{1/4}$ in our context, where $C$ is the constant in \eqref{eq:pt-cloud-error} (replace $X_0$ in \eqref{eq:javadi-noisy-pt-cloud} with the point cloud matrix $R$ from out oracle procedure, and $X$ with the point cloud matrix $\hat{R}$ from Definition \ref{actual-procedure}). 

    The main theoretical result of \cite{javadi2020nonnegative} is that their vertex hunting procedure is robust to noise in the point cloud. 
    \begin{theorem}[Theorem 1 of \cite{javadi2020nonnegative}]\label{theorem: F2}
        Suppose $X_0$ satisfies the $\alpha$-uniqueness assumption, and $\text{conv}(X_0)$ contains a $(K-1)$-dimensional ball of radius $\mu > 0$. Consider the vertex hunting procedure defined by \eqref{javadi-VH}, with $\delta = \max_{i \in [m]}\|Z_{i*}\|_2$. If 
        $$\max_{i \in [m]}\|Z_{i*}\|_2 \leq \frac{\alpha \mu }{30K^{3/2}} $$
        then 
        \begin{equation}
            \|\hat{H} - H_0\|_F^2 \leq \frac{C^2K^5}{\alpha^2} \delta^2
        \end{equation}
        Here, the constant $C$ may depend on $\mu$ and the maximum/minimum singular values of $H_0$, and we ignore the vertex label permutation (by redefining $\hat{H}$ if necessary).
    \end{theorem}

    Using similar proof techniques as in the above theorem, we can also show the following robustness result for Archetype Analysis without using the $\alpha$-uniqueness condition (the proof is omitted for brevity). In \eqref{yating-javadi}, we do not need to assume separability (in which case one has $\mathscr{D}(H_0; X_0) = 0$), but we want the distance from the vertices in $H_0$ to the convex hull of the point cloud $X_0$ to be no larger than $\delta$. Again, $\delta \asymp \left(\frac{\log p_n}{nN}\right)^{1/4}$ when applied to our topic modeling setup. 

    \begin{theorem}\label{theorem: F3} Using the same assumptions as in Theorem \ref{theorem: F2} except the $\alpha$-uniqueness condition, if $\max_{i\in [m]}\|Z_{i*}\|_2 \leq  \delta \leq \frac{\mu}{2K+2}$, the vertex hunting procedure \eqref{javadi-VH} satisfies for some constants $C_1, C_2 > 0$:
    \begin{equation}\label{yating-javadi}
         \|\hat{H} - H_0\|_F^2 \leq C_1 \mathscr{D}(H_0; X_0) + C_2 \delta^2
    \end{equation}
    \end{theorem}

In practice, the vertex hunting procedure defined \eqref{javadi-VH} is difficult to use. When applied on real dataset, one may prefer to work with the Lagrangian form of \eqref{javadi-VH}: 
    \begin{equation}
        \hat{H}_\lambda = \argmin_{H}[\mathscr{D}(X; H) + \lambda \mathscr{D}(H;X)]
    \end{equation}
Algorithms to solve this non-convex optimization problem are available in Section 4 of \cite{javadi2020nonnegative}.
    
\newpage 

\section{Synthetic experiments: additional results}\label{appendix:experiment}
This appendix supplements the synthetic experiments presented in the main text. In particular, we discuss and illustrate here in more details the Zipf Law used in our  experiments and provide further plots to assess the impact of additional model's parameters (including the intensity of the anchor words, value of the Zipf Law coefficients, etc). 

\subsection{Zipf's Law: Illustration and Comparison} 

Most of the synthetic experiments we use in this paper rely on the generation of documents where word frequencies follow a Zipf's law distribution (see Equation~\ref{eq:zipf}). Figure~\ref{fig:zipfs} illustrates instances of such frequency distributions for a dictionary of size $p=10,000$ words as we vary the parameters of this distribution (namely, the values of $\alpha_{zipf}$ and $\beta_{zipf}$). 
Figure~\ref{fig:comp_zipf_unif} compares the frequency heterogeneity resulting from sampling frequencies $f_{(j)}$'s from a Zipf law distribution (Equation~\ref{eq:zipf}) to frequencies sampled from a Uniform distribution:
$$ f_{(j)} \propto \text{Uniform}(0,1).$$ These two figures illustrate in particular the fast decay in word frequencies under the Zipf Law. In fact, for the reference Zipf law ( $a_{\text{zipf}} = 1$ and $b_{\text{zipf}} = 2.7$), only $10\%$  of words have frequencies above $0.001$ -- making the rest of the words extremely rare. As per Figure~\ref{fig:zipfs}, this decay increases rapidly as the parameter $\alpha_{zipf}$ increases.  By comparison, for the uniform distribution, all word frequencies are of the same order of magnitude. One of the main assumptions in this paper is the weak sparsity of the row sums of the topic matrix $A$, described in Assumption \ref{ass:sparsity} of the main text: $\max_{k \in [K]} \max_{j \in [p]} j A_{(j)K}^q\leq s$, which, as argued in the main text,  is best reflected by the Zipf Law.

\begin{figure}[h!]
     \centering
     \begin{subfigure}[t]{0.50\textwidth}
         \centering
    \includegraphics[width=\textwidth]{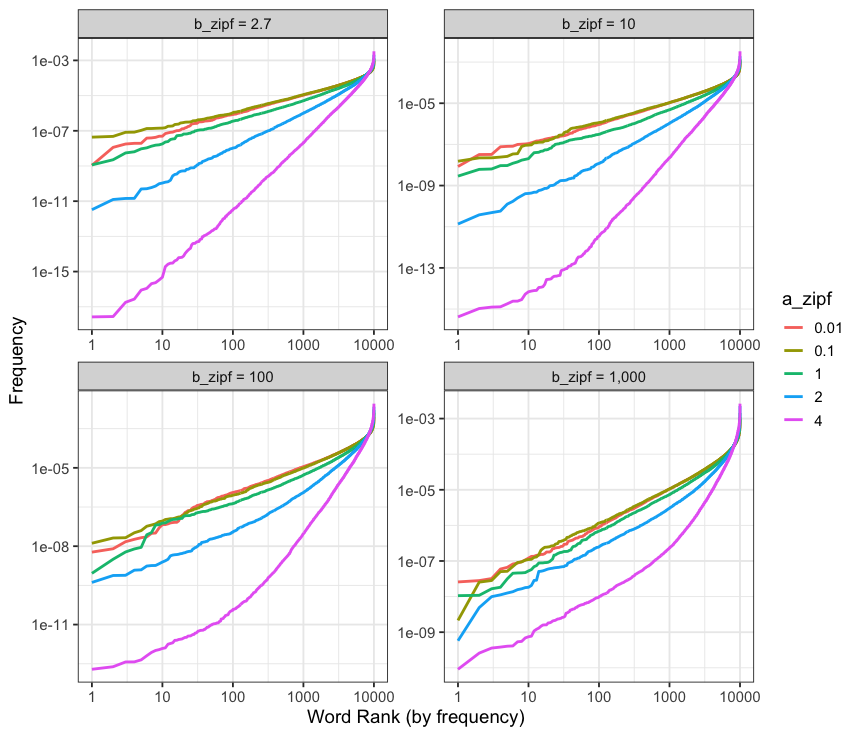}
    \caption{Examples of word frequency distributions under various Zipf law parameters. Note the quick decay of the frequencies as the value of $\alpha$ increases.}
    \label{fig:zipfs}
     \end{subfigure}
     \hfill
     \begin{subfigure}[t]{0.48\textwidth}
         \centering
         \includegraphics[width=\textwidth]{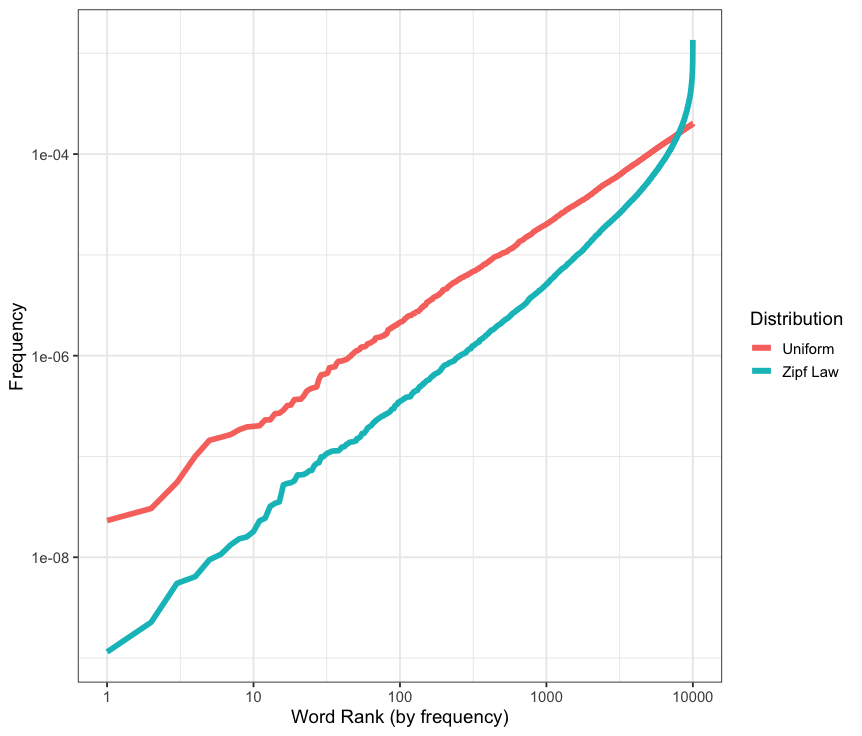}
         \caption{Comparison of the word frequencies generated under a Uniform distribution (red) and a Zipf law (blue) with parameter $a_{\text{zipf}} = 1$ and $b_{\text{zipf}} = 2.7$. }
         \label{fig:comp_zipf_unif}
     \end{subfigure}
     \caption{Comparison of word frequencies in a dictionary of size $p=10,000$ under various generation mechanisms. The $x$ and $y$ axis are presented on a log-scale.}\label{fig:word_frequencies}
\end{figure}

\subsection{Synthetic Experiment: Uniform Distribution of non-anchor words} 

In this paragraph, we describe the results of our synthetic experiments using a Uniform distribution for the generation of non-anchor words. 

\xhdr{Data Generation mechanism} In this setting, the data is generated as follows:
\begin{equation}\label{eq:uniform_gen_mech}
    \begin{split}
        \forall i \in \{1, \cdots, n\}, &\quad W_i \sim \text{Dirichlet}(\mathbf{1}_K)\\
      \forall j \in \{1, \cdots, 5\}, &\quad A_{j + k(i-1), k} =\delta\\
      \forall j \in \{5K, \cdots p\}, &\quad A_{jk}  \sim \text{Uniform}(0,1)\\\
    \end{split}
\end{equation}

This is thus an identical generation setup as in the main text, but using a Uniform distribution rather than a Zipf Law to generate the frequencies. As noted in the previous section, this corresponds to a setting in which all frequencies have roughly the same amplitude (homogeneous frequencies), which does not agree (a) with our Assumption~\ref{ass:sparsity} (weak-sparsity); and (b) observations or realistic models of word frequencies in real documents \citep{corral2015zipf}. Nonetheless, we run our simulation pipeline in this setting and report the results in Figure~\ref{fig:uniform_distribution}. Results are averaged over 50 experiments, using $K=5$, $\alpha_{dirichlet} = 1$, and 5 anchor words with intensity $\delta_{anchor} = 0.001$.

\begin{figure}
\begin{subfigure}[b]{\textwidth}
    \centering
    \includegraphics[width=\textwidth]{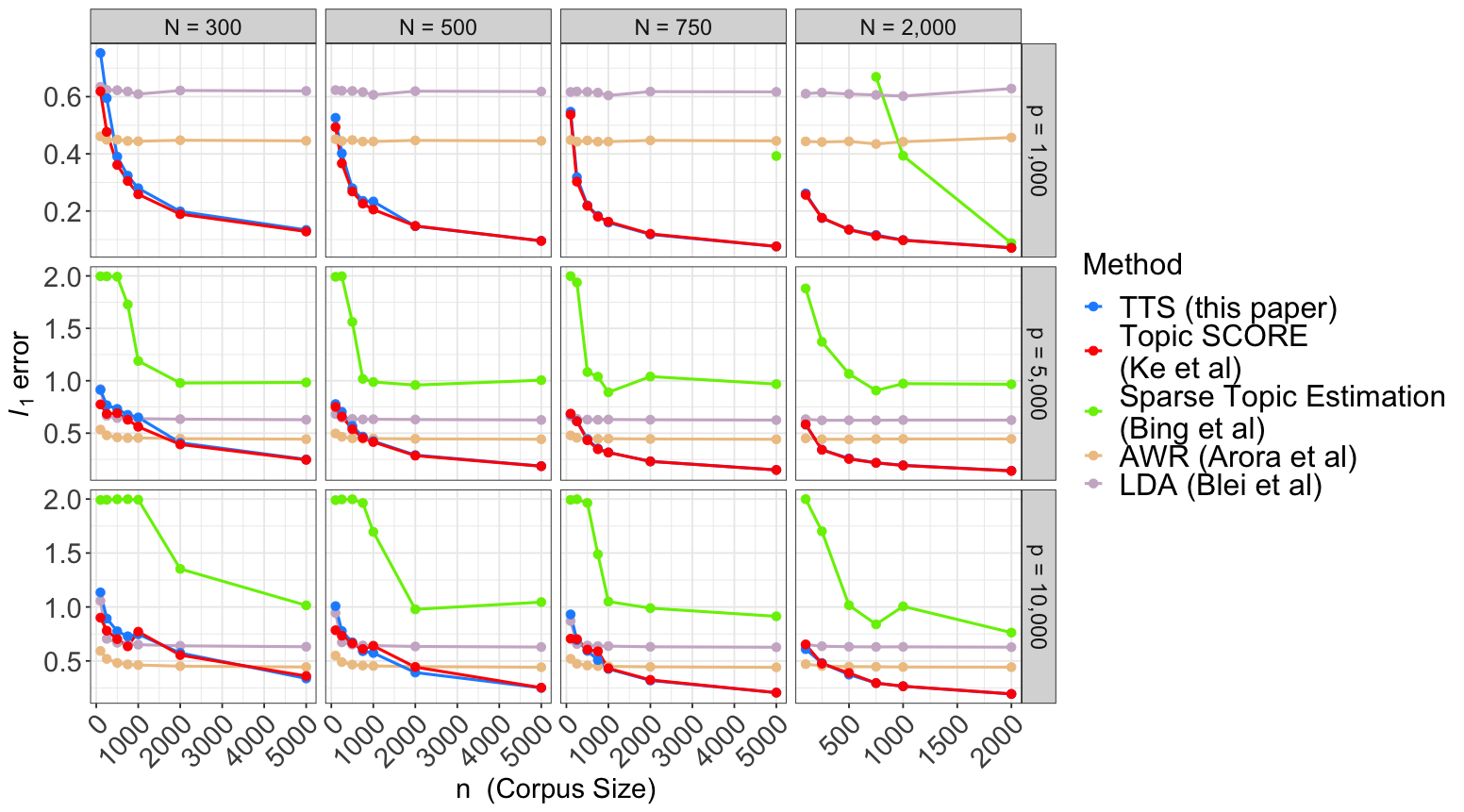}
    \caption{Median $\ell_1$ error  $\mathcal{L}_1(\hat{A}, A) = \min_{\Pi \in \mathcal{P}} \frac{1}{K}\| \hat{A}\Pi - A \|_1 $ for the different methods.  Rows indicate the size of the vocabulary $p$, while columns indicate the document length $N$. }
    \label{fig:uniform_distribution}    
\end{subfigure}
~
\begin{subfigure}[b]{\textwidth}
      \centering
    \includegraphics[width=\textwidth]{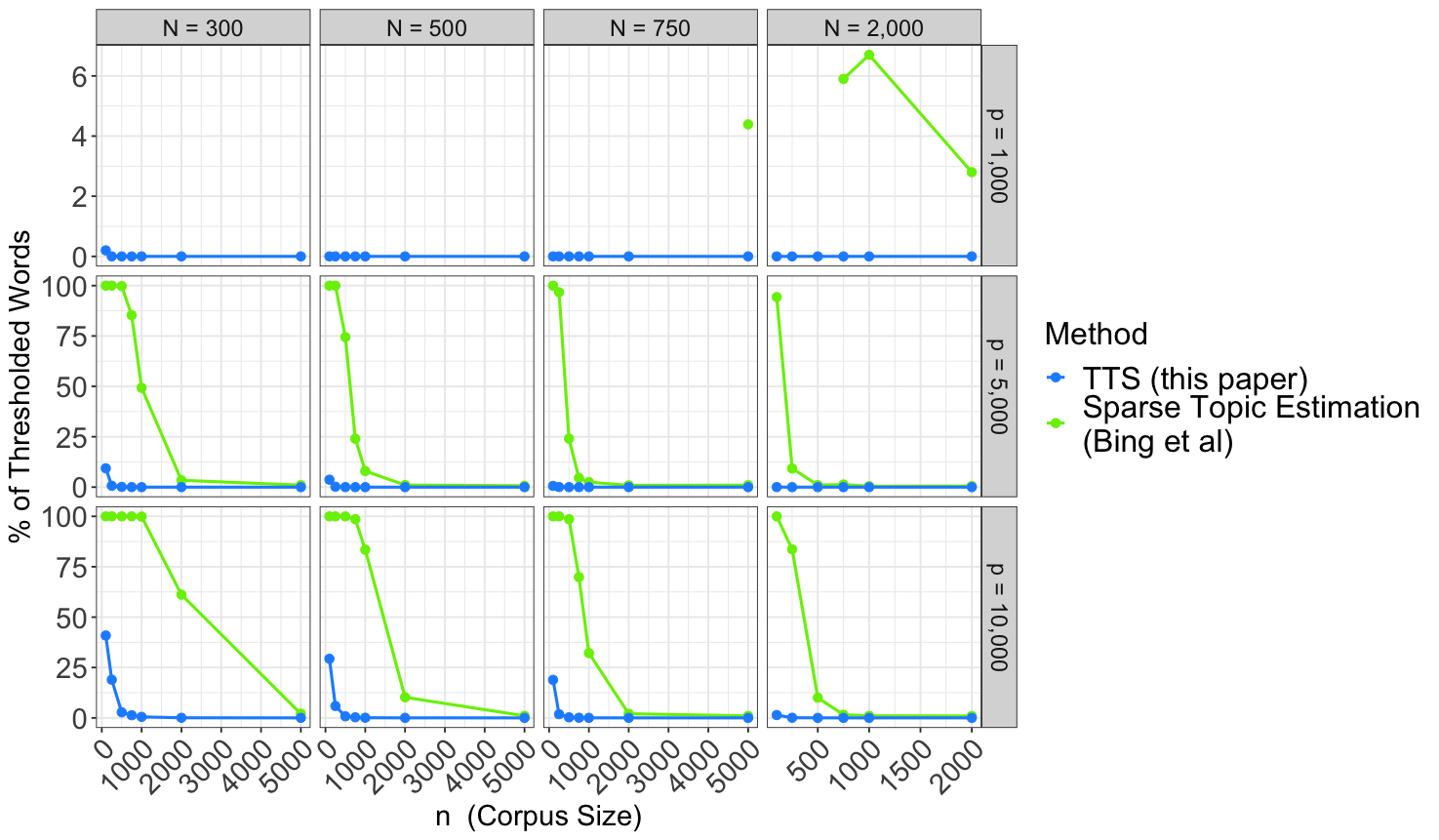}
    \caption{Percentage of thresholded words as a function of corpus length, dictionary size and document length.  }
    \label{fig:uniform_distribution_freq}  
\end{subfigure}

    \caption{Comparison of the different methods under the uniform frequency generation mechanism detailed in Equation \ref{eq:uniform_gen_mech}. Results are here shown for a fixed number of topics $K=5$, averaged over  50 independent trials, and plotted as a function of corpus size $n$. For low values of $p$, the ``Sparse Topic Estimation'' method of \cite{bing2020optimal} does not appear as the number of estimated topics that it estimated was less than the true value $K=5$; therefore we were unable to evaluate it and do not report it in the plots. }\label{fig:uniform_res}
\end{figure}

\begin{figure}
    \centering
    \includegraphics[width=\textwidth]{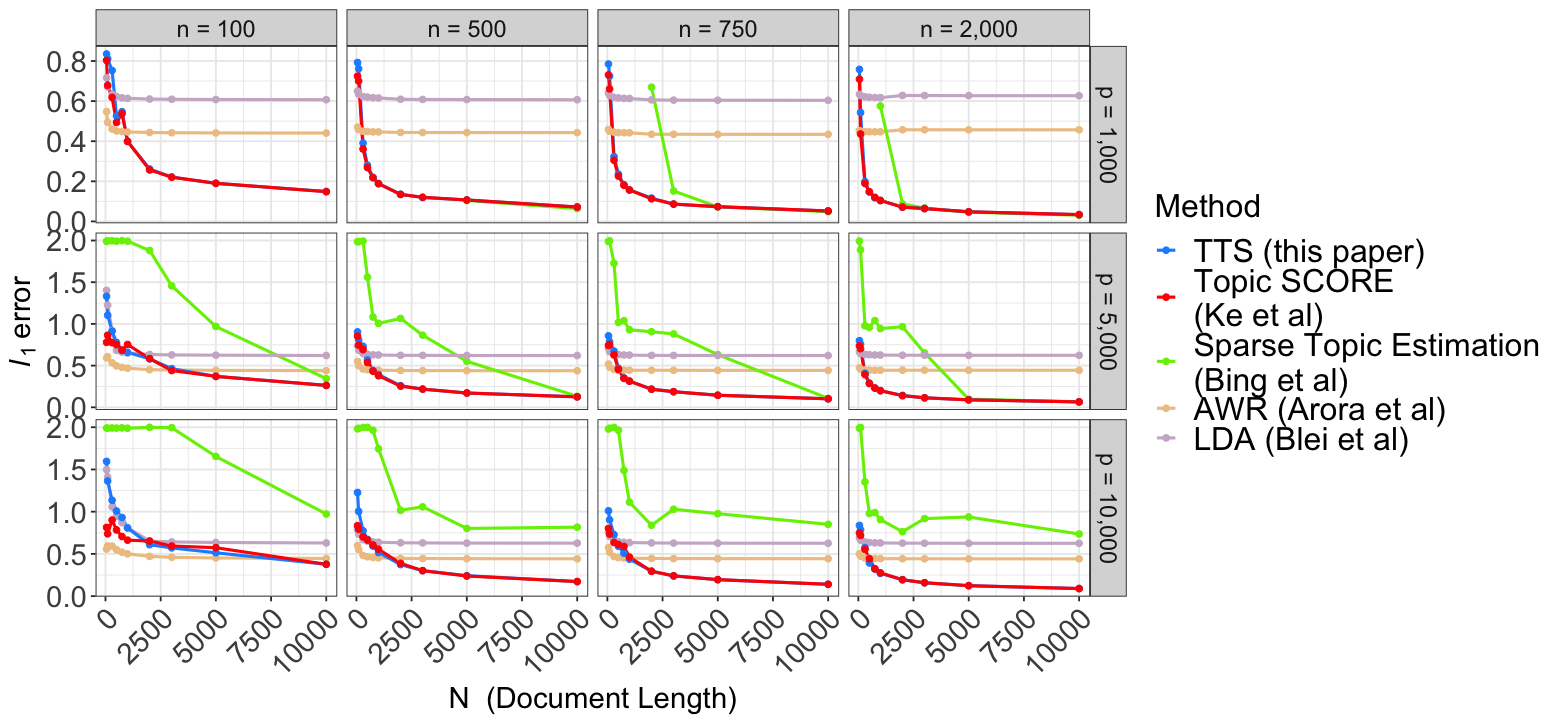}
    \caption{Same experiments as depicted in Figure ~\ref{fig:uniform_res}. Results are here displayed as a function of document length $N$.}
    \label{fig:enter-label}
\end{figure}

\xhdr{Analysis of the results} We note that Topic SCORE \citep{ke2022using} and our method perform similarly in these conditions. As shown in Figure~\ref{fig:uniform_distribution_freq}, our method does not threshold much in this regime --- as is probably to be expected, since the frequencies of all the words are of the same order. As highlighted in the main text, this suggests that, whilst designed to leverage the weak sparsity of the matrix $A$, our method is robust to various frequency regimes. Our method is, in fact, able to perform satisfactorily even in homogeneous (and less sparse) settings. Interestingly, in these regimes, AWR becomes an alternative approach (in particular, as the number of words $p$ increases).

\subsection{Varying additional parameters}

In this subsection, we propose evaluating the effect of other parameters in our data generation mechanism detailed in Section~\ref{sec:experiments} a. In particular, we assess here (a) the effect of the anchor word frequency $\delta_{\text{anchor}}$ (Figure~\ref{fig:AW-freq}); and (b) the impact of the Zipf law coefficient $a_{zipf}$.  

\begin{figure}
    \centering
    \includegraphics[width=\textwidth]{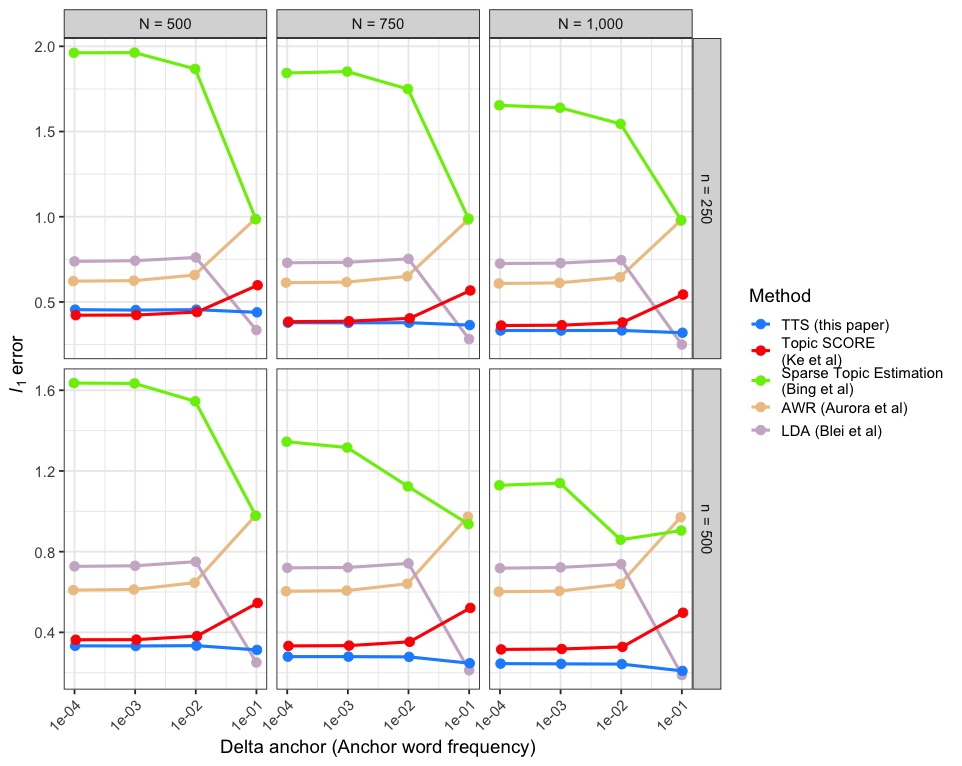}
    \caption{Effect of the anchor word frequency $\delta_{anchor}$ on the results. Here the number of topics is fixed to $K=3$ and the dictionary has size $p=5,000.$ Consistently with the experiments presented in the main text, each topic has 5 anchor words.}
    \label{fig:AW-freq}
\end{figure}

\xhdr{Discussion: Effect of the frequency of the anchor words} As observed in Figure~\ref{fig:AW-freq}, the frequency of the anchor words does not appear to have a great impact on the results of the SCORE-based methods. Increasing the frequency of anchor words does seem to improve the performance of the Sparse Topic Estimation Method of \cite{bing2020fast} and for LDA\cite{blei2003latent}.

\xhdr{Discussion: Effect of the parameter $a_{zipf}$} As observed in Figure~\ref{fig:effect_azipf}, our method offers significant improvement over others as the word frequency heterogeneity increases. In fact, both the Sparse Topic Estimation of \cite{bing2020fast} and our method's $\ell_1$ errors decrease significantly as the heterogeneity of the word frequency (high $a_{zipf}$) increases. This is reassuring, since both methods are the only out of the 5 actively leveraging the sparsity of the topic matrix to improve estimation.

\begin{figure}
    \centering
    \includegraphics[width=\textwidth]{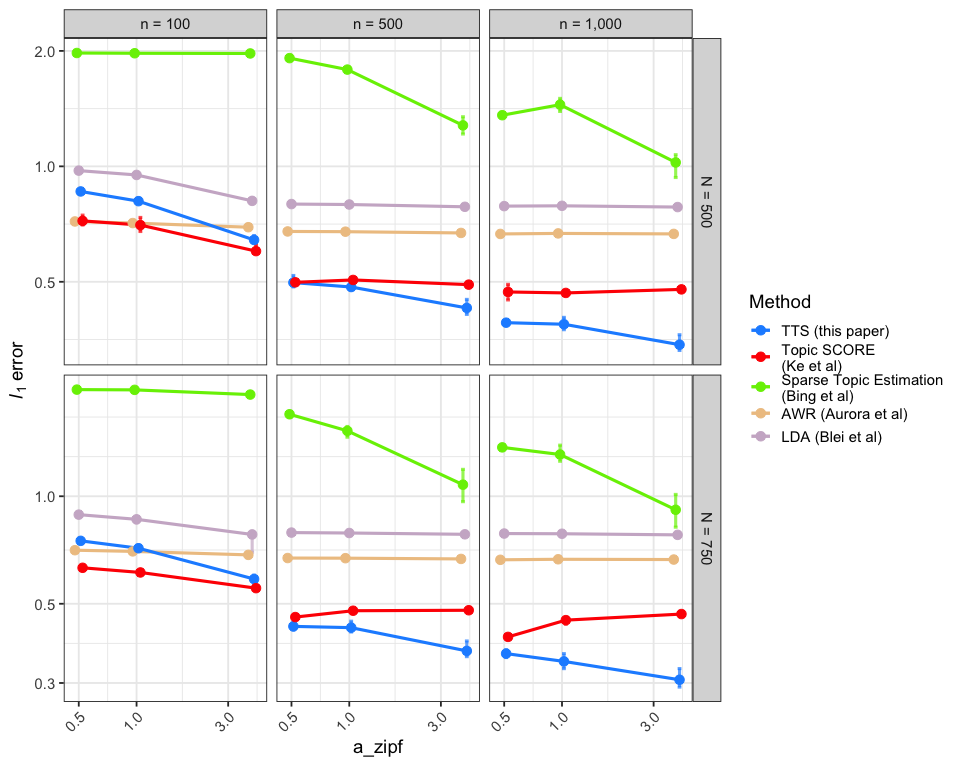}
    \caption{Effect of the parameter $a_{zipf}$ (see equation~\ref{eq:zipf}) on the results. The number of topics is fixed to $K=5$ and the dictionary has size $p=10,000.$ Consistently with the experiments presented in the main text, each topic has 5 anchor words with frequency $\delta = 0.001$.}
    \label{fig:effect_azipf}
\end{figure}

\subsection{Semi-synthetic experiments}

\xhdr{Data Generation Mechanism} To ensure the realism of our experiments, we now propose to evaluate the performance of our method on a semi-synthetic dataset. To this end, we consider the Associated Press Dataset, a dataset of 2,246 documents with 10,473 terms. We fit an LDA model on the induced Document-term matrix and vary the number of topics to obtain an underlying $A_0$ and $W_0$.
For each experiment, we select a subset of documents by drawing without resampling $i \in \{1,\cdots, n\}$, and concatenate the selected columns of $W_i$ to construct the topic-to-document matrix $W$. Having obtained $A$ and $W$, we then sample each column of the matrix $D$ from a multinomial distribution as before. We note that: (a) this data generation mechanism matches the one assumed by LDA. It thus comes at no surprise that LDA performs better in this set of experiments, compared to in the previous subsection; and (b) this data generation mechanism does not impose as much sparsity as the Zipf law. Indeed, the $A$s imputed by the LDA and which serve as ground truth here, and is less heterogeneous than in our other synthetic experiments. Nonetheless, our method improves slightly over Topic SCORE for corpora of moderate size, across most data regimes. We also note that both SCORE methods outperform AWR across all data regimes. Interestingly, LDA exhibits impressive accuracy in this particular regime.

\begin{figure}[h!]
    \centering
    \includegraphics[width=1.1\textwidth]{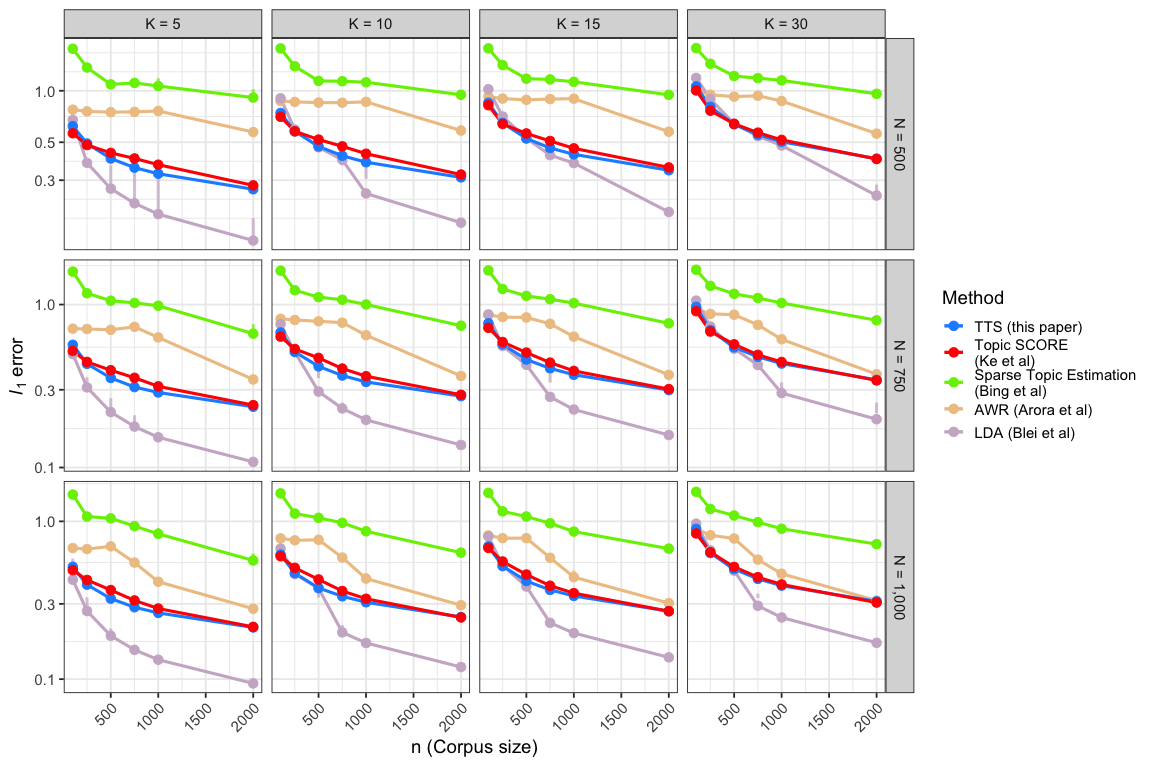}
    \caption{Semi-synthetic experiment: Comparison of the $\ell_1$ error in the reconstruction of $A$ averaged over 50 independent experiments. The estimation error is plotted as a function of corpus size $n$ (on the x-axis), as a function of document length $N$ (one per column) and number of topics $K$ (one per row). }
    \label{fig:semi-synthetic}
\end{figure}

\newpage

\xhdr{Computational Speed}

Finally, we report here the computational speed measured in our experiments. All experiments here were performed in the University of Chicago Cluster.  We observe a substantial increase in the time required by Topic SCORE as the number of documents $n$ increases.

\begin{figure}[h!]
    \centering
    \includegraphics[width=\textwidth]{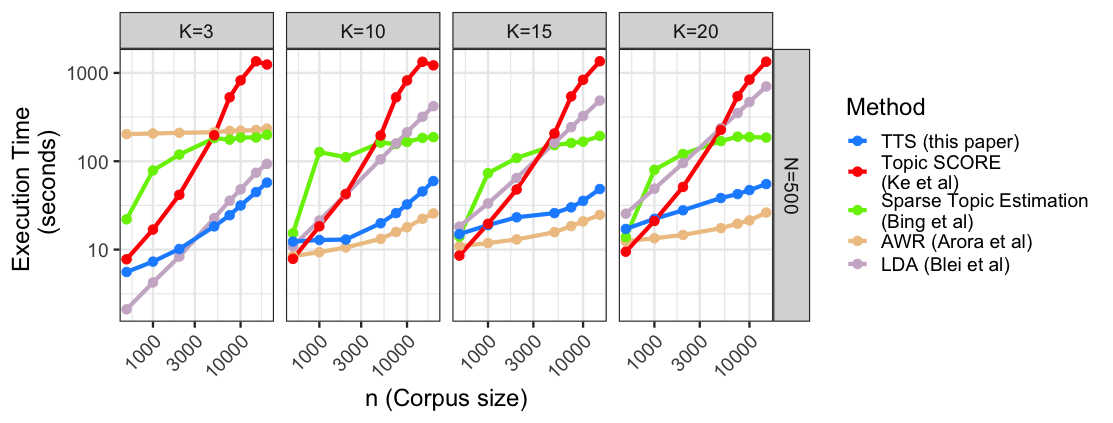}
    \caption{Median computational time for each of the methods as a function of the dictionary length $p$ over 50 independent synthetic experiments.}
    \label{fig:time}
\end{figure}

\newpage

\section{Real-World experiments: additional results}\label{app: real data experiment}

In this section, we provide additional visualization plots in microbiome data analysis to evaluate the estimated topic representatives by LDA, Topic SCORE, and our methods. We focus on two microbiome datasets, one is the colon dataset of \cite{yachida2019metagenomic} as before with data regime of low $p$, low $n$, high $N$, another comes from the vaginal microbiome data of \cite{callahan2017replication} with data regime of mid $p$, mid $n$, and high $N$.
\subsection{Microbiome dataset from \texorpdfstring{\cite{yachida2019metagenomic}}{Yachida et al. (2019)} (low \texorpdfstring{$p$}{p}, low \texorpdfstring{$n$}{n}, high \texorpdfstring{$N$}{N})}
As depicted in Figure \ref{fig:resolution-microbiome-cosine}, our approach consistently demonstrates a higher average topic resolution compared to Topic SCORE. 
To understand the reasons for the gap in performance between the Topic SCORE method \citep{ke2022using} and ours, we compare their respective point clouds in  Figure~\ref{fig:point-cloud-mic}. The left plot in figure \ref{fig:sub2} shows the adverse impact of pre-SVD normalization on the Topic SCORE's point cloud. It is evident that the outliers from small $h_j$ skew the point cloud, and thereby compromising quality of its subsequent vertex hunting. In comparison, our point cloud in Figure \ref{fig:sub1} shows an equilateral triangle formed by the three estimated vertices. Our method combined the thresholding phase with the spectral decomposition guarantees a higher signal-to-noise ratio.  
\begin{figure}[htbp]
    \centering
    \begin{subfigure}[b]{0.45\textwidth}
        \includegraphics[width=0.8\textwidth]{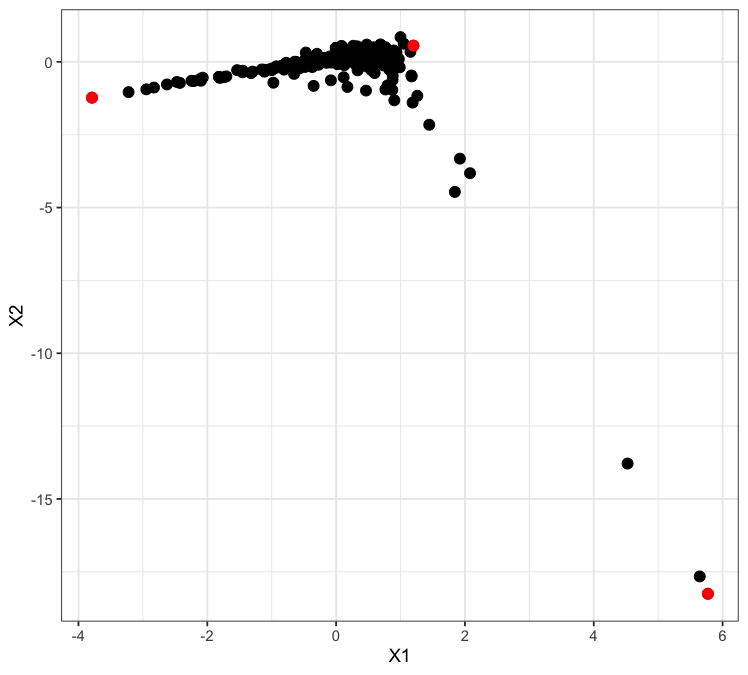}
        \caption{Point Cloud for $K=3$ in Topic Score}
        \label{fig:sub1}
    \end{subfigure}
    \hfill
    \begin{subfigure}[b]{0.45\textwidth}
        \includegraphics[width=0.8\textwidth]{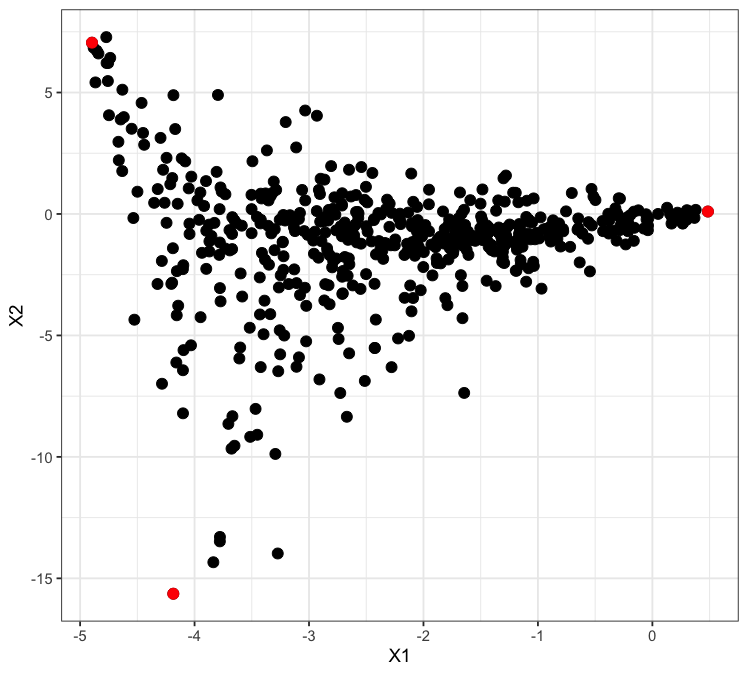}
        \caption{Point Cloud for our method}
        \label{fig:sub2}
    \end{subfigure}
    \caption{Comparison of the point clouds $R$ obtained by our method (right) and Topic Score (left). Vertices are colored red. Note in particular how the point cloud on the left is stretched by a few outliers.}
    \label{fig:point-cloud-mic}
\end{figure}
\vspace{1cm}

\newpage    

\subsection{Microbiome dataset from \texorpdfstring{\cite{callahan2017replication}}{Callahan et al. (2017)}(high \texorpdfstring{$N$}{N}, mid \texorpdfstring{$p$}{p}, mid \texorpdfstring{$n$}{n})}
We also reanalyse the dataset of \cite{callahan2017replication}, which serves as an example in \cite{fukuyama2021multiscale} to exemplify their topic refinement procedure. This dataset comprises ASV counts for 2,699 different bacterial species from 2,179 longitudinal samples collected throughout pregnancy in 135 individuals \citep{callahan2017replication}. In this case, the average sample length is around $N=157,500$. In \cite{fukuyama2021multiscale}, based on the refinement results of the LDA, the authors conclude that the topic analysis should be done using $K=7$ topics, or with up to $K=12$ if one allows for the possibility of spurious topics. We thus fit up to 12 topics and plot the average resolution (Figure~\ref{fig:resolution-vm-cosine}) and refinement of the methods in Figure~\ref{fig:callahan-res}. We find that our method compares favorably to Topic SCORE in terms of its average resolution and similarly to LDA as a whole. For a low number of topics ($K\leq7$), our method seems even preferable to LDA in terms of topic resolution, achieving better resolution at much greater speed. The topics found by The topic coherence seem a little higher for LDA at $K=7$ (the recommended choice of $K$ by \cite{fukuyama2021multiscale}) than for the others.  For $K \geq 8$, our method seems to yield topics that are more recombined from one level of the hierarchy to the next than what is observed in LDA. This concurs with the choice of $K=7$ by the authors, but seems to highlight the fact that for high values of $K$, in datasets of moderate sizes, LDA appears a preferable choice. However, we note that Topic SCORE seems to  (1) exhibit more recombination of topics early on in the hierarchy (see bottom topics in Figure~\ref{fig:sub1vaginal}); and (2) put very little mass on topics as $K$ increases (indicated by the small size of the rectangles in Figure~\ref{fig:sub1vaginal}): as $K$  increases, most of the dataset's mass is distributed along roughly 5 or  topics. This could in particular mean that the method does not really identify more than 6 topics. 

\begin{figure}[h!]
    \centering
    \includegraphics[width=0.8\textwidth]{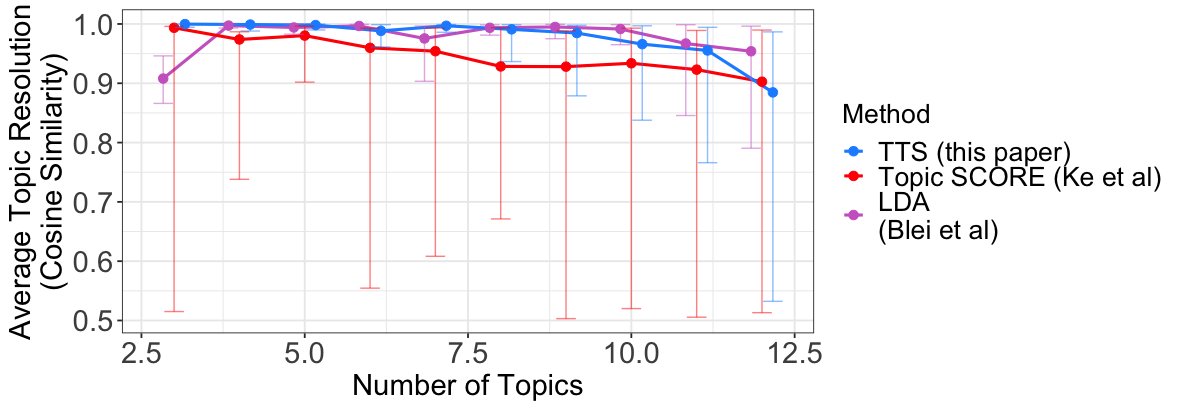}
    \caption{Topic Resolution (measured by the average cosine similarity between halves of the data) of our method (in blue), Topic Score (red), and LDA (purple).}
    \label{fig:resolution-vm-cosine}
\end{figure}

\begin{figure}
    \centering
    \begin{subfigure}[b]{0.62\textwidth}
        \includegraphics[width=\textwidth]{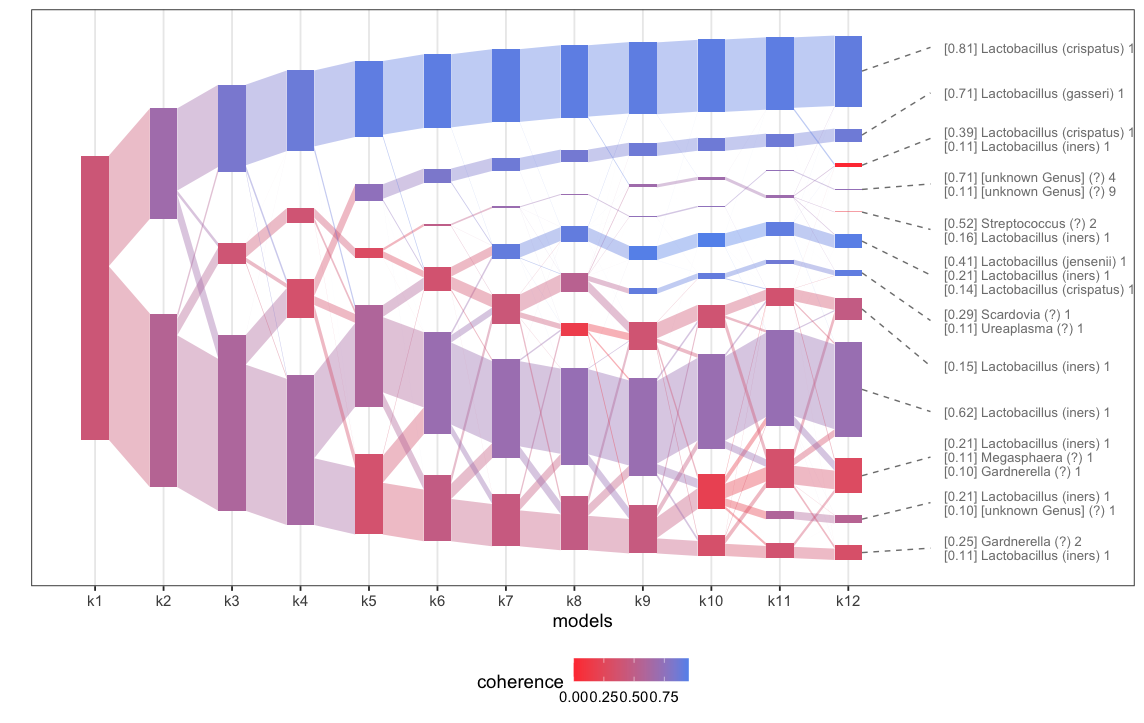}
        \caption{Coherence between topics for Topic SCORE.}
        \label{fig:sub1vaginal}
    \end{subfigure}
    \hfill
    \begin{subfigure}[b]{0.62\textwidth}
        \includegraphics[width=\textwidth]{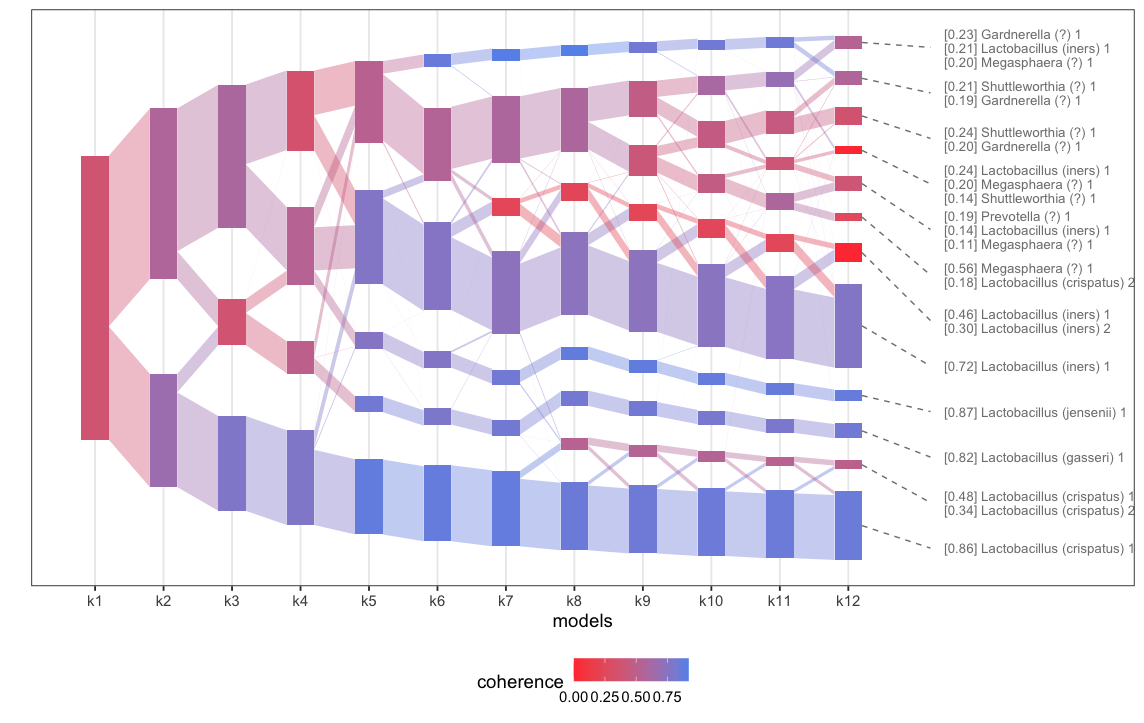}
 \caption{Coherence between topics for TTS (our method).}
        \label{fig:sub2vaginal}
    \end{subfigure}
    \hfill
    \begin{subfigure}[b]{0.62\textwidth}
        \includegraphics[width=\textwidth]{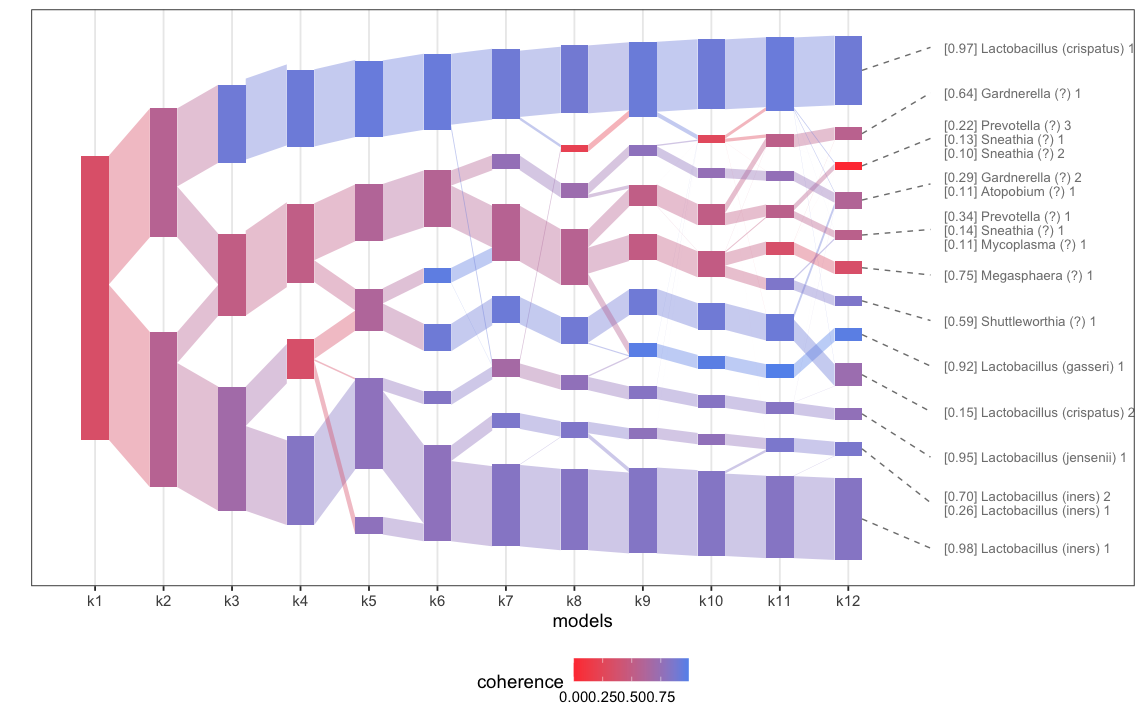}
 \caption{Coherence between topics for LDA.}
        \label{fig:sub3vaginal}
    \end{subfigure}
    \caption{Topic Coherence and refinement (as computed by the method of \cite{fukuyama2021multiscale}) for the different methods in the Vaginal Microbiome of \cite{callahan2017replication}. Topics are here colored by coherence.}
    \label{fig:callahan-res}
\end{figure}
\newpage

\end{appendix}

\end{document}